\documentclass[12pt]{article}

\usepackage[onehalfspacing]{setspace}
\usepackage{amsthm}
\usepackage{hyperref}
\usepackage{graphics,graphicx,amssymb,mathtools}
\usepackage{caption}
\usepackage{booktabs}
\usepackage{dsfont}
\usepackage{xcolor}

\usepackage[sort&compress,numbers]{natbib}

\usepackage{footmisc}

\newtheorem{remark}{Remark}
\newtheorem{example}{Example}
\newtheorem{cexample}{Counterexample}
\newtheorem{prop}{Proposition}
\newtheorem{thm}{Theorem}
\newtheorem{assump}{Assumption}
\newtheorem{corollary}{Corollary}
\newtheorem{definition}{Definition}
\newtheorem{lemma}{Lemma}
\newtheorem{fact}{Fact}

\usepackage{xr}
\usepackage{dsfont}
\usepackage{amssymb}

\newcommand{\uu}{{\mathbf u}}
\newcommand{\vv}{{\mathbf v}}

\newcommand{\uup}{\mathbf{u}^{\prime}}

\newcommand{\RR}{\mathbb{R}}
\newcommand{\calS}{\mathcal{S}}
\newcommand{\So}{\mathcal{S}_{o}}
\newcommand{\Su}{\mathcal{S}_{u}}
\newcommand{\Io}{I_t}
\newcommand{\Ho}{J_t}
\newcommand{\Hts}{J_{t_{0}+s}}
\newcommand{\rtd}{\tilde{r}_{t}}
\newcommand{\Kt}{K_t}
\newcommand{\rd}{\tilde{r}}
\newcommand{\ud}{\tilde{u}}

\newcommand{\bpi}{\bar{\pi}}

\newcommand{\leqc}{\preccurlyeq}

\newcommand{\prob}{{\mathbb P}}
\newcommand{\KK}{\mathbb{K}}
\newcommand{\KKt}{\mathbb{K}_t}
\newcommand{\KKts}{\mathbb{K}_{t_{0}+s}}
\newcommand{\KKK}{\tilde{\mathbb{K}}}

\newcommand{\KKtts}{\tilde{\mathbb{K}}_{t_{0}+s}}

\newcommand{\ww}{\mathbf{w}}

\newcommand{\fil}{{\mathcal F}}
\newcommand{\HH}{\mathcal{H}}
\newcommand{\Acal}{\mathcal{A}}

\newcommand{\ind}{\mathds{1}}
\let\hat\widehat
\let\tilde\widetilde

\newcommand{\EE}{\mathbb{E}}
\newcommand{\dd}{\stackrel{d}{=}}
\newcommand{\pD}{\delta_{\text{P}}}
\newcommand{\sD}{\delta_{\text{S}}}

\newcommand{\done}{\delta_{1}}
\newcommand{\dtwo}{\delta_{2}}

\newcommand{\PiInf}{\pi_{\infty}}

\DeclareMathOperator*{\argmax}{arg\,max}

\usepackage{dsfont}

\let\hat\widehat
\let\tilde\widetilde

\usepackage{algorithm}
\usepackage{algorithmic}
\usepackage{array}
\newcolumntype{L}[1]{>{\raggedright\let\newline\\\arraybackslash\hspace{0pt}}p{#1}}
\newcolumntype{C}[1]{>{\centering\let\newline\\\arraybackslash\hspace{0pt}}p{#1}}
\newcolumntype{R}[1]{>{\raggedleft\let\newline\\\arraybackslash\hspace{0pt}}p{#1}}

\title{Optimal Parallel Sequential Change Detection under Generalized Performance Measures}

\author{
  Zexian Lu\\
  University of Minnesota\vspace{0.5cm}\\  
Yunxiao Chen\\
London School of Economics and Political Science\vspace{0.5cm}\\
  Xiaoou Li\\
  University of Minnesota}
\date{}

\begin{document}
\maketitle

\doublespacing

\begin{abstract}
This paper considers the detection of change points in parallel data streams, a problem widely encountered when analyzing large-scale real-time streaming data. Each stream may have its own change point, at which its data has a distributional change. With sequentially observed data, a decision maker needs to declare whether changes have already occurred to the streams at each time point.
Once a stream is declared to have changed, it is deactivated permanently so that its future data will no longer be collected.
This is a compound decision problem in the sense that the decision maker may
want to optimize certain compound performance metrics
that concern all the streams as a whole. Thus, the decisions are not independent for different streams. Our
contribution is three-fold.  First, we propose a general framework for compound performance metrics that includes the ones considered in the existing works as special cases and introduces new ones that connect closely with the performance metrics for single-stream sequential change detection and large-scale hypothesis testing. 
Second, data-driven decision procedures are developed under this framework. Finally, optimality results are established for the proposed decision procedures. 
The proposed methods and theory are evaluated by simulation studies and a case study. \end{abstract}


\section{Introduction}

Sequential change detection aims to detect distributional changes in sequentially observed data. Classical methods focusing on change detection in a single data stream   have received wide applications in 
various fields, including engineering, education, medical diagnostics and finance \cite{shewhart1931economic,page1954continuous, shiryaev1963optimum, roberts1966comparison}. Several metrics have been proposed for evaluating their performance, under which optimality theory has been established  \cite{lorden1971procedures, pollak1985optimal, pollak1987average, moustakides1986optimal}; see \cite{lai2001sequential, poor2008quickest, basseville1993detection} for a review. 

The emergence of large-scale real-time
streaming data has motivated multi-stream sequential change detection problems. One problem concerns detecting a common change shared by a subset of the streams \cite{chan2017optimal,chen2015graph,chen2019sequential,mei2010efficient,xie2013sequential,fellouris2016second}. This problem is commonly seen in surveillance applications, where each data stream corresponds to a sensor, and the change point is caused by a failure in a subset of the sensors. A related problem, which has received much attention recently  and will be the focus of the current work, considers a setting that each stream has its own change point \cite{chen2020false,chen2019compound,chen2020item, chen2016non}.  More specifically, a decision maker needs to declare whether a change has already occurred for each stream at each time point. 
Once a stream is declared to have changed, it is deactivated permanently so that its data is no longer collected. This problem will be referred to as a parallel sequential change detection problem.

The parallel sequential change detection problem is widely encountered in the real world. For example, \cite{chen2020false,lai2008quickest} consider an application to a multichannel dynamic spectrum access problem for cognitive radios. Each cognitive radio channel corresponds to a data stream, and the change corresponds to the time at which the primary user of the channel starts to transmit
signals. A false discovery rate (FDR)   is proposed to measure the proportion of false discoveries (i.e., unused channels) among the ones detected as occupied by primary users. \cite{chen2019compound,chen2020item} consider monitoring an item pool for standardized educational testing. In this application, each stream corresponds to a test item that is reused in multiple test administrations, and the change point corresponds to the time at which the item is leaked to the public. A certain false non-discovery rate (FNR) is proposed to measure the proportion of leaked items among the non-detections (i.e., items that are not detected as having leaked). There are many other potential applications, such as the detection of credit card fraud \cite{dal2017credit}, for which each stream corresponds to a credit card account, and the change point corresponds to a fraud event. 
 
We note that it is often not a good idea to run a single-stream change detection procedure independently on individual streams. This is because the decision maker may want to control a certain compound risk that concerns all the streams as a whole, such as the FDR and FNR measures. Consequently, each decision at one time point requires all the up-to-date information from all the streams, making the parallel sequential change detection a challenge. 
 
Several methods have been proposed in \cite{chen2020false,chen2019compound,chen2020item} to control the above compound risk measures in parallel sequential change detection problems. However, these methods, along with their theoretical properties, are established under relatively restrictive model assumptions and for specific risk measures.
Specifically, 
\cite{chen2020false}  proposes a method based on the Benjamini-Hochberg method \cite{benjamini1995controlling} for FDR control and establishes its asymptotic results. However, no results are given on the method's optimality. Under a Bayesian setting, \cite{chen2019compound} and \cite{chen2020item} propose methods for controlling a certain FNR measure at all time points. As shown in \cite{chen2019compound}, under a geometric change point model and assuming the same pre- and post- change distribution for all the streams, this method maximizes the expected number of remaining streams at all time points while controlling the FNR to be no greater than a pre-specified tolerance level. However, it is unclear whether this optimality theory can be extended to more general models and other sensible risk measures.

The parallel sequential change detection problem is also closely related to the sequential multiple testing problem. The latter can be viewed as a special case when a stream can only change
at the beginning of the process or never change.  Several methods have been proposed for the sequential multiple testing problem, controlling  compound risks. Specifically, \cite{bartroff2014sequential}, 
\cite{bartroff2020sequential}, and 
\cite{song2019sequential} consider controlling a familywise error rate, an FDR/FNR, and a generalized familywise error rate, respectively. While the risk measures may be relevant, their methods and theoretical results can hardly be extended to the current change detection problem.  

This work provides a unified decision theory framework for parallel sequential change detection problems under general classes of change point models and performance measures. A computationally efficient sequential method is developed under the proposed framework. Two optimality criteria are introduced, for which the proposed method is shown to be optimal under suitable conditions.

Our contributions are summarized below:
\begin{itemize}
    \item 
We propose a general class of performance metrics to evaluate the sequence procedures. This class of metrics not only includes existing metrics 
as special cases (e.g.,   FDR  \cite{chen2020false} and the local FNR metric  \cite{chen2019compound}) but also introduces new metrics that are closely related to the  metrics for single-stream change detection and multiple hypothesis testing. See Section~\ref{sec:example} and Section~\ref{sec:aggregated-risk-theory} for more examples.

Thanks to the generality of these performance metrics, the proposed method can also be used to solve problems considered in  \cite{bartroff2020sequential,song2019sequential,song2017asymptotically} for sequential multiple testing. See Section~\ref{sec:aggregated-risk-theory} for a discussion on the connections with several recent works \cite{song2017asymptotically,song2019sequential,bartroff2018multiple,bartroff2014sequential,bartroff2020sequential}.
\item  We propose a sequential procedure (Algorithms~\ref{alg:gen}--\ref{alg:dec2}) that is easy-to-implement and is data-driven. 
It automatically adapts to various model settings when controlling the risk measures to a pre-specified tolerance level, without requiring additional Monte Carlo simulation or bisection search commonly used in sequential problems to determine decision boundaries (see, e.g., \cite{bartroff2008modern}).

\item We provide two optimality criteria for the parallel sequential change detection problem, including the local and uniform optimalities. The local optimality concerns the maximization of a utility measure in the next step, and uniform optimality refers to the maximization of the utility measure at all time.
We show that the proposed method is locally optimal under very mild conditions and uniformly optimal under stronger conditions (Theorems~\ref{thm:one-step optimality for general risk}--\ref{thm:uniform optimality}).

We note that the precise characterization of the conditions for uniform optimality requires the analysis of stochastic processes on a special non-Euclidean space. To this end, we develop new analytical tools for comparing vectors and stochastic processes with different dimensions, possibly due to early stopping. This analytical tool may be useful in the theoretical analysis of other sequential decision problems.
\end{itemize}

The remainder of the paper is organized as follows. In Section~\ref{sec:problem-setup}, we describe the change point models, the class of parallel sequential change detection methods, a general class of performance metrics, and the optimality criteria. We also provide examples of generalized performance metrics.
In Section~\ref{sec:method}, we propose a parallel change detection method (Algorithms~\ref{alg:gen} and \ref{alg:dec}) and provide a simplified version of this method under mild conditions on the performance measures (Algorithms~\ref{alg:simp} and \ref{alg:dec2}).  Section~\ref{sec:theory} provides theoretical results for the proposed methods including their optimality properties and the connection with recent works.  In Sections~\ref{sec:sim}  and~\ref{sec:case-study},
 we evaluate the performance of the proposed method through simulation studies and a case study. Concluding remarks and future directions are given in
 Section~\ref{sec:conc}. For space reasons, all the proofs of the theoretical results and part of the simulation results are postponed to the Appendix in the supplementary material.

\section{Problem Setup}\label{sec:problem-setup}

\subsection{Model Assumptions}
Consider the case where there are $K \geq 2$ data streams, and let $\langle K \rangle$ denote the set $\{1,\cdots, K\}$. At each time epoch $t\in \mathbb{Z}_+=\{1,2,\cdots\}$, an observation $X_{k,t}$ is obtained from the $k$th data stream, for $k\in \langle K \rangle$. Each data stream $k$ is associated with a change point $\tau_k\in \{0\}\cup\{\infty\}\cup\mathbb{Z}_+$ for $k\in\langle K \rangle$. %
Under a Bayesian parallel change point model, the change points $\tau_1,\cdots,\tau_K$ are assumed to be independent and identically distributed (i.i.d.) with
\begin{equation}\label{eq:prior}
\prob(\tau_{k}=s)=\pi_s
\end{equation}
for $s\in \{0\}\cup\{\infty\}\cup\mathbb{Z}_+$ and $k\in\langle K \rangle$.
Given $(\tau_1,\cdots,\tau_K)$, $\{X_{k,t}\}_{t\in\mathbb{Z}_+}$ are independent for $k\in\langle K \rangle$, and have conditional density 
\begin{equation}\label{eq:model}
   X_{k,t}|\tau_k, \{X_{k,s}\}_{1\leq s\leq t-1} \sim 
   \begin{cases}
       &p_{k,t} \text{ if } t\leq \tau_k \\
       & q_{k,t} \text{ if } t\geq \tau_k + 1
   \end{cases}
\end{equation}
with respect to some baseline measure.  That is, 
$X_{k,t}$ are independent given the change points, and follow  pre- and post- change density functions $p_{k,t}$ and $q_{k,t}$, respectively. In particular, $\tau_k=\infty$ corresponds to the case where the change point never occurs to the $k$th stream. That is,  $X_{k,t}$ follows the pre-change density function $p_{k,t}$ for all $t\in \mathbb{Z}_+$.

\subsection{Parallel Sequential Change Detection Procedures}\label{sec:class-of-decisions}
A decision maker sequentially observes data from the parallel data streams and determines whether change points have already occurred to these data streams at each time. Once a change point is declared, the corresponding data stream is deactivated and its data are no longer collected. This decision process is characterized by an index set process $S_t\subset \langle K \rangle$ for $t\in\mathbb{Z}_+$, where $k\in S_t$ if and only if the decision maker has not declared a change in the $k$th stream at time $t$ yet (i.e., stream $k$ is active at time $t$). 
Specifically, the
available information at time $t$ is contained in the historical data  $H_t=\{\{X_{k,s}\}_{k\in S_s, 1\leq s\leq t}, \{S_s\}_{1\leq s\leq t}\}$ and, equivalently, the induced information $\sigma$-field $\mathcal{F}_t=\sigma(H_t)$. At each time $t$, the decision maker selects the index set $S_{t+1}\subset S_t$ based on the current information $\mathcal{F}_t$. That is, $S_{t+1}$ is measurable with respect to $\mathcal{F}_t$. Denote by $\mathcal{D}$ the set of all such compound sequential decisions. A graphical illustration of the decision process is given in Figure~\ref{fig:flowchart}.
\begin{figure}[]
\begin{center}
\includegraphics[scale = 0.27]
{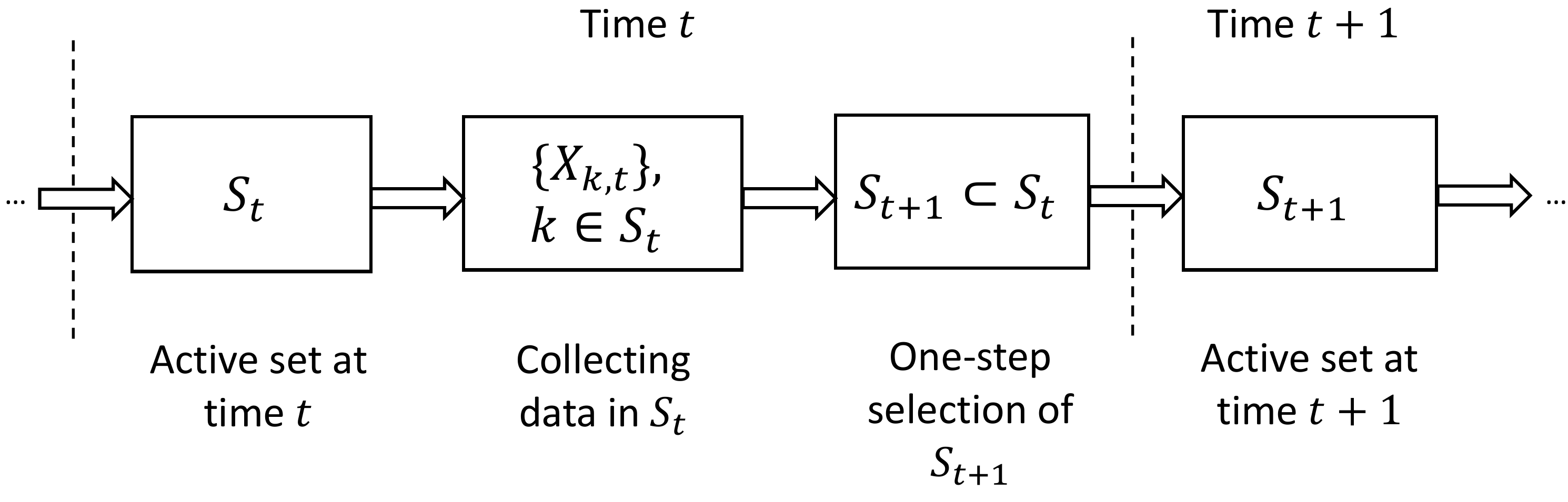}
\end{center}
\caption{A flowchart of a sequential decision in $\mathcal{D}$}
\label{fig:flowchart}
\end{figure}

We make a few remarks on the information filtration and the decision process. First, we require $S_1=\langle K \rangle$, meaning that all the streams are initially active and data from all the streams are collected at time $1$. Second, $\{S_s\}_{1\leq s\leq t}$ is measurable with respect to $\mathcal{F}_t$, meaning that the decision history is tracked in the current information. Third, $\{X_{k,s}\}_{k\in S_s, 1\leq s\leq t}$ is measurable with respect to $\mathcal{F}_t$, indicating that $X_{k,s}$ is observed if and only if stream $k$ is active at time $s$ and $s\leq t$ (i.e., $k\in S_s$). Fourth, $S_{t+1}$ is required to be measurable with respect to $\mathcal{F}_t$, meaning that the decision maker selects the active streams for time $t+1$ based on all the information available at time $t$. Lastly, $S_{t+1}$ is required to be a subset of $S_t$ for all $t\in\mathbb{Z}_+$, meaning that the deactivation of streams is permanent. That is, no future data will be collected at a stream, once a change is declared at that stream.

\begin{remark} \label{remark:connection-with-Poor}
Although described in a different way, the class of sequential decisions defined above is equivalent to that in \cite{chen2020false}. In \cite{chen2020false}, a parallel sequential procedure is defined through a sequence of stopping times $\{T_q\}_{q\geq 1}$  along with a sequence of index sets $\{D_q\}_{q \geq 1}$. At each stopping time $T_q$, a decision maker declares change points for streams in $D_q$ and exclude those streams from the future decision process. Then, the sequences  $\{T_q\}_{q\geq 1}$ and $\{D_q\}_{q \geq 1}$ can be represented using the sequence $\{S_t\}_{t\geq 1}$ as $T_q = \min\{t > T_{q-1} : S_{t} \setminus S_{t+1} \neq \varnothing \}$ and $D_q = S_{T_q} \setminus S_{T_{q}+1}$ where  $T_0 = 0$, $q = 1$. An example where $K=3$ is given in Figure~\ref{fig:plot_ST} for a graphical illustration.

\end{remark}
\begin{figure}[t]
\begin{center}
\includegraphics[scale = 0.29]
{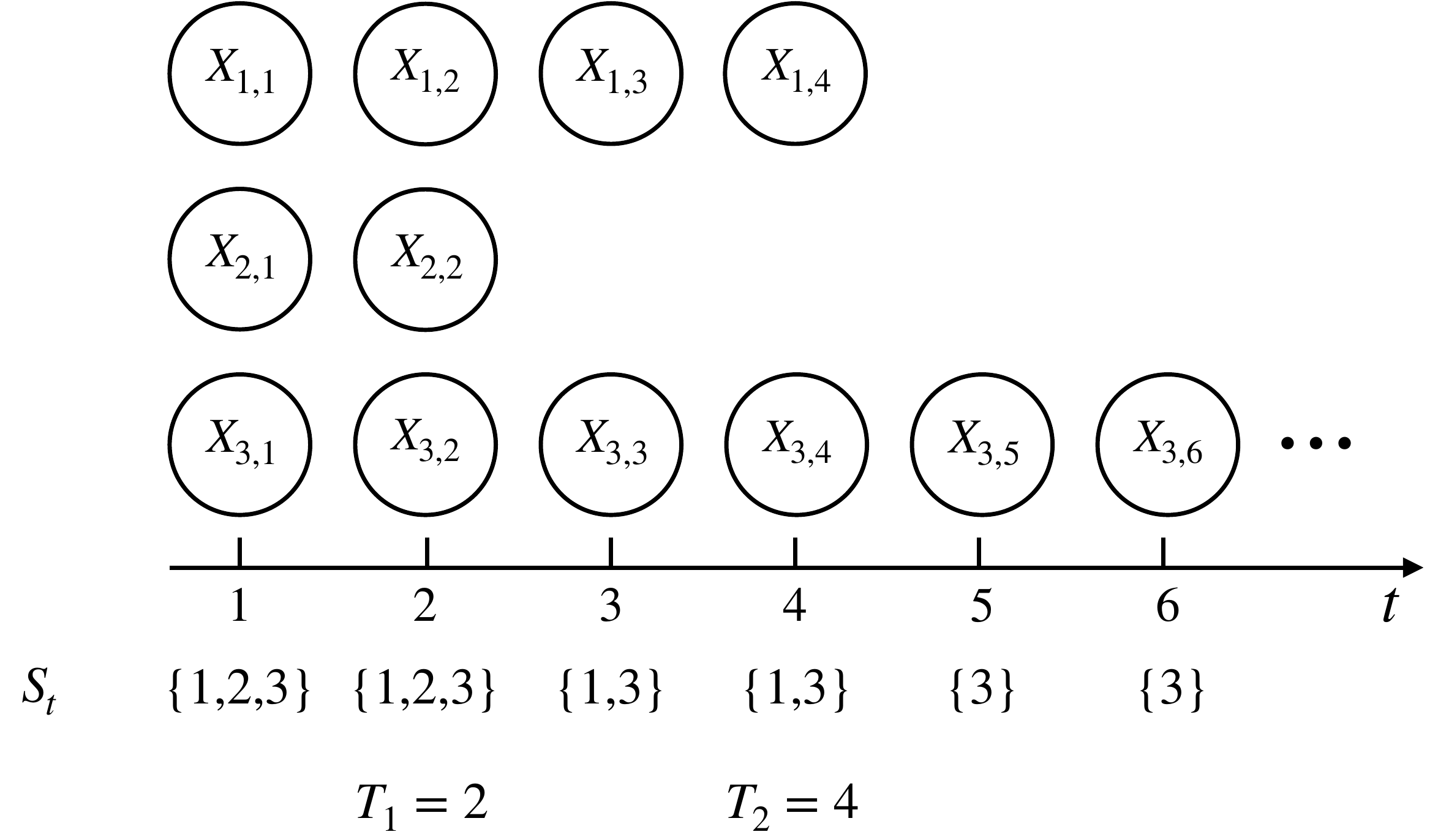}
\end{center}
\caption{An example of a parallel sequential change detection procedure where $K=3$, stream 2 is deactivated at time $t=2$,  stream 3 is deactivated at time $t=3$, and no more stream is deactivated before $t=6$. As a result, $S_1=S_2=\{1,2,3\}$, $S_3=S_4=\{1,2\}$, $S_5=S_6=\{1\}$. Correspondingly, $T_1 = 2$, $T_2 = 4$, and $T_3>6$. }
\label{fig:plot_ST}
\end{figure}

Another way to understand a compound sequential change detection procedure is to view it as a sequence of mappings $\delta=(d_1,d_2,\cdots,d_t,\cdots)$, where each $d_t$ determines $S_{t+1}$ according to the historical information $H_t$. That is, $d_t$ is a measurable function with respect to $\mathcal{F}_t$ and $S_{t+1}=d_t(H_t)$ satisfying that $d_t(H_t)\subset S_t$ for all $t\in\mathbb{Z}_+$.  %

\subsection{Generalized Performance Measures and Optimality Criteria}
Ideally, a perfect sequential change detection procedure collects all the pre-change streams in the set $S_t$ at each time point (i.e., $S_t=\{k:\tau_k\leq t\}$). However, this is not achievable by any sequential decision because $\tau_k$s are unobserved. To this end, we consider a general class of performance measures to  compare the performance of different sequential decisions.
We assume each sequential decision is associated with a risk process, denoted by $\{R_t\}_{t\in\mathbb{Z}_+}$, and a utility process, denoted by $\{U_t\}_{t\in\mathbb{Z}_+}$. The risk process is used to quantify the loss of a sequential decision at time $t$ due to the false detections of pre-change streams and/or the non-detection of post-change streams, while the utility process is used to reward the correct decisions. Our goal is to find a good sequential decision that has a relatively small $R_t$ and a relatively large $U_t$ at every time point. 
Below, we first give formal statements of the optimality criteria, and then introduce several examples of $R_t$ and $U_t$ in Section~\ref{sec:example}, followed by additional discussions.

Let
\begin{equation}
W_{k, t}=\prob(\tau_{k}<t \mid \mathcal{F}_{t})
\end{equation}
be the posterior probability that the change point $\tau_k$ has already occurred at time $t$ for the $k$-th stream given the information up to time $t$. Under the Bayesian setting, $W_{k,t}$ is also the best estimator (under the squared error loss) of $\ind(\tau_k<t)$, where $\ind(\cdot)$ denotes the indicator function.  A simple iterative updating rule is derived to calculate  $W_{k,t}$ at  each time, which will be discussed in Section~\ref{sec:method}.

Throughout the paper, we consider risk and utility processes that are functions of $(\{W_{k,t}\}_{k\in S_t}, S_t, S_{t+1})$. That is, there are pre-specified functions $\{r_t\}_{t\in\mathbb{Z}_+}$ and $\{u_t\}_{t\in\mathbb{Z}_+}$ such that
\begin{equation}\label{eq:risk}
  R_{t} = r_{t}(\{W_{k,t}\}_{k\in S_{t} }, S_{t}, S_{t+1}), 
\end{equation}
and 
\begin{equation}\label{eq:utility}
  U_{t} = u_{t}(\{W_{k,t}\}_{k\in S_{t}}, S_{t}, S_{t+1}).
\end{equation}

Let $\alpha\in\mathbb{R}$ denote a pre-specified tolerance level, and let
$$
\mathcal{D}_{\alpha} =\left\{ \delta \in \mathcal{D}: R_{t}(\delta) \leq \alpha \text { a.s., for all } t=1,2, \cdots\right\},
$$
where $R_t(\delta)$ denotes the risk process associated with the sequential decision $\delta$, and  $\mathcal{D}$ denotes the  entire set of parallel sequential detection procedures described in Section~\ref{sec:class-of-decisions}. The set $\mathcal{D}_{\alpha}$ collects all sequential decisions that control the risk process to be no greater than the tolerance level $\alpha$ at all time points.

We note that risk  process $\{R_t\}_{t\in\mathbb{Z}_+}$ is an adaptive stochastic process with respect to the information filtration $\{\fil_t\}_{t\in\mathbb{Z}_+}$. 
It is easy to verify that $\mathbb{E}[R_{t}(\delta)]\leq\alpha$ for  $\delta\in\mathcal{D}_{\alpha}$. That is, the expected risk is also controlled below or equal to the same tolerance level. In addition, any weighted average of $R_t(\delta)$ across different time points are also controlled. We provide additional discussion and theoretical results regarding this point in Section~\ref{sec:aggregated-risk-theory}.

The following regularity assumptions over the risk and utility functions are imposed throughout the paper. %
\begin{assump}\label{assump:non-empty}
For any $\{W_{k,t}\}_{k\in S_{t} }$  and $S_t$, $\min_{S\in\{\emptyset,S_t\}}r_{t}(\{W_{k,t}\}_{k\in S_{t} }, S_{t}, S)\leq\alpha$.
In addition, the utility function $u_t$ is bounded at each time $t$. 
\end{assump}
The assumption on $r_t$ guarantees that the class of sequential decisions controlling the risk process at a pre-specified level is non-empty, i.e., $\mathcal{D}_{\alpha}\neq \emptyset$. The boundedness assumption on $u_t$ is a mild condition to ensure the integrability of the utility process.

Given a pre-specified tolerance level $\alpha$ and sequences of functions $\{r_t\}_{t\in\mathbb{Z}_+}$ and $\{u_t\}_{t\in\mathbb{Z}_+}$, we define two optimality criteria for  sequential decisions in $\mathcal{D}_{\alpha}$.
\begin{definition}[Uniform Optimality]\label{def:uniform-optimal}
A sequential decision $\delta^{*} \in \mathcal{D}_{\alpha}$ is called uniformly optimal if
$$
\mathbb{E}\left(U_{t}\left(\delta^{*}\right)\right)=\sup_{\delta \in \mathcal{D}_{\alpha}} \mathbb{E}\left(U_{t}(\delta)\right),
$$
for all $t\in\mathbb{Z}_+$,
where $U_t(\delta^*)$ and $U_t(\delta)$ denote the utility process associated with sequential decisions $\delta^*$ and $\delta$, respectively.
\end{definition}

\begin{definition}[Local Optimality]\label{def:local-optimal}
A sequential decision $\delta^{*}=(d^*_1,d^*_2,\cdots,d^*_t,\cdots) \in \mathcal{D}_{\alpha}$ is called locally optimal at  time $t$,  if
$$
\mathbb{E}(U_{t}(\delta^{*}))\geq 
\EE(U_{t}(\delta) )
$$
for any $\delta=(d_1,d_2,\cdots,d_t,\cdots)\in\mathcal{D}_{\alpha}$ satisfying $d_s= d^*_s$, for $s = 1, \ldots, t-1$.
\end{definition}

We make a few remarks on the above optimality criteria. First, in most applications, there is a trade-off between minimizing the risk and maximizing the utility. That is, a sequential decision that has relatively small risk tends to have relatively small utility at the same time. Thus, we define both uniform and local optimality through constrained optimization problems, where the  overall goal is to find a sequential decision so that its corresponding risk process is controlled to be no greater than the tolerance level while the expected utility is
no less than 
any other sequential decisions that control the risk process at the same level. 
Second, a uniformly optimal sequential decision has the largest expected utility among all decisions in $\mathcal{D}_{\alpha}$ at {\em every} time point. In contrast, a locally optimal sequential decision only has the largest expected utility at {\em a given time point} $t$ given the decisions at previous time points. Thus, uniform optimality is a stronger notion than local optimality. A sequential decision that is locally optimal at every time point does not necessarily imply that it is also uniformly optimal. In later sections, we show that locally optimal sequential decisions exist under very weak assumptions on the risk and utility measures, while uniformly optimal sequential decisions only exist under stronger assumptions of the change point model and the performance measures. Third, we assume the same tolerance level $\alpha$ for every time $t$ for ease of presentation. Our methods and theory can be easily extended to the class of sequential decisions whose risk is controlled at different levels at different time points. That is, $\{\delta\in\mathcal{D}: R_t(\delta)\leq \alpha_t\text{ for all }t\}$ for a sequence of constants $\alpha_t$. We can see this by redefining the risk process as 
$R_t-\alpha_t$ and replacing $\alpha_t$ by $0$.

\subsection{Examples of Generalized Performance Measures}\label{sec:example}

We start with several examples of performance measures in the forms of \eqref{eq:risk} and \eqref{eq:utility}, which are motivated by common risk measures in the literature of multiple hypotheses testing \cite{efron2004large,efron2001empirical,efron2012large,benjamini1995controlling}. All of the risk measures discussed in this section satisfy Assumption~\ref{assump:non-empty} for $\alpha\geq 0$.

For the consistency of notation, the sum over an empty set is defined to be $0$ (i.e., $\sum_{i\in \emptyset}a_i=0$), and the product over an empty set is defined to be $1$ (i.e., $\prod_{i\in\emptyset} a_i=1$).

\begin{example}[Local family-wise error rate (LFWER)]\label{eg:LFWER}
Consider the event
\begin{equation}
    E_{1,t} = \{\text{There exists  } k \in \langle K \rangle \text{ such that } \tau_k<t, k\in S_{t+1} \},
\end{equation}
which happens when at least one false non-detection occurs at time $t$. 
Because $E_{1,t}$ is not directly observed,  we consider the its posterior probability given the information up to time $t$,
\begin{equation}\label{eq:LFWER}
   \text{LFWER}_t:= \prob(E_{1,t}|\fil_t) = 1-\prod_{k\in S_{t+1}}(1-W_{k,t}).
\end{equation}

\end{example}
\begin{example}[Generalized local family-wise error rate (GLFWER)]\label{eg:GLFWER}
Given $m\geq 1$, we consider the event
\begin{equation} \label{eq:def-Emt}
    E_{m,t} = \{ |\{k \in \langle K \rangle \text{ such that } \tau_k<t, k\in S_{t+1} \}|\geq m\}.
\end{equation}
This event happens when false non-detections occur in at least $m$ data streams. Its posterior probability given information up to time $t$ is
\begin{align}\label{eq:GLFWER}
   &\text{GLFWER}_{m,t} : = \prob(E_{m,t}|\fil_t) \\
   =& 1-\sum_{j = 0}^{m-1} \sum\limits_{\substack{I \subset S_{t+1} \\ |I| = j}} \left(\prod_{i \in I} W_{i,t} \right)  \prod_{k\in S_{t+1} \setminus I}(1-W_{k,t}).
\end{align}
In addition, $\text{GLFWER}_{m,t} = 0$ if $S_{t+1}=\varnothing$.
\end{example}
Comparing \eqref{eq:LFWER} with \eqref{eq:GLFWER}, we can see that GLFWER extends LFWER by allowing for more false non-detections. 
Under a large-scale setting with many data streams, it may be more sensible to use GLFWER with its $m$ value chosen based on the total number of streams $K$ to achieve a balance between false detections and false non-detections. 
 Similar risk measures have been proposed for sequential multiple testing \cite{song2019sequential}.
\begin{example}[Local false non-discovery rate (LFNR)]\label{eg:LFNR}
    Local false non-discovery rate  (LFNR) is defined in \cite{chen2019compound}, which extends the concept of LFNR in multiple testing to parallel sequential change detection. It is defined as follows. First,  the false non-discovery proportion (FNP) is defined as
    \begin{equation}\label{eq:fnp}
    \text{FNP}_{t}
    :=\frac{\sum_{k \in S_{t+1}} \mathds{1}\left(\tau_{k}<t\right)}
    {|S_{t+1}| \vee 1}.
    \end{equation}
FNP describes the proportion of post-change streams among the active ones. Then, the local false non-discovery rate (LFNR) at time $t$ is defined as the Bayes estimator (i.e., posterior mean) of $\text{FNP}_t$ given information up to time $t$. That is,
\begin{equation} \label{eq:LFNR}
    \text{LFNR}_t := \EE (\text{FNP}_{t} \mid \fil_t) = \frac{\sum_{k \in S_{t+1}} W_{k, t}}
    {|S_{t+1}| \vee 1}.
\end{equation}
\end{example}
Compared with LFWER and GLFWER, LFNR depends on $W_{k,t}$s in a linear rather than multivariate polynomial form. In addition, LFNR is scalable under a large-scale setting  in the sense that the same tolerance level $\alpha\in (0,1)$ can be used as $K$ grows large.

\begin{example}[Local False Discovery Rate (LFDR)]\label{eg:LFDR}
False discovery proportion (FDP) and local false discovery rate (LFDR) are defined by replacing $\tau_k<t$ and $S_{t+1}$ with $\tau_k\geq t$ and $S_{t}\setminus S_{t+1}$ respectively in \eqref{eq:fnp} and \eqref{eq:LFNR}. That is,
\begin{equation} \label{eq:FDP}
    \text{FDP}_t := \frac{\sum_{k \in S_t \setminus S_{t+1}} \mathds{1}\left(\tau_{k} \geq t\right)} 
{|S_t \setminus S_{t+1}| \vee 1},
\end{equation}
and
\begin{equation}
\text{LFDR}_t := \EE (\text{FDP}_{t} \mid \fil_t) = \frac{\sum_{k \in S_t \setminus S_{t+1}}  (1- W_{k, t}) } {|S_t \setminus S_{t+1}| \vee 1}.
\end{equation}

\end{example}

Similar to LFNR, LFDR also has the appealing feature of scalability for large $K$. The difference between LFNR and LFDR lies in whether focusing on false detections or false non-detections. 

In \cite{chen2020false}, an aggregated version of false discovery rate (AFDR)\footnote{In \cite{chen2020false}, this risk measure is referred to as `false discovery rate (FDR)'. Here, we name it as AFDR to distinguish it from LFDR.} is considered, which can be viewed as the expectation of a weighted average of LFDR at different time points. More discussions on the connection between LFDR and AFDR will be provided in Section~\ref{sec:theory}.

Next, we provide two examples of performance measures motivated by single-stream sequential change detection. Denote by $N_k$ the detection time of the $k$th stream,
\begin{equation} \label{eq:Nk}
N_{k}=\sup \left\{t: k \in S_{t}\right\}.
\end{equation}
Note that $N_k$ plays a similar role as the stopping time in the standard single-stream sequential change detection problem. Indeed, $N_k$ is a stopping time with respect to $\{\fil_t\}_{t\in\mathbb{Z}_+}$ for all $k\in\langle K \rangle$.

\begin{example}[Incremental Average Run Length (IARL)]\label{eg:ARL}
We define the incremental run length (IRL) aggregated over different streams as
\begin{equation}
\begin{split}
   \text{IRL}_t: & = \sum_{k=1}^K \{\tau_k \wedge N_k\wedge (t+1)\} -  \sum_{k=1}^K \{\tau_k \wedge N_k\wedge t\} \\
   & = \sum_{k\in S_{t+1}}^K \ind(\tau_k>t)
\end{split}
\end{equation}
IRL indicates the total number of pre-change streams being used at a given time. We refer to its posterior mean as the incremental average run length (IARL), defined as
\begin{equation}
    \text{IARL}_t := \EE (\text{IRL}_t \mid \fil_t) = \sum_{k\in S_{t+1}} \{1-g(W_{k,t})\},
\end{equation}
where 
\begin{equation} \label{eq: tau leq t}
g(W_{k,t}) = \prob(\tau_k\leq t|\fil_t) = \bar{\pi}_{t}^{-1} \pi_{t} + \big(1- \bar{\pi}_{t}^{-1} \pi_{t} \big) W_{k, t},    
\end{equation}
$\bar{\pi}_{s}=\prob \big (\tau_{k} \geq s
\big)=\pi_{\infty}+\sum_{l=s}^{\infty}\pi_l$, and the proof for equation \eqref{eq: tau leq t} is given in Appendix~\ref{sec:proof-eqs}.
\end{example}

IRL and IARL are closely related to the average run length to false alarm (ARL2FA) that is commonly used to measure the propensity for making a false detection in a single-stream sequential change detection problem. Specifically, taking summation of $\text{IRL}_t$ over $t$, we obtain
\begin{equation}
\sum_{s=0}^{t-1} \text{IRL}_s = \sum_{k=1}^K \big( \tau_k\wedge N_k\wedge t \big),
\end{equation}
which is the total run length from different data streams up to the change point by time $t$. Moreover, we have
\begin{equation}
\mathbb{E}( \sum_{s=0}^{t-1} \text{IARL}_s) = \mathbb{E}(\sum_{s=0}^{t-1} \text{IRL}_s )
= \sum_{k=1}^K \mathbb{E} \big( \tau_k\wedge N_k\wedge t \big).
\end{equation}
Thus, the sum of the expected value of IARL across time leads to the total averaged run length up to the change point.

\begin{example}[Incremental Average Detection Delay (IADD)]\label{eg:detection-delay}
We define the incremental detection delay (IDD) aggregated over all the streams as
\begin{equation} \label{eq:IDD}
\begin{split}
    \text{IDD}_t: &=  \sum_{k = 1}^K  \{  (N_{k}\wedge (t+1) - \tau_{k} -1)_+ -  (N_{k}\wedge t - \tau_{k} -1)_+  \}\\
    &= \sum_{k\in S_{t+1}}\ind(\tau_k< t).
\end{split}
\end{equation}
IDD counts the total number of post-change streams that are active at a given time. We refer to its posterior mean as the incremental average detection delay (IADD), defined as
\begin{equation}\label{eq:IADD}
    \text{IADD}_t := \EE (\text{IDD}_{t} \mid \fil_t) = \sum_{k\in S_{t+1}} W_{k,t}.
\end{equation}

\end{example}
IDD and IADD are incremental-and-compound versions of detection delay and average detection delay (ADD), which are commonly used to measure false non-detection in single-stream sequential change detection \cite{tartakovsky2014sequential}. Specifically, by taking summation over $t$, we have
\begin{equation} \label{eq:sum-idd}
    \sum_{s=0}^{t-1} \text{IDD}_s = \sum_{k = 1}^K (N_{k}\wedge t - \tau_{k}-1)_+
\end{equation}
and
\begin{equation}
    \mathbb{E}\big(\sum_{s=0}^{t-1} \text{IADD}_s \big)  = \EE \big( \sum_{k = 1}^K  (N_{k}\wedge t - \tau_{k} -1)_+ \big).
\end{equation}

Among the above examples, LFWER, GLFWER, and LFNR  are error rates for false non-detections, LFDR  is an error rate for false detections, IARL estimates the number of pre-change streams that are active, and IADD estimates the number of post-change and active streams. Because a  small value of LFWER (or GLFWER/LFNR/IADD) and a large value of IARL (or minus LFDR) is desired, we could choose the risk process $R_{t}\in\{\text{LFWER}_t, \text{GLFWER}_t, \text{LFNR}_t,\text{IADD}_t\}$ and the utility process $U_t\in\{\text{IARL}_t, -\text{LFDR}_t\}$, or $R_t\in\{\text{LFDR}_t, -\text{IARL}_t\}$ and $U_t\in \{-\text{LFWER}_t, -\text{GLFWER}_t, -\text{LFNR}_t,-\text{IADD}_t\}$. Note that in the above examples, there is a trade-off between $R_t$ and $U_t$. That is, if one declares detection at more data streams, then the corresponding LFWER, GLFWER, LFNR, and IADD tend to be smaller and IARL and minus LFDR tend to be smaller as well. Thus, the optimality criteria (Definitions~\ref{def:uniform-optimal} and \ref{def:local-optimal}) formulated through constrained optimization are reasonable.

The choices of $R_t$ and $U_t$ should be application-driven. In practice, we suggest to choose the risk process $R_t$ with a known range so that the tolerance level is easy to specify. For example, LFWER, GLFWER, LFNR, and LFDR  represent certain probability/expected proportions that are known to be between $[0,1]$. Thus, they are sensible choices of $R_t$, for which setting the tolerance level $\alpha\in[0,1]$ is relatively straightforward.

\section{Proposed Sequential Decision Procedures}\label{sec:method}
In this section, we first provide a formula for computing the posterior probability $W_{k,t}=\prob(\tau_k<t|\fil_t)$, which is a key quantity in computing the risk and utility measures. Then, we present our proposed sequential decisions for controlling the risk process at a given level, followed by a simplified version of the algorithm to reduce the computational complexity.
\subsection{Recursive Formula for \texorpdfstring{$W_{k,t}$}{Wk,t}}
Recall $\pi_{s}=\prob(\tau_{k}=s)$ and $\bar{\pi}_{s}=\prob \big (\tau_{k} \geq s \big)$. Let
\begin{equation}
    Q_{k,t}
    = \bar{\pi}_{t}^{-1}\sum_{s=0}^{t-1} \pi_s \prod_{r=s+1}^{t} \frac{q_{k,r}\left(X_{k,r} \right)}{p_{k,r}\left(X_{k,r}\right)} \text{ with } Q_{k,0}=0.
\end{equation}
Given $Q_{k,t}$ and $X_{k,t+1}$, $Q_{k,t+1}$ can be computed using the recursive formula 
\begin{equation} \label{eq:rec2}
Q_{k, t+1} = \bar{\pi}_{t+1}^{-1}\big(\bar{\pi}_{t}Q_{k, t}+ \pi_{t}  \big) L_{k,t+1},
\end{equation}
where we define $L_{k,t+1}= q_{k,{t+1}}\left(X_{k,{t+1}}\right) / p_{k,{t+1}}(X_{k,{t+1}})$. Then, we obtain
\begin{equation} \label{eq:rec}
W_{k, 0}=0 \text{ and  } W_{k,t} =\frac{Q_{k, t}}{Q_{k, t}+ 1}.
\end{equation}

The above recursive equations \eqref{eq:rec2} and \eqref{eq:rec} are extensions of classic results in single-stream Bayesian sequential change detection problems \cite{polunchenko2012state}. Their rigorous justifications are given in Appendix~\ref{sec:proof-eqs}.

\subsection{Proposed Sequential Decision for Unstructured Risk and Utility}
We first propose a one-step selection rule to select $S_{t+1}$, given $S_t$ and $\{W_{k, t}\}_{k \in S_{t}}$ so that 
the risk $R_{t}$ is controlled to be no greater than $\alpha$. This one-step selection rule goes over all $2^{|S_t|}$ possible subsets of $S_t$, and then select the one which attains the highest utility $U_t$. Algorithm \ref{alg:gen} implements this idea.
\begin{algorithm}%
	\caption{One-step selection rule at time $t$.}
	\label{alg:gen}
	\begin{algorithmic}[1]
    	\STATE \textbf{Input:} Tolerance level $\alpha$, the current index set $S_{t}$, and posterior probabilities $\{W_{k, t}\}_{k \in S_{t}}$, where $W_{k, t}=\prob\left(\tau_{k}<t \mid \mathcal{F}_{t}\right)$ is computed according to \eqref{eq:rec}.
    	
        \STATE For all $S \subset S_{t}$, compute $\gamma_{S} = r_{t}(\{W_{k,t}\}_{k\in S_{t} }, S_{t}, S)$ and $\mu_{S} = u_{t}(\{W_{k,t}\}_{k\in S_{t}} S_{t}, S)$.
        
    	\STATE \textbf{Output:} 
    	\begin{equation*}
    	\begin{split}
    	        S_{t+1}=&\argmax_{S} \mu_{S} 
    	       \quad\text{ subject to }  \gamma_{S}\leq \alpha \text{ and }  S\subset S_t.\footnotemark[1]
    	\end{split}
    	\end{equation*}
	\end{algorithmic}
\end{algorithm}
\footnotetext[1]{If the solution is not unique, $S_{t+1}$ can be any one of the solutions.}
According to Assumption~\ref{assump:non-empty}, $\{S: \gamma_{S}\leq \alpha \text{ and }  S\subset S_t \}\neq \emptyset$. Thus,
$S_{t+1}$ in line 3 of the above algorithm is well-defined.
The next proposition states that
the above one-step selection rule can control the risk process at any given level.
\begin{prop} \label{prop:onestep}
Under Assumption~\ref{assump:non-empty}, the index set $S_{t+1}$ selected by Algorithm \ref{alg:gen} satisfies
$
R_t \leq \alpha 
$ a.s. 
\end{prop}
Note that Proposition~\ref{prop:onestep} does not require any assumptions on $R_t$ and $U_t$ except for Assumption~\ref{assump:non-empty}, which ensures the existence of the set $S_{t+1}$ in the last line of Algorithm~\ref{alg:gen}.

Next, we combine Algorithm~\ref{alg:gen} at different time points to obtain a sequential decision in $\mathcal{D}_{\alpha}$. At each time $t$, this sequential decision selects $S_{t+1}$ using Algorithm \ref{alg:gen} and deactivates data streams that are not in the index set. Algorithm~\ref{alg:dec} below implements this idea.

\begin{algorithm}%
	\caption{Proposed sequential decision $\pD$.}
	\label{alg:dec}
	\begin{algorithmic}[1]
		\STATE {\textbf{Input:} Tolerance level $\alpha$.}
        
        \STATE{Initialize: set $t=1$, $S_{t}=\langle K \rangle$ and compute $W_{k, t}$ for $k \in S_{t}$ using equations \eqref{eq:rec2} and \eqref{eq:rec}.}
        
        \STATE{Select: input $\alpha, S_{t}$ and $\left(W_{k, t}\right)_{k \in S_{t}}$ to Algorithm \ref{alg:gen}, and obtain $S_{t+1}$.}
        
        \STATE{Update: deactivate streams in $S_t \setminus S_{t+1}$. If $S_{t+1} = \varnothing$, stop; otherwise, update $\{W_{k, t+1}\}_{k \in S_{t+1}}$ using equations \eqref{eq:rec2} and \eqref{eq:rec}.}
        
        \STATE{Iterate: set $t = t+1$ and return to line 2.}

		\STATE \textbf{Output:} $\{S_{t}\}_{t \geq 1}$.
	\end{algorithmic}
\end{algorithm}

\begin{prop} \label{prop:pD}
Under Assumption~\ref{assump:non-empty}, $\pD \in \mathcal{D}_{\alpha}$. That is, the proposed sequential decision  given by Algorithm \ref{alg:dec} controls the risk process at level $\alpha$ at every time point.
\end{prop}

\subsection{Simplified Sequential Decision for `Monotone' Risk}
At each time $t$, directly applying Algorithm \ref{alg:gen} requires evaluating and comparing the risk and utility associated with $2^{|S_t|}$ subsets, which is computationally intensive when $|S_t|$ is large. In many cases where the risk and utility satisfy additional monotonicity assumptions, this algorithm can be simplified, reducing the computational complexity  significantly.
In this section, we provide one such assumption, under which the proposed sequential decision only requires evaluating and comparing the risks associated with $|S_t|+1$ subsets.
\begin{assump} \label{assump:one-step optimal}
For all non-empty $S_0\subset \langle K \rangle$,   $\ww=(w_1,\cdots,w_{|S_0|})\in [0,1]^{|S_0|}$, $S\subset S_0$, $i\in S, j\in S_0\setminus S$, and $w_i \geq w_j$, we have $r_t(\ww, S_0, S) \geq r_t(\ww, S_0, (S\setminus\{i\})\cup \{j\})$ and $u_t(\ww, S_0, S) \leq u_t(\ww, S_0, (S\setminus\{i\})\cup \{j\})$.
\end{assump}

Under Assumption~\ref{assump:one-step optimal}, $R_t$ tends to become larger and $U_t$ tends to become smaller if we keep streams with relatively smaller posterior probability active.
Under this assumption, Algorithm~\ref{alg:gen} can be simplified to the following Algorithm~\ref{alg:simp}, %
and it also controls $R_{t}$ to be below a pre-specified level $\alpha$.  As will be discussed in Corollary~\ref{corollary:eg:one-step optimal-monotone risk}, all the risk and utility measures presented in Examples \ref{eg:LFWER} -- \ref{eg:detection-delay} satisfy this assumption.

The following Algorithm \ref{alg:simp}  selects $S_{t+1}$ so that streams with relatively large posterior probabilities are detected and those with relatively small posterior probabilities are kept active. The cut-off point for the detection is decided by maximizing the utility while controlling the risk at time $t$. Because Algorithm \ref{alg:simp} restricts $S_{t+1}$ to be a subset of streams with relatively small posterior probability, it only involves evaluating and comparing the risk and utility functions associated with $|S_t|+1$ subsets, and, thus, reduces the computational complexity to the order $O(|S_t|\log(|S_t|))$.

\begin{algorithm}%
	\caption{Simplified one-step selection rule.}
	\label{alg:simp}
	\begin{algorithmic}[1]
		\STATE \textbf{Input:} Tolerance level $\alpha$, the current index set $S_{t}$, and posterior probabilities $\{W_{k, t}\}_{k \in S_{t}}$. 
		\STATE Arrange posterior probabilities in an ascending order\footnotemark[2], i.e. $$W_{k_{1}, t} \leq W_{k_{2}, t} \leq \cdots \leq W_{k_{| S_{t}|},t}.$$ 
		\STATE For $n=1, \ldots,\left|S_{t}\right|,$ compute 
		$$\gamma_{n}= r_t (\{W_{k,t}\}_{k\in S_{t}}, S_{t}, \{k_i\}_{i=1}^n)$$ 
		and $$\mu_{n} = u_t ( \{W_{k,t}\}_{k\in S_{t}}, S_{t}, \{k_i\}_{i=1}^n).$$
		For $n = 0$, compute $\gamma_{0}= r_t (\{W_{k,t}\}_{k\in S_{t}}, S_{t}, \varnothing)$ and $\mu_{0} = u_t (\{W_{k,t}\}_{k\in S_{t}}, S_{t}, \varnothing).$
        
        \STATE Set $n^*\in\{0,\cdots, |S_t|\}$ as the solution to the problem\footnotemark[3]
        $$
        n^*=\argmax_n\mu_n  \text{ subject to }
        \gamma_{n} \leq \alpha.
        $$
		\STATE \textbf{Output:} $S_{t+1}=\left\{k_{1}, \ldots, k_{n^*}\right\}$ if $n^* \geq 1$ and $S_{t+1}=\varnothing$ if $n^*=0$.
	\end{algorithmic}
\end{algorithm}
\footnotetext[2]{If $W_{k_i,t}=W_{k_j,t}$ for $1\leq i<j\leq n$, we choose $k_i<k_j$ to avoid additional randomness because of ties.}
\footnotetext[3]{If the solution is not unique, we choose $n^*$ to be the largest solution. \label{footnote:largest-n}}
Note that under Assumption~\ref{assump:non-empty},  $\{n:\gamma_n\leq\alpha\}\neq\emptyset$. Thus, the forth line of the above Algorithm~\ref{alg:simp} is well-defined.
The following Algorithm~\ref{alg:dec2} gives an overall sequential decision rule $\sD$ by adopting Algorithm~\ref{alg:simp} at every time point.
\begin{algorithm}%
	\caption{Simplified decision procedure $\sD$.}
	\label{alg:dec2}
	\begin{algorithmic}[1]
		\STATE \textbf{Input:} Tolerance level $\alpha$. %
		\STATE Initialize: set $t=1$. $S_{t}=\langle K \rangle$ and compute $W_{k, t}$ for $k \in S_{t}$ using equations \eqref{eq:rec2} and \eqref{eq:rec}. 
        \STATE Select: input $\alpha, S_{t}$ and $\left(W_{k, t}\right)_{k \in S_{t}}$ to Algorithm \ref{alg:simp}, and obtain $S_{t+1}$.
        \STATE Update: deactivate streams in $S_t \setminus S_{t+1}$. If $S_{t+1} = \varnothing$, stop; otherwise, update $\{W_{k, t+1}\}_{k \in S_{t+1}}$ using equations \eqref{eq:rec2} and \eqref{eq:rec}.
        \STATE Iterate: set $t = t+1$ and return to Step 3.
		\STATE \textbf{Output:} $\{S_{t}\}_{t \geq 1}$.
	\end{algorithmic}
\end{algorithm}

\begin{prop} \label{prop:simp}
Under Assumptions~\ref{assump:non-empty}%
, the sequential decision $\sD$ given by Algorithm \ref{alg:dec2} satisfies $\sD\in\mathcal{D}_{\alpha}$.
\end{prop}

\section{Theoretical Properties of Proposed Methods}\label{sec:theory}
In this section, we first show that the proposed sequential decision $\sD$ is locally optimal under very weak assumptions in Section~\ref{sec:one-step-optimality}. Then, we show that the simplified sequential decision $\sD$ is uniformly optimal under stronger model assumptions and additional monotonicity assumptions on risk and utility measures in Section~\ref{sec:uniform-optimality}. We also provide theoretical results on aggregated risk and utility measures of the proposed methods in Section~\ref{sec:aggregated-risk-theory}.

\subsection{Local Optimality Results}\label{sec:one-step-optimality}
The following two theorems show that the proposed sequential decision $\pD$ is locally optimal under Assumption~\ref{assump:non-empty} while $\sD$ is locally optimal under Assumptions~\ref{assump:non-empty} and \ref{assump:one-step optimal}.
That is, they satisfy Definition~\ref{def:local-optimal}. 
\begin{thm} \label{thm:one-step optimality for general risk}
Under Assumption~\ref{assump:non-empty}, the sequential decision $\pD$ described in Algorithm \ref{alg:dec} is locally optimal.
\end{thm}

\begin{thm} \label{thm:one-step optimality for monotone risk}
Under Assumptions \ref{assump:non-empty} and \ref{assump:one-step optimal}, the sequential decision $\sD$ described in Algorithm \ref{alg:dec2} is locally optimal.
\end{thm}

The next corollary applies the above results to examples given in Section~\ref{sec:example}.
\begin{corollary} \label{corollary:eg:one-step optimal-monotone risk}
If $\alpha>0$ and $R_t\in\{\text{LFWER}_t, \text{GLFWER}_t, \text{LFNR}_t,\text{IADD}_t\}$, $U_t\in\{\text{IARL}_t, -\text{LFDR}_t\}$, or $R_t = \text{LFDR}_t$ and $U_t\in \{-\text{LFWER}_t, -\text{GLFWER}_t, -\text{LFNR}_t,-\text{IADD}_t\}$, then the simplified sequential decision $\sD$ is locally optimal. 

If $\alpha <0$ and $R_t = -\text{IARL}_t$ and $U_t\in \{-\text{LFWER}_t, -\text{GLFWER}_t, -\text{LFNR}_t,-\text{IADD}_t\}$, then the simplified sequential decision $\sD$ is locally optimal. 
\end{corollary}

\subsection{Uniform Optimality Results}\label{sec:uniform-optimality}
In this section, we show that the proposed sequential decision rule $\sD$ defined in Algorithm \ref{alg:dec2} is uniformly optimal under stronger assumptions.
We note that the uniform optimality results developed in the current work are non-trivial extensions of those in \cite{chen2019compound}. In particular, we consider a general class of risk and utility measures while \cite{chen2019compound} only allows the risk measure to be LFNR.  Moreover, time-heterogeneous pre/post- change distributions and non-geometric priors for the change points are allowed in the current work. These extensions require a delicate analysis of a special class of monotone functions and stochastic processes defined over a non-Euclidean space. %

The assumptions for establishing the uniform optimality results include monotonicity assumptions on the risk and utility processes 
and assumptions on the pre- and post- change distributions.
We point out that the monotonicity assumptions are made on functions over a special non-Euclidean space
\begin{equation}\label{eq:so}
\So 
 =
  \bigcup_{k=1}^{K}\{(v_{1}, \cdots, v_{k}): 0 \leq v_{1} \leq \cdots v_{k} \leq 1\} \cup\{\varnothing\},
\end{equation}
which contains ordered vectors of different dimensions.
Thus, the definition of monotonicity is non-standard.

Specifically, for functions maps $ \So$ to $\RR$, we define two types of monotonicity. 

\begin{definition}[Entrywise increasing functions] \label{def:entrywise-increasing}
A function $f:\So\to\RR$ is  {\em``entrywise increasing"}, if
    $f (\uu) \leq f (\vv)$ 
    for all $m\in\langle K \rangle$, $\uu = (u_1,\cdots, u_m), \vv = (v_1,\cdots, v_m) \in \So$, satisfying  $u_j \leq v_j$ for $1 \leq j \leq m$. In addition, a function $f$ is {\em ``entrywise decreasing"} if $-f$ is ``entrywise increasing".
\end{definition}

\begin{definition}[Appending increasing functions] \label{def:appending-increasing}
A function $f:\So \rightarrow \RR $ is {\em ``appending increasing''}, if for all $m\in\langle K \rangle$, 
$\uu = (u_1, \cdots, u_m) \in \So$, $f (u_1, \cdots, u_k)  \leq f (\uu)$,  for all $k \leq m$. In addition, $f(\varnothing) \leq f(u_1)$ for $u_1\in [0,1]$.
\end{definition}

For each vector $\vv=(v_1,\cdots,v_m)$, denote its order statistic by $[\vv]=(v_{(1)},\cdots,v_{(m)})$. That is, $[\vv]$ is a permutation of $\vv$ satisfying $v_{(1)}\leq \cdots \leq v_{(m)}$. We can see that if $v_k\in[0,1]$ for all $k\in\langle K \rangle$, then $[\vv]\in\So$.

\begin{assump} \label{assump:risk}
There exists a measurable function $\rtd: \So\to\RR$ such that $r_{t}(\{W_{k,t}\}_{k\in S_{t} }, S_{t}, S_{t+1}) =\rtd \big([\{W_{k,t}\}_{k \in S_{t+1}}]\big)$. In addition, $\rtd$ is entrywise increasing and appending increasing. 
\end{assump}

\begin{assump} \label{assump:utility}
There exists a measurable function $\tilde{u}_t:\So\to\RR$ such that $u_{t}(\{W_{k,t}\}_{k\in S_{t} }, S_{t}, S_{t+1} ) = \tilde{u}_{t} \big([ \{W_{k,t}\}_{k \in S_{t+1}} ]\big)$.
In addition, $\tilde{u}_{t}$ is entrywise decreasing and appending increasing.
\end{assump}

\begin{assump} \label{assump:iid}
The pre-  and post-change distributions $\{p_{k,t}\}_{t \geq 1}$ and $\{q_{k,t}\}_{t \geq 1}$ are the same for different $k\in\langle K \rangle$. That is,  $p_{1,t} = \ldots = p_{K,t}$ and $q_{1,t} = \ldots = q_{K,t}$ for all $t$.
\end{assump}

\begin{thm} \label{thm:uniform optimality}
Under Assumptions \ref{assump:non-empty}, \ref{assump:risk}, \ref{assump:utility} and \ref{assump:iid}, the sequential decision $\sD$ described in Algorithm \ref{alg:dec2} is uniformly optimal.
\end{thm}
\begin{proof}
The proof is involved that requires monotone coupling for stochastic processes living on the space $\So$. It is given in Appendix~\ref{sec:proof-uniform-optimality}.
\end{proof}
We make several remarks on the above theorem. First, under Assumptions~\ref{assump:risk} and \ref{assump:utility},  risk and utility measures are symmetric functions $(W_{k,t})_{k\in S_{t+1}}$. These assumptions rule out the cases (e.g., LFDR defined in Example~\ref{eg:LFDR}) where the risk also depends on $W_{k,t}$s for $k\notin S_{t+1}$, without which the uniform optimal solution may not exist (see Counterexample~\ref{eg:counterexample} below).
Second, under the monotonicity assumptions that $\rtd$s are entrywise increasing, the risk process tends to be larger if the posterior probability of the change points associated with the selected streams is larger. It is also appending increasing, meaning that the risk tends to be larger if more streams are kept active. Similarly, the utility process tends to be larger if fewer streams are kept active and the posterior probabilities associated with the selected streams are smaller. Third, we require the pre- and post- stream distributions $p_{k,t}$ and $q_{k,t}$ to be identical for different streams. In this case, the process $\{W_{k,t}\}_{t\in\mathbb{Z}_+}$ has identical distribution for different $k$ and contributed in a symmetric way to the risk and utility processes.

For most of applications, it is easy to check Assumptions~\ref{assump:non-empty} and \ref{assump:iid}. In some cases, additional efforts are needed to verify monotonicity assumptions described in Assumptions~\ref{assump:risk} and \ref{assump:utility}. Below we provide sufficient conditions for the monotonicity conditions. Note that the risk and utility measures described in Examples~\ref{eg:LFWER}, \ref{eg:GLFWER}, \ref{eg:LFNR}, \ref{eg:ARL}, and \ref{eg:detection-delay} are all symmetric multivariate polynomials of the posterior probabilities. Thus, we restrict the analysis to the polynomial case in the next proposition.

\begin{prop}[Polynomial case] \label{prop:poly1}
Let $\rd: \So \rightarrow \RR$ be a function in the following form
\begin{equation}\label{eq:poly}
\rd(\uu) = \sum_{p=0}^{\infty} \mathds{1}( \dim(\uu) = p) \sum_{k=1}^{p}  C_{p,k} \sum\limits_{i_1 < i_2 < \cdots < i_k} \prod_{j=1}^{k}u_{i_j}.    
\end{equation}
Note that $\rd (\varnothing) = 0$.
If $\rd(\cdot)$ satisfies 
\begin{equation} \label{eq:poly-entrywise-increasing}
\rd(\uu^{i,p-i}) \leq \rd(\uu^{i-1,p-i+1}), \quad \text{ for all } i \in\langle p \rangle \text{ and } p\geq 1,
\end{equation}
where
$\langle p \rangle=\{1,\cdots,p\}$ and $\uu^{i,p-i}$ denotes the $p$ dimensional vector whose first $i$ elements are $0$ and last $p-i$ elements are all $1$, then $\tilde{r}$ is entrywise increasing.

Moreover, if $\tilde{r}$ satisfies \eqref{eq:poly-entrywise-increasing} and 
\begin{equation} \label{eq:poly-appending-increasing}
    \rd(\uu^{i,p-i}) \leq \rd(\uu^{i+1,p-i}), \quad
    \text{ for all } i \in \{0,\cdots, p\} \text{ and } p\geq 0, %
\end{equation}
then $\rd$ is also appending increasing.
\end{prop}

\begin{remark} \label{corollary:poly2}
The inequalities \eqref{eq:poly-entrywise-increasing} and \eqref{eq:poly-appending-increasing} are equivalent to
\begin{align}
&\sum_{k=1}^{p-i}  C_{p,k}   { p-i-1 \choose k-1}  \geq 0 ,
\quad \text{ for i} =0, \cdots, p-1, \\
&\sum_{k=1}^{p-i}  \left(C_{p+1,k} - C_{p,k} \right) {p-i \choose k} \geq 0, \text{  for i} =0,\cdots, p, \label{eq:appending}
\end{align}
and all $p\geq 0$. We leave the rigorous justification for the above statements in Appendix~\ref{sec:proof-uniform-optimality}.
\end{remark}

Now we apply the uniform optimality result in Theorem~\ref{thm:uniform optimality} to performance measures described in Examples \ref{eg:LFWER}, \ref{eg:GLFWER}, \ref{eg:LFNR} and \ref{eg:ARL}.
\begin{corollary} \label{corollary:eg}
If $\alpha>0$, $R_t\in\{\text{LFNR}_t,\text{LFWER}_t, \text{GLFWER}_t, \text{IADD}_t\}$, and $U_t= \text{IARL}_t$, then under Assumption \ref{assump:iid}, $\sD$ is uniformly optimal.
\end{corollary}

We point out that $\text{LFDR}_t$ described in Example~\ref{eg:LFDR} does not satisfy the assumptions made in Theorem~\ref{thm:uniform optimality}. Thus, we do not have uniform optimality results for it. Indeed, if   $R_t=\text{LFDR}_t$, then the uniformly optimal sequential decision may not exist. A counterexample is given below.
\begin{cexample} \label{eg:counterexample}
	Let $K = 3$,  $p_{k,t}(x)=(0.01)^x(0.99)^{1-x}$, $q_{k,t}(x)=(0.99)^x(0.01)^{1-x}$, and $\mathbb{P}(\tau_k=l)=1/3$ for $k=1,2,3$, $x=0,1$,  $l=0,1,2$, and $t\geq 1$. That is, the pre- and post- change distributions are Bernoulli distributions with parameters $0.01$ and $0.99$, respectively, and $\tau_k$s are uniformly distributed over $\{0,1,2\}$.
	We further assume that the tolerance level $\alpha = 0.51$, the risk process $R_t=\text{LFDR}_t$ (see Example \ref{eg:LFDR}) and the utility process $U_t=-\text{IADD}_t$ (see Example \ref{eg:detection-delay}). 
	
	In this setting, there does not exist a sequential decision achieving the maximum of the expected utility at both times $1$ and $2$. This implies that there is no uniformly optimal sequential decision. We leave detailed calculation in Appendix \ref{sec:proof-uniform-optimality}.
\end{cexample}

\subsection{Implications on Aggregated Risk}\label{sec:aggregated-risk-theory}
Let  $\{ a_t \}_{t \geq 1}$ be a sequence of non-negative random variables satisfying $\sum_{t = 1}^{\infty} a_t = 1$, and $\{b_t\}_{t\geq 1}$ be a sequence of non-negative constants. Consider the following aggregated risk (AR) and aggregated utility (AU),
\begin{equation}\label{eq:AR}
    \text{AR}=\EE\big(\sum_{t=1}^{\infty}a_t R_t\big) \text{ and }  \text{AU}=\EE\big(\sum_{t=1}^{\infty}b_t U_t\big).
\end{equation}
The aggregated risk and utility metrics defined above provide a summary of the performance across time. These types of risk and utility measures are considered in many recent works on multi-stream sequential change detection and hypothesis testing, including \cite{bartroff2020sequential,song2019sequential,song2017asymptotically, chen2020false}.

The next proposition shows that if the risk process is controlled at the desired tolerance level at every time point, then the aggregated risk is also controlled at the same level.
\begin{prop} \label{prop:AR}
Let $\delta\in\mathcal{D}_{\alpha}$ and $\text{AR}(\delta)$ be the corresponding aggregated risk defined in \eqref{eq:AR}. Then,
$
\text{AR} (\delta) \leq \alpha.
$
\end{prop}

Note that the reverse statement does not hold. That is, the aggregated risk being controlled does not imply the risk at each time $t$ being controlled.

The next proposition shows that a uniformly optimal sequential decision also maximizes the aggregated utility. 
\begin{prop} \label{prop:AU}
Suppose that $\delta$ is uniformly optimal in $\mathcal{D}_{\alpha}$. Then,
for the aggregated utility defined in \eqref{eq:AR},
$$
\text{AU} (\delta)  
= \sup_{\delta^{\prime} \in \mathcal{D}_{\alpha}} \text{AU} (\delta^{\prime}),
$$
where $\text{AU}(\delta)$ and $\text{AU}(\delta^{\prime})$ denote the aggregated utility associated with $\delta$ and $\delta^{\prime}$, respectively.
\end{prop}

Next, we use Propositions~\ref{prop:AR} and \ref{prop:AU} to make a connection between the current results and recent works on  the sequential multiple testing and parallel sequential change detection \cite{song2017asymptotically,chen2020false,chen2019compound}.

\subsubsection{Controlling generalized error rates in multi-stream sequential hypothesis testing}
 Note that if $\pi_0+\pi_{\infty}=1$ (i.e., change points either occur at the beginning or never occur), the sequential change point detection problem reduces to a sequential multiple hypotheses testing problem, where the goal is to choose between $H_0^k$ and $H_1^k$ for $k=1,\cdots K$,
 $$
 H_0^k: X_{k,t} \sim p_{k,t} \text{ for all } t \text{  against  }  H_1^k:  X_{k,t}\sim q_{k,t} \text{ for all }t,
 $$
 under a Bayesian setting, where $\prob(H_0^k \text{ holds})=\pi_{\infty}$ and $\prob(H_1^k \text{ holds})=\pi_0$. In addition, we assume that $X_{k,t}$ are jointly independent.%

Let $m \geq1$. We define the generalized family-wise error rate (GFWER) as
\begin{equation} \label{eq:GFWER}
\text{GFWER}_m: = \prob(E_{m,T}),
\end{equation}
where $T$ is a stopping time and the event $E_{m,t}$ is defined in \eqref{eq:def-Emt}. $\text{GFWER}_m$ can be viewed as a generalized family-wise error rate measuring type-II errors in sequential multiple hypotheses testing, which takes a similar form as the generalized type-II error rate in \cite{song2017asymptotically,song2019sequential,bartroff2018multiple,bartroff2014sequential}. Specifically, if we reject $H_0^k$ at time $t$ if and only if $k\in S_{t+1}$. Then,
$$
E_{m,t} = \Big\{ \sum_{k=1}^K \ind (H_1^k \text{ holds and } H_0^k \text{ is chosen at time }t  ) \geq m \Big\}
$$
in the context of multiple hypotheses testing. 

The next corollary of  Proposition~\ref{prop:AR} shows that the proposed method $\sD$ controls the GFWER in the perspective of the hypothesis testing problem.
\begin{corollary} \label{corollary:HT}
Given the tolerance level $\alpha$, consider the sequential decision $\sD$ in Algorithm \ref{alg:simp} with the risk process $R_t=\text{GLFWER}_{m,t}$ defined in Example \ref{eg:GLFWER} and any utility process. Then, the generalized family-wise error rate $\text{GFWER}_m$ defined in \eqref{eq:GFWER} is also controlled to be no greater than $\alpha$ for any stopping time $T$.
\end{corollary}
Note that the above Corollary \ref{corollary:HT} holds for any stopping time $T$. In particular, if we let $T$ grow to infinity, then Corollary~\ref{corollary:HT} states that the GFWER accumulated over all the time points is controlled to be no greater than $\alpha$. If we let $T= T_1$ as defined in Remark~\ref{remark:connection-with-Poor}, then different data streams are stopped at the same detection time $T_1$. In this case, the proposed sequential procedure belongs to the class of sequential multiple testing procedures described in \cite{song2019sequential}.

\subsubsection{Controlling aggregated false discovery rate}
The aggregated false discovery rate (AFDR) is considered in \cite{chen2020false},
\begin{align} %
\text{AFDR} &= \EE \left( \frac{\sum_{t = 1}^{\bar{N}-1} \sum_{k = 1}^K \mathds{1}\left(\tau_{k} \geq t, N_k = t \right)}{ \left(\sum_{t = 1}^{\bar{N}-1} \sum_{k = 1}^K \mathds{1}\left( N_k = t \right)\right) \vee 1} \right), \label{eq:fdr}
\end{align}
where $\bar{N}$ is a positive integer that is referred to as a `deadline'. 
The next proposition states that any decision that controls $\text{LFDR}_t$ at every time also controls AFDR asymptotically.

\begin{prop} \label{prop:coonection-to-poor}
Let $R_t=\text{LFDR}_t$ and $\delta \in \mathcal{D}_{\alpha}$. 
Assume that there exist a sequence of constants $\{C_t\}_{t\geq1}$ and a sequence of random variables $\{A_t\}_{t\geq1}$ such that $K^{-1} \sum_{k = 1}^K \ind \left( N_k = t \right)$ converges to $C_t$ in probability and $\text{FDP}_t$ converges to $A_t$ in probability for all $t\geq 1$ as $K$ grows to infinity. Then,  $\lim_{K \rightarrow \infty} \text{AFDR}(\delta) \leq \alpha$. That is, $\text{AFDR}$ is controlled to be no greater than $\alpha$ asymptotically.
\end{prop}

\subsubsection{Maximizing total average run length}
Let the total average run length (TARL) be
\begin{equation}
    \text{TARL} = \sum_{k=1}^K (\tau_k\wedge N_k),
\end{equation}
where $N_k$ is defined in \eqref{eq:Nk}. TARL aggregates IRL$_t$ across different time points, and can be viewed as an extension of the classic ARL2FA to multi-stream problems.
The next corollary of Proposition~\ref{prop:AU} shows that the proposed method also maximizes TARL under certain conditions.
\begin{corollary}\label{corollary:EU}
Under Assumptions~\ref{assump:iid}, and $R_t\in \{\text{LFNR}_t,\text{LFWER}_t, \text{GLFWER}_t, \text{IADD}_t\}$,
$$
\EE(\text{TARL}(\sD)) = \sup_{\delta\in\mathcal{D}_{\alpha}}
\EE(\text{TARL}(\delta)),
$$
where $\sD$ is obtained from Algorithm~\ref{alg:dec2} by letting $U_t=\text{IARL}_t$, and
 $\text{TARL}(\sD)$ and $\text{TARL}(\delta)$ denote the the total average run length (TARL) of the decision $\sD$ and $\delta$, respectively.
\end{corollary}

\section{A Simulation Study}\label{sec:sim}

In this section, we study the performance of the proposed sequential decision $\sD$ defined in Algorithm~\ref{alg:dec2} through a simulation study. We choose $R_t=\text{LFDR}_t$ and $U_t=-\text{IADD}_t$ in  the simulation study and 
let the tolerance level $\alpha = 0.1$. We also compare the performance of the proposed method with the MD-FDR method proposed in \cite{chen2020false}.

We also conduct a simulation study where $R_t=\text{LFNR}_t$.
For space reason, we leave it to the appendix in the supplementary materials.

We let
$p_{k,t}(x) = (2\pi)^{-1/2}e^{-\frac{x^2}{2}}$
and $q_{k,t}(x) = (2\pi)^{-1/2}e^{-\frac{(x-1)^2}{2}}$
for all $k$ and $t$. In addition, let  $\pi_{\infty}=0.2$, and $
\pi_t = 0.1\cdot 0.8\cdot (0.9)^{t}
$ for $t\geq 0$.  
That is, we set the pre- and post-change probability distributions to be the Gaussian distributions $\mathcal{N}(0,1)$ and $\mathcal{N}(1,1)$, respectively, and set the prior distribution for the change point $\tau_k$ to be a mixture of a point mass at infinity and a geometric distribution.

We assess and compare the performance of two  sequential decisions. The first sequential decision is the the proposed method $\sD$ described in Algorithm \ref{alg:dec2} with $R_t=\text{LFDR}_t$ (defined in Example \ref{eg:LFDR}) and $U_t=-\text{IADD}_t$ (defined in Example \ref{eg:detection-delay}). With this choice of $R_t$ and $U_t$, line 4 in Algorithm \ref{alg:simp} can be simplified as
$$
 n^* =\arg\min\limits_{n = 0,1, \ldots,\left|S_{t}\right| } \{n: \gamma_{n} \leq \alpha \}.
$$
The other sequential decision is the MD-FDR method developed in \cite{chen2020false}. Following the MD-FDR method, the risk measure AFDR defined in \eqref{eq:fdr} is guaranteed to be no greater than the tolerance level $\alpha$.

We first compare the proposed method with the MD-FDR method in terms of their FDP$_t$ (defined in \eqref{eq:FDP}) and $\text{IDD}_t$ (defined in \eqref{eq:IDD}) for  fixed $K=500$ with $1000$ independent Monte Carlo simulations. The  averaged $\text{FDP}_t$ and $\text{IDD}_t$ across the $1000$ replications are plotted in Figures \ref{fig:sim1_FDPt} and \ref{fig:sim1_IDDt} as functions of $t$. 
According to Figure~\ref{fig:sim1_FDPt},  the averaged $\text{FDP}_t$ of both methods are below $0.1$ for all $t$ with a trend of first increasing and then decreasing as $t$ increases. The $\text{FDP}_t$ of the proposed method has a plateau near $0.1$ for $t \in [5,20]$. In addition, the $\text{FDP}_t$ of the proposed method is larger than that of the MD-FDR method, which suggests that the proposed method is less conservative while still controlled under the target tolerance level. Figure~\ref{fig:sim1_IDDt} compares the averaged $\text{IDD}_t$ of the proposed method and the MD-FDR method for different $t$. 
It displays that, for both methods, 
$\text{IDD}_t$ first increases and then decreases as $t$ increases.   The proposed method has a lower averaged $\text{IDD}_t$ than the MD-FDR method for all $t$, indicating a smaller detection delay.

Next, we compare the two methods in terms of aggregated performance measures. In particular, we
consider the aggregated risk AFDR defined in \eqref{eq:fdr}, where we set the deadline parameter $\bar{N} = 500$. For aggregated utility, we consider the  the total average detection delay (TADD), defined as 
\begin{equation} \label{eq:TADD}
\text{TADD} = \EE\Big(\sum_{s=0}^{\bar{N} -1} \text{IDD}_s \Big) = \EE\Big(\sum_{s=0}^{\bar{N} -1} \text{IADD}_s \Big),
\end{equation}
where $\text{IDD}_s$ is defined in \eqref{eq:sum-idd}. Then, we let the aggregated utility  be $\text{AU}=-\text{TADD}$. A higher utility, which corresponds to a lower TADD,  reflects a quicker detection of the changes. 

Tables~\ref{tbl:sim1a_fdr} and \ref{tbl:sim1a_add} compare the two methods in terms of their aggregated risk AFDR and the aggregated utility TADD,  respectively, which are estimated based on a Monte Carlo simulation with 1000 replications. From Table~\ref{tbl:sim1a_fdr}, we can see that both the proposed method and MD-FDR method control AFDR below the tolerance level $\alpha = 0.1$, while the MD-FDR method is more conservative. We also note that as $K$ grows larger, AFDR of the proposed method is approaching $\alpha=0.1$. %
From Table \ref{tbl:sim1a_add}, we can see that the TADD of the proposed method is significantly less than that of the MD-FDR method, indicating that the proposed method detects changes faster than the MD-FDR method, when the AFDR of both methods are controlled at the same level.
An interesting observation is that  TADD of both methods scale with $K$ as $K$ grows. That is, $\text{TADD}/K$ seems to converge to a constant as $K$ grows large. 
Specifically, for the proposed method, $\text{TADD}/K$ is around $3.9$. For  the MD-FDR method, $\text{TADD}/K$ is around $6.5$ for large $K$.

Overall, these results suggests that the proposed method is less conservative and adapts better to the tolerance level than the MD-FDR method.
\begin{figure}
\begin{center}
\includegraphics[scale = 0.50]
{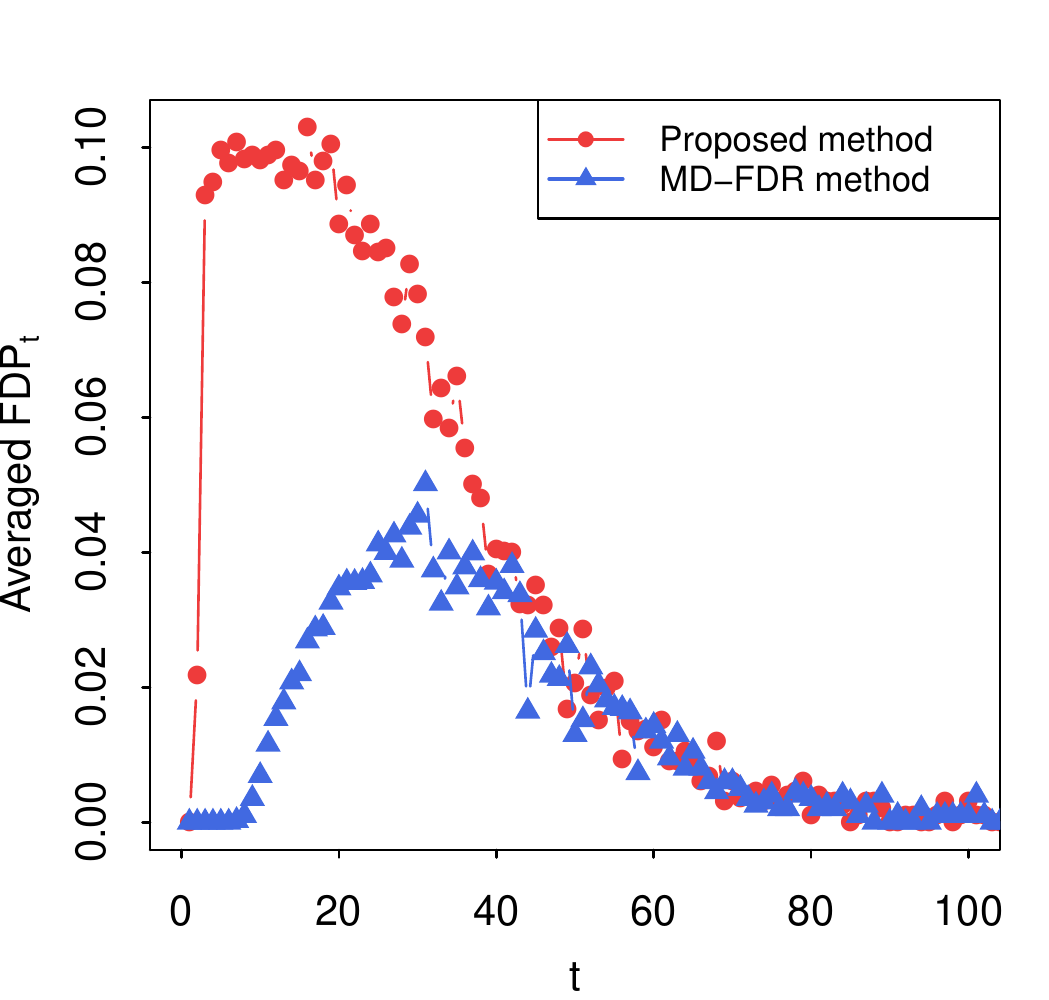}
\end{center}
\caption{$\text{FDP}_t$ averaged over 1000 Monte Carlo simulations at $K = 500$
}
\label{fig:sim1_FDPt}
\end{figure}

\begin{figure}
\begin{center}
\includegraphics[scale = 0.50]
{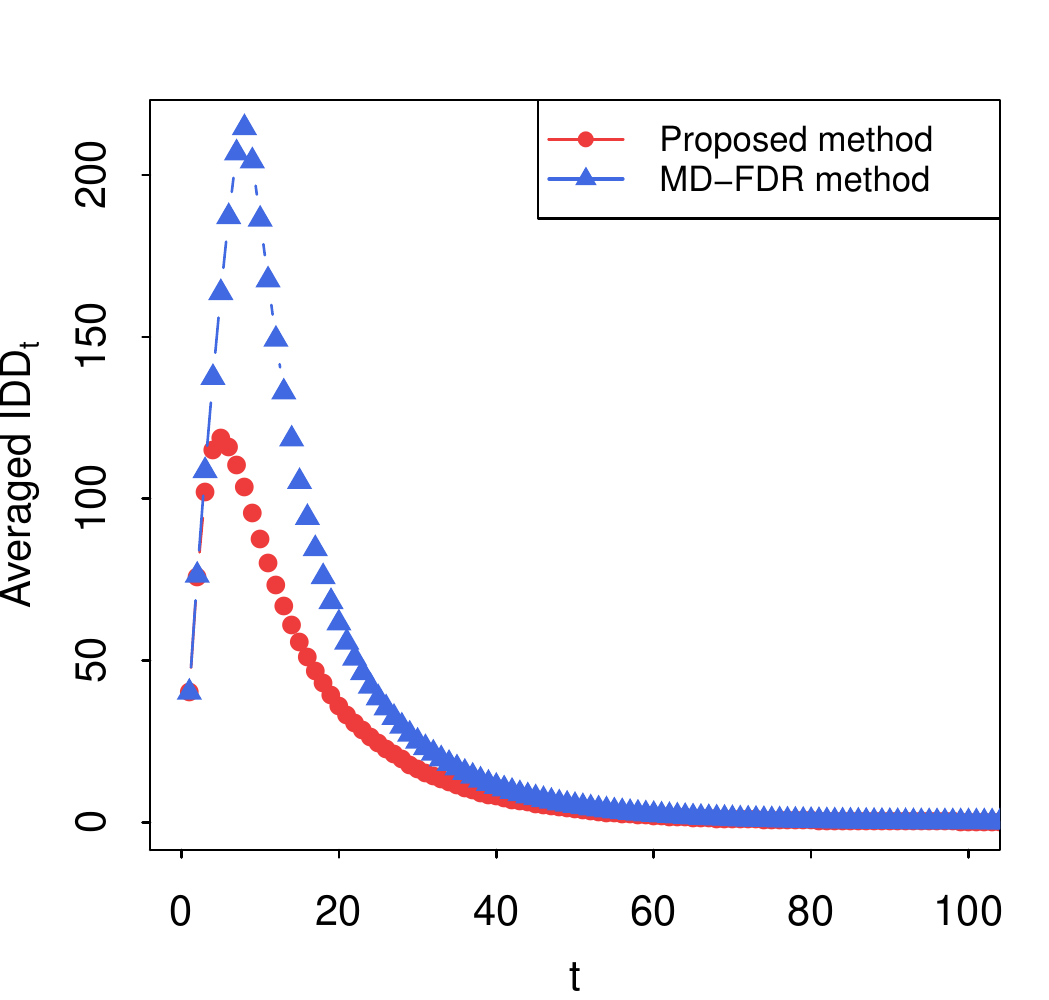}
\end{center}
\caption{$\text{IDD}_t$ averaged over 1000 Monte Carlo simulations at $K = 500$ 
}
\label{fig:sim1_IDDt}
\end{figure}

\begin{table}[ht]
\centering
\begin{tabular}{ccc}
\hline
K & Proposed method & MD FDR method \\ 
\hline
10 & $7.0\times 10^{-2}$ ($3\times10^{-3}$) & $2.8\times 10^{-2}$ ($2\times10^{-3})$ \\ 
100 & $8.6\times 10^{-2}$ ($9\times10^{-4}$) & $2.4\times 10^{-2}$ ($5\times10^{-4}$) \\ 
200 & $9.2\times 10^{-2}$ ($7\times10^{-4}$) & $2.4\times 10^{-2}$ ($4\times10^{-4}$) \\ 
500 & $9.6\times 10^{-2}$ ($5\times10^{-4}$) & $2.3\times 10^{-2}$ ($2\times10^{-4}$) \\ 
1000 & $9.8\times 10^{-2}$ ($3\times10^{-4}$) & $2.4\times 10^{-2}$ ($2\times10^{-4}$) \\ 
   \hline
\end{tabular}
\caption{Estimated AFDR 
(standard deviation in parenthesis) } 
\label{tbl:sim1a_fdr}
\end{table}

\begin{table}[ht]
\centering
\begin{tabular}{ccc}
  \hline
K & Proposed method & MD FDR method \\ 
  \hline
 10 & 45.8 (0.5) & 61.4 (0.5) \\ 
  100 & 413.8 (1.3) & 650 (1.7) \\ 
  200 & 799.8 (1.9) & 1304.1 (2.4) \\ 
  500 & 1964.9 (3.0) & 3264 (3.7) \\ 
  1000 & 3891.4 (4.0) & 6535.3 (5.2) \\ 
   \hline
\end{tabular}
\caption{Estimated TADD 
(standard deviation in parenthesis)} 
\label{tbl:sim1a_add}
\end{table}

\section{A Case Study: Multi-Channel Spectrum Sensing in Cognitive Radios}
\label{sec:case-study}
In this section, we conduct a case study on a multi-channel spectrum sensing problem for cognitive radios, following the settings described in \cite{chen2020false}. 
Cognitive radios are radios that can dynamically and automatically adjust their operational parameters according to the environment so that the spectrum is utilized more efficiently \cite{haykin2005cognitive,mitola1999cognitive}. To make the most out of a spectrum, a cognitive user is allowed to use the idle spectrum band when the primary user is not transmitting. However, when the primary user starts transmission, the cognitive user should detect the change and vacate the spectrum band as soon as possible. The detection of the transmission of the primary user can be formulated as a sequential change detection problem, where the transmission time corresponds to the change point \cite{lai2008quickest,chen2020false}.

We consider a multi-channel spectrum sensing problem for cognitive radios, where there are $K$ independent frequency channels assigned to $K$ independent primary users. 
The cognitive users monitor the spectrum bands and collect signal samples sequentially. The distribution of the signals will change when a primary user starts transmission. As soon as the change is detected, the cognitive user vacates the spectrum band, so that the primary user can use it without interference. Here, each channel corresponds to a data stream, and the time that a primary user starts transmission corresponds to a change point in that data stream. 
Our goal is to have a sequential decision that can detect the transmission of the primary user at each channel quickly to reduce the interference, while controlling the false discovery rate, which corresponds to the expected proportion of unoccupied channels among the detected ones.

Specifically, we assume that  $X_{k,t}$ is the signal collected from the $k$th cognitive user at time $t$, $\tau_k$ is the time when the $k$-th primary user starts transmission, and $X_{k,t}$s and $\tau_k$s follow the change point model described in \eqref{eq:prior} and \eqref{eq:model}.
For the change point $\tau_k$, we further assume that 
$$
\pi_t =  (1 - \PiInf)\theta (1 - \theta)^{t}
$$
with $\PiInf = 0.1$ and $\theta = 0.05$. That is, $\tau_k$ follows a mixture distribution of a point mass at infinity and a geometric distribution.

For the pre- and post- change distributions, we assume 
$$ 
X_{k,t}=
\begin{cases}
Y_{k,t}, &\text{ if } t\leq \tau_k\\
Y_{k,t}+Z_{k,t},
&\text{ if } t>\tau_k
\end{cases},
$$
where $Y_{k,t}$ denotes a Gaussian white noise and $Z_{k,t}$ denotes the faded received primary radio signal at the cognitive user's end. We further assume that $Y_{k,t}\sim \mathcal{C N}(0,\sigma^2)$, $Z_{k,t}\sim \mathcal{C N}(0,\lambda_k)$, and $Y_{k,t}$s and $Z_{k,t}$s are independent, where $\mathcal{C N}(0,\sigma^2 )$ and $\mathcal{C N}(0,  \lambda_k)$ denote the circularly-symmetric complex Gaussian distributions with mean $0$ and the complex variance $\sigma^2$ and $\lambda_k$, respectively.
Note that a complex random variable  has a circularly-symmetric complex Gaussian distribution with a variance $\sigma^2$ if its real and imaginary parts are independent and identically distributed univariate Gaussian random variables with the mean zero and the variance $\sigma^2/2$. 

Under this model,  $X_{k,t}$ has the distribution
$$ 
X_{k,t}\sim
\begin{cases}
\mathcal{C N}(0,\sigma^2 ), &\text{ if } t\leq \tau_k\\
\mathcal{C N}(0, \sigma^2 + \lambda_k ),
&\text{ if } t>\tau_k
\end{cases}.
$$

Notice that in this setting, the streams share the same pre-change distribution, but have different post-change distributions characterized by their different variances. The above distribution assumptions are commonly adopted in the literature \cite{lai2008quickest,chen2020false}.

In this case study, we assume $\sigma^2 =2$ and sample independent $\lambda_k$s from a uniform distribution over $[1,2]$. We then treat $\lambda_k$ as known parameters. Here, we sample $\lambda_k$ from an interval to mimic the practical situation where the signals sent by the primary users may experience channel attenuation at the cognitive user's end, which results in a range of variance-distinct post-change signals.

Let the tolerance level $\alpha = 0.1$. We compare the performance of the proposed sequential decision following Algorithm~\ref{alg:simp} (with $R_t=\text{LFDR}_t$ and $U_t=-\text{IADD}_t$) and the MD-FDR method proposed in \cite{chen2020false}.  
We also consider the aggregated risks AFDR (defined in \eqref{eq:fdr}) and the aggregated utility TADD (defined in \eqref{eq:TADD}). %

Figures~\ref{fig:case_FDPt} and \ref{fig:case_IDDt} show the averaged $\text{FDP}_t$ and $\text{IDD}_t$ for different $t$ based on a Monte Carlo simulation with 1000 replications. We see that $\text{FNP}_t$ of both methods are below 0.1 with a peak at around $t = 12$ and $t = 42$, respectively. According to Figure~\ref{fig:case_FDPt}, the averaged $\text{FNP}_t$ of the MD-FDR method appears to be smaller than that of the proposed method for time $t < 50$, while both of them decline at a similar rate after time $t = 50$. For larger $t$, the  averaged $\text{FNP}_t$ of both methods are close to zero. According to Figure~\ref{fig:case_IDDt}, the averaged $\text{IDD}_t$ of the MD-FDR method is larger than that of the proposed method for all $t$, which suggests that the proposed method detects changes more quickly than that of the MD-FDR method. 

\begin{figure}%
	\begin{center}
		\includegraphics[scale = 0.50]
		{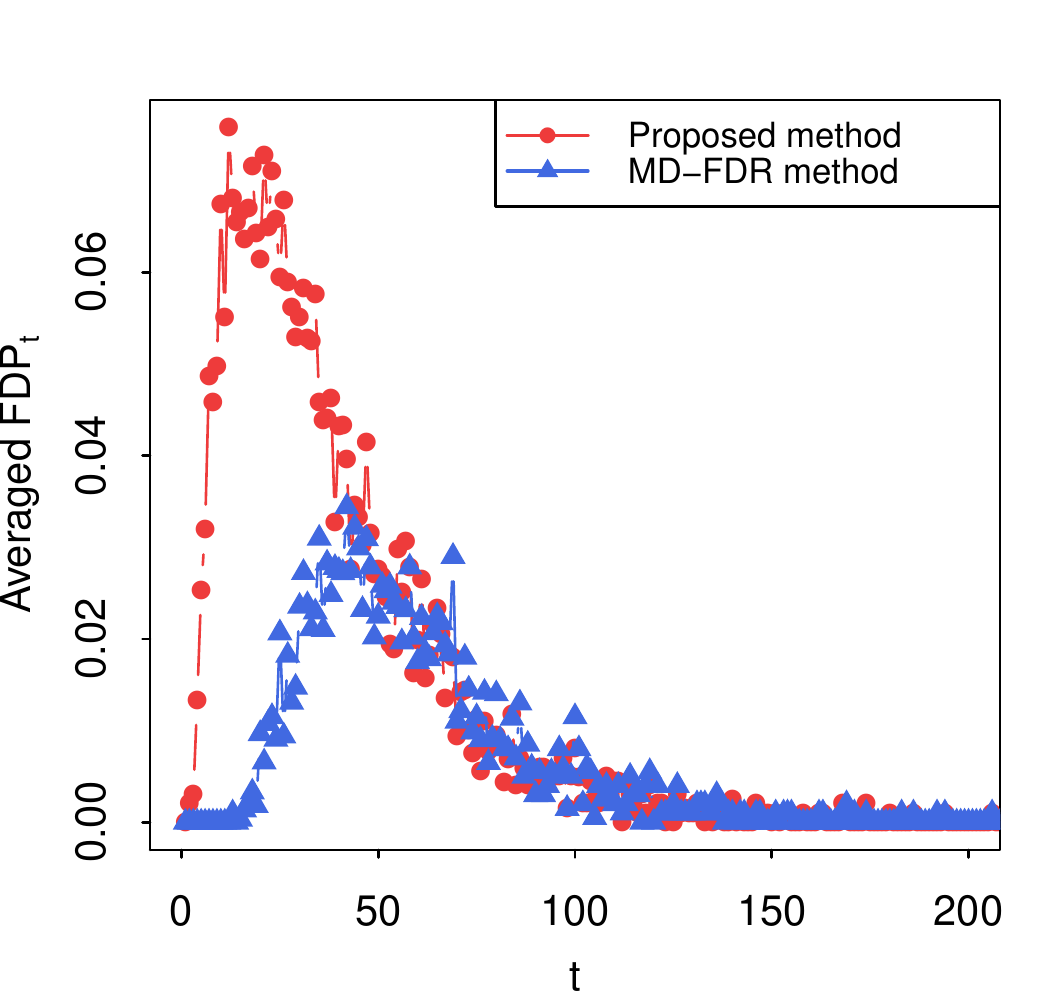}
	\end{center}
	\caption{$\text{FDP}_t$ averaged over 1000 Monte Carlo simulations at $K = 100$ in Case Study}
	\label{fig:case_FDPt}
\end{figure}

\begin{figure}%
	\begin{center}
		\includegraphics[scale = 0.50]
		{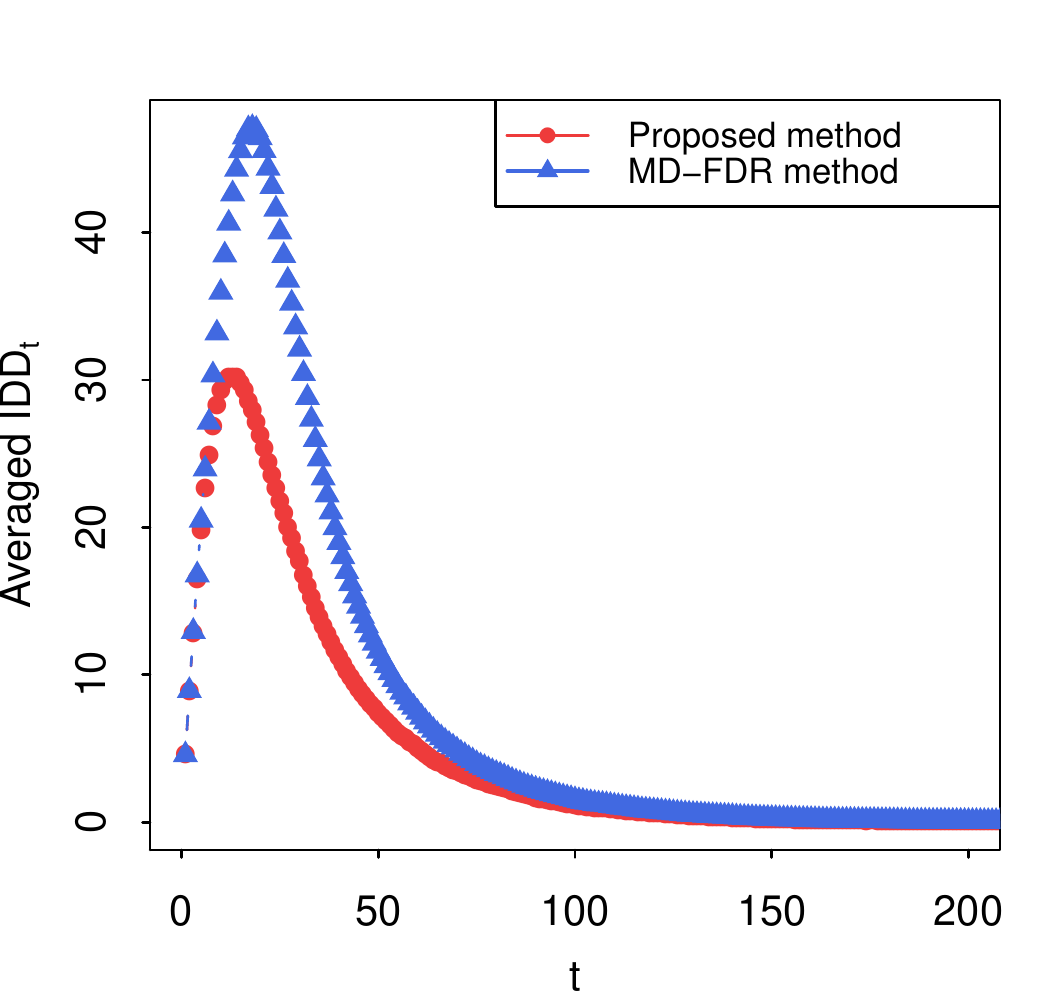}
	\end{center}
	\caption{$\text{IDD}_t$ averaged over 1000 Monte Carlo simulations at $K = 100$ in Case Study}
	\label{fig:case_IDDt}
\end{figure}

Tables~\ref{tbl:case_fdr} and \ref{tbl:case_add} show the AFDR and TADD for both methods for $K\in\{10,100,200,500,1000\}$. According to the tables, the AFDR of both methods are controlled to be less than  $\alpha=0.1$, with the AFDR of the MD-FDR method smaller than that of the proposed method for all $K$. This indicates that the proposed method is less conservative in controlling FDR-type of risks, when compared with the MD-FDR method.
Moreover, the proposed method has a much smaller TADD that that of the MD-FDR method for all $K$, indicating that the proposed method has a smaller detection delay.

\begin{table}%
\centering
\begin{tabular}{ccc}
  \hline
K & Proposed method & MD FDR method \\ 
  \hline
  10 & $6.7\times10^{-2}$ ($3\times10^{-3}$) & $3.2\times10^{-2}$ ($2\times10^{-3}$) \\ 
   100 & $8.5\times10^{-2}$ ($9\times10^{-4}$) & $2.9\times10^{-2}$ ($6\times10^{-4}$) \\ 
   200 & $9.0\times10^{-2}$ ($7\times10^{-4}$) & $2.9\times10^{-2}$ ($4\times10^{-4}$) \\ 
   500 & $9.5\times10^{-2}$ ($4\times10^{-4}$) & $2.8\times10^{-2}$ ($2\times10^{-4}$) \\ 
  1000 & $9.7\times10^{-2}$ ($3\times10^{-4}$) & $2.8\times10^{-2}$ ($2\times10^{-4}$) \\ 
   \hline
\end{tabular}
\caption{Estimated AFDR in case study (standard deviation in parenthesis)} 
\label{tbl:case_fdr}
\end{table}

\begin{table}%
\centering
\begin{tabular}{ccc}
  \hline
K & Proposed method & MD FDR method \\ 
  \hline
 10 & 122.1 (1.2) & 162 (1.4) \\ 
  100 & 1115.8 (3.7) & 1708.5 (4.6) \\ 
  200 & 2178.2 (5.1) & 3434.8 (6.6) \\ 
  500 & 5293.4 (8.1) & 8609.4 (10.4) \\ 
  1000 & 10460.1 (11.3) & 17246.7 (14.8) \\ 
   \hline
\end{tabular}
\caption{Estimated TADD in case study (standard deviation in parenthesis)} 
\label{tbl:case_add}
\end{table}

\section{Conclusions}\label{sec:conc}
The parallel sequential change detection problem
is widely encountered in the analysis of large-scale real-time streaming data. This study introduces a general decision theory framework for this problem, covering many compound performance metrics. It further proposes a sequential procedure under this general framework and proves its optimal properties under reasonable conditions. Simulation and case studies evaluate the performance of the proposed method and compare it with the method proposed in \cite{chen2020false}. The results support the theoretical developments and also show that the proposed method outperforms in our simulation studies and case study.

The current study can be extended in several directions. First, the current parallel sequential change detection framework may be extended to account for multiple types of decisions, including alerting the changes without stopping the streams and 
diagnosis of the post-change distribution upon stopping, which is also known as the sequential change diagnosis   \cite{ma2020two,dayanik2008bayesian}. Second, in many applications, the post-change distribution of data is challenging to obtain. Also, it is sometimes difficult to specify a prior distribution for the change points. In these cases, it is desirable to formulate the problem in a non-Bayesian decision theory framework, and develop a flexible parallel sequential change detection method that is robust for unknown post-change distributions under this framework. Third, we assume that the change points are independent for different data streams. For some applications, it is reasonable to extend the methods to the case where the change points are dependent. For example, the change points may be driven by the same event \cite{xie2013sequential} or propagated by each other \cite{zhang2019online}.

\newpage
\appendix
\begin{center}
	\Large\bf Appendix
\end{center}
\section{An additional simulation study}
In this simulation study, we let
$p_{k,t}(x) = (2\pi)^{-1/2}e^{-\frac{x^2}{2}}$
and $q_{k,t}(x) = (2\pi)^{-1/2}e^{-\frac{(x-1)^2}{2}}$
for all $k$ and $t$ and $\pi_{\infty}=0.2$, and $\pi _t = 
(t+2)!/(2\, t!) \cdot 0.8 \cdot (0.1)^3 (0.9)^t$ for $t\geq0$. 
In addition, we set $K = 100$ and consider the risk process $R_t=\text{LFNR}_t$ (see Example~\ref{eg:LFNR}) and the utility process $U_t=\text{IRL}_t$ (see  Example~\ref{eg:ARL}).

\begin{figure}[ht]
\begin{center}
\includegraphics[scale = 0.4]
{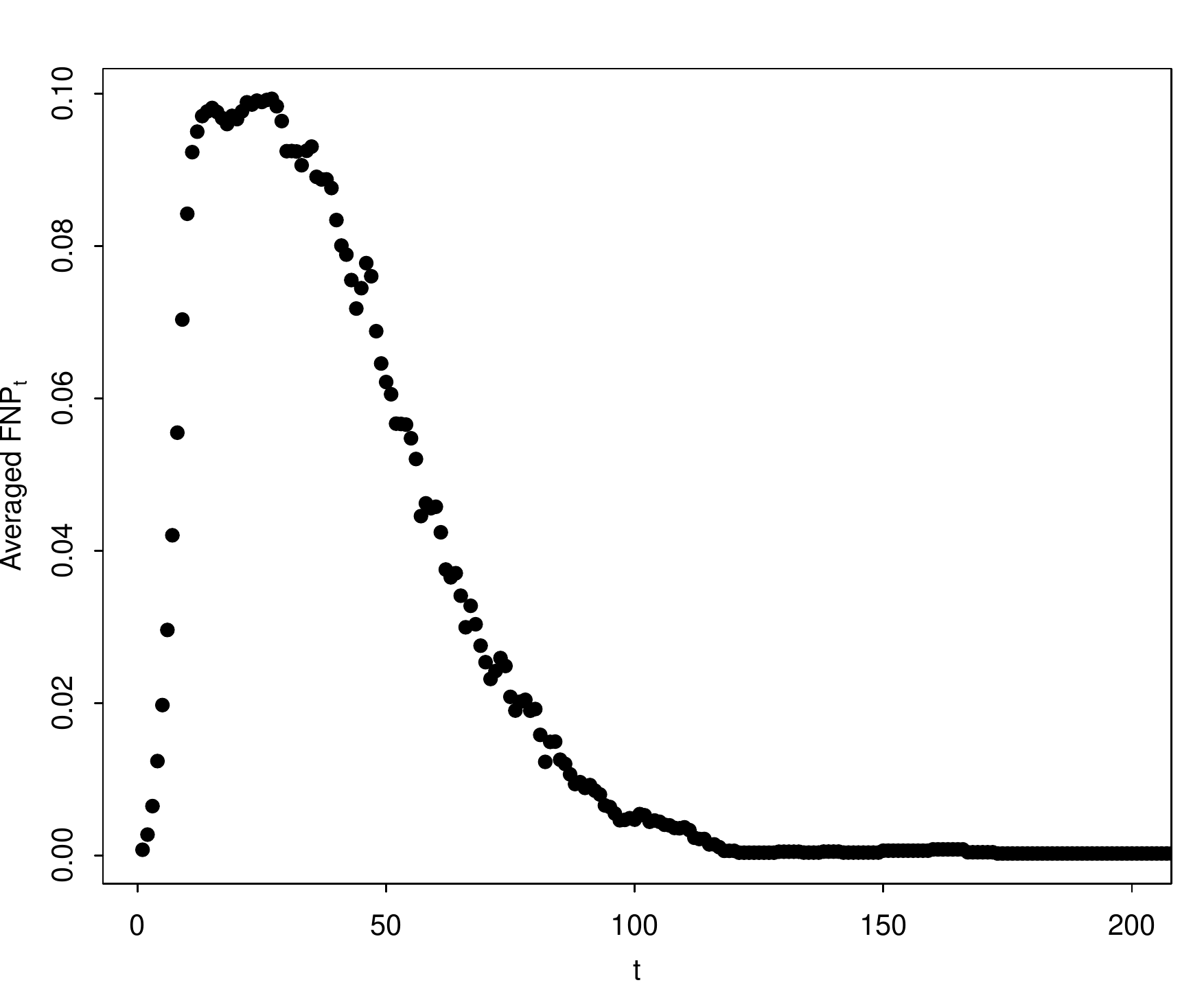}
\end{center}
\caption{$\text{FNP}_t$ averaged over a Monte Carlo simulation with 1000 replications at $K = 100$
}
\label{fig:sim1b_FNPt}
\end{figure}

\begin{figure}[ht]
\begin{center}
\includegraphics[scale = 0.4]
{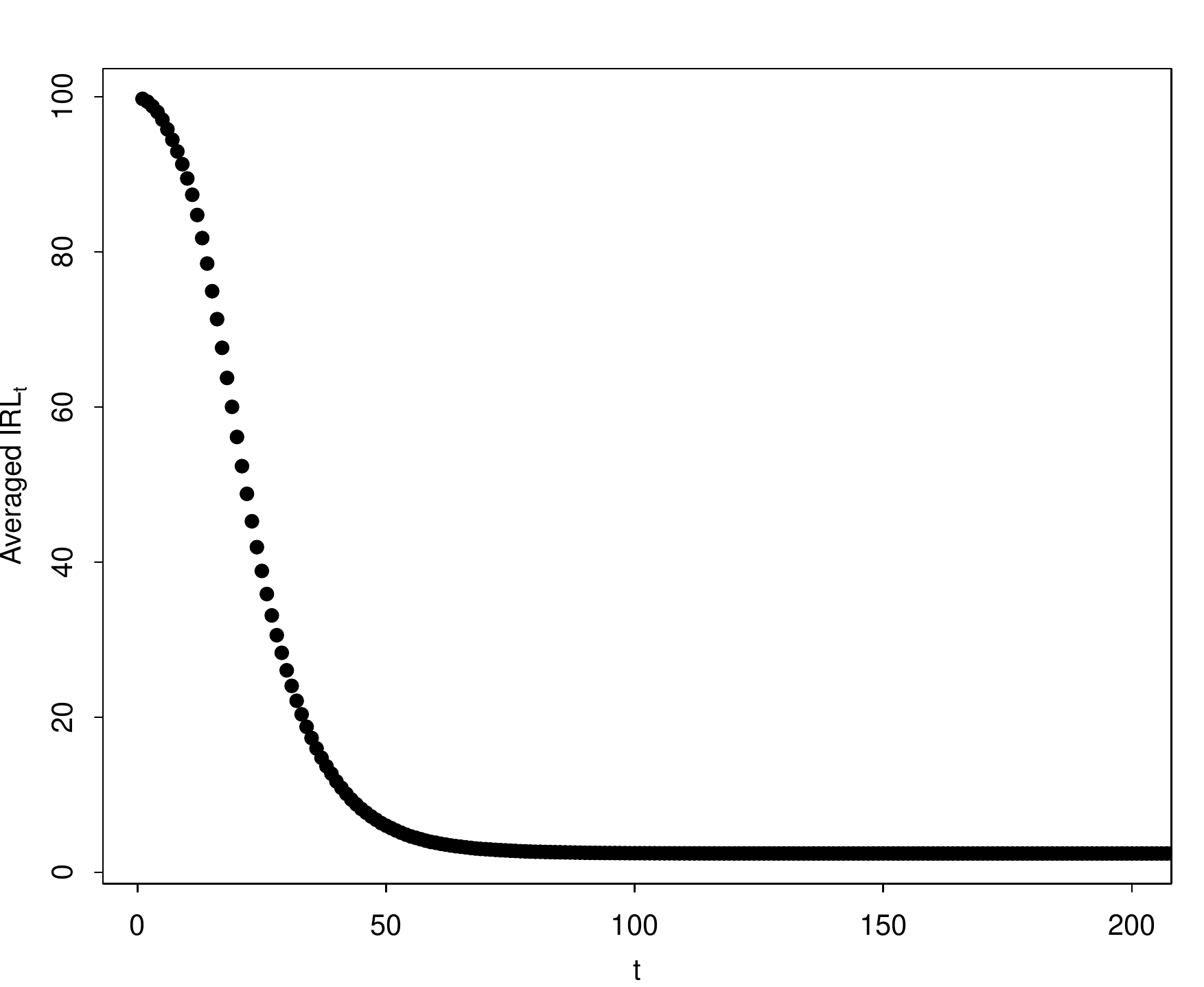}
\end{center}
\caption{$\text{IRL}_t$ averaged over 1000 Monte Carlo simulations at $K = 100$%
}
\label{fig:sim1b_IRLt}
\end{figure}

We plot the averaged risk measure $\text{FNP}_t$ and the averaged utility measure  of the proposed method $\sD$ defined in Algorithm \ref{alg:dec2} in Figures \ref{fig:sim1b_FNPt} and \ref{fig:sim1b_IRLt} based on a Monte Carlo simulation with 1000 replications. From Figure \ref{fig:sim1b_FNPt} we can see that the averaged $\text{FNP}_t$ is below 0.1 with a peak at around $t = 27$, which is consistent with Proposition~\ref{prop:simp}. From Figure \ref{fig:sim1b_IRLt} we can see that $\text{IRL}_t$ is decreasing to $0$ as $t$ increases.

\section{Proof of Results in Section~\ref{sec:method}}
\begin{proof}[Proof of Proposition~\ref{prop:onestep}]
	According to the second and third lines of Algorithm \ref{alg:gen}, $S_{t+1}$ output by Algorithm \ref{alg:gen} belongs to the set $\{S: \gamma_{S}\leq \alpha \text{ and }  S\subset S_t \}$. Thus, 
	$
	R_t=r_{t}(\{W_{k,t}\}_{k\in S_{t} }, S_{t}, S_{t+1})=\gamma_{S_{t+1}} \leq \alpha.
	$
\end{proof}

\begin{proof}[Proof of Proposition~\ref{prop:pD}]
	For each $t$, $S_{t+1}$ is obtained through Algorithm \ref{alg:gen}. Thus, $R_t(\pD) \leq \alpha$ a.s. for all $t\in \mathbb{Z}_+$, according to Proposition \ref{prop:onestep}. This implies $\pD\in\mathcal{D}_{\alpha}$.
\end{proof}

\begin{proof}[Proof of Proposition~\ref{prop:simp}]
	Under Assumption~\ref{assump:non-empty}, $S_{t+1}$ is obtained by Algorithm \ref{alg:simp} satisfies
	$R_t=r_{t}(\{W_{k,t}\}_{k\in S_{t} }, S_{t}, S_{t+1}) %
	$. According to the third and fourth lines of Algorithm \ref{alg:simp}, it satisfies
	$R_t= r_t ( \{W_{k,t}\}_{k\in S_{t}},S_{t}, S_{t+1} ) \leq \alpha$. Thus, $\sD\in\mathcal{D}_{\alpha}$.
\end{proof}

\section{Proof of Results in Section~\ref{sec:one-step-optimality}
}\label{sec:proof-thm}

\begin{proof}[Proof of Theorem \ref{thm:one-step optimality for general risk}]
First, we know that $\pD\in\mathcal{D}_{\alpha}$ according to Proposition \ref{prop:pD}.

Next, we compare the proposed sequential decision $\pD = (d^*_1,d^*_2,\cdots,d^*_t,\cdots)$ with an arbitrary sequential decision $\delta=(d_1,d_2,\cdots,d_t,\cdots)\in\mathcal{D}_{\alpha}$ satisfying $d_s= d^*_s$, for $s = 1, \cdots, t-1$. 
Let $\{S_{t}^*\}_{t\in\mathbb{Z}_+}$ be the index set of active streams following $\pD$ at all time points, and $S_{t+1}$ be the set selected by $\delta$ at time $t+1$. Note that both $\pD$ and $\delta$ select $S_1^*,\cdots, S_{t}^*$ as the index sets at time $1,\cdots, t$, according to the assumption that $d_s=d_s^*$ for $s=1,\cdots,t-1$.

According to the second and third line of Algorithm \ref{alg:gen}, $S_{t+1}^*$ satisfies
$
u_{t}(\{W_{k,t}\}_{k\in S^*_{t} }, S^*_{t}, S^*_{t+1}) 
= \max_{S\subset S^*_t} u_{t}(\{W_{k,t}\}_{k\in S^*_{t} }, S^*_{t}, S)
$
subject to $\gamma_{S}\leq \alpha$. Because $\delta\in\mathcal{D}_{\alpha}$, and the index set selected by $\delta$ and $\pD$ at time $t$ are both $S_t^*$, we have $R_{t}(\delta)=\gamma_{S_{t+1}}\leq\alpha$ a.s. This further implies
$
u_{t}(\{W_{k,t}\}_{k\in S^*_{t} }, S^*_{t}, S^*_{t+1}) 
\geq  u_{t}(\{W_{k,t}\}_{k\in S^*_{t} }, S^*_{t}, S_{t+1}).
$
That is, $U_t(\pD)\geq U_t(\delta)$. 
The proof is completed by taking expectation on both sides.
\end{proof}
\begin{proof}[Proof of Theorem \ref{thm:one-step optimality for monotone risk}]
First, according to Theorem~\ref{thm:one-step optimality for general risk}, the sequential decision $\pD$ obtained from Algorithm~\ref{alg:dec} is locally optimal. Thus, to show $\sD$ obtained from Algorithm~\ref{alg:dec2} is also locally optimal, it suffices to show that $\sD$ is a special case of $\pD$ under Assumption~\ref{assump:one-step optimal}.
Let $\{S_{t}^*\}_{t\in\mathbb{Z}_+}$ be the index set of active streams following the decision $\sD$ at every time. For each $t$, it is sufficient to show that $S_{t+1}^*=\{k_1,\cdots, k_{n^*}\}$ obtained in the forth line of Algorithm~\ref{alg:simp} also solves the optimization problem in the third line of Algorithm~\ref{alg:gen}. To see this, it is sufficient to show
\begin{equation}\label{eq:to-show}
    u_t(\{W_{k,t}\}_{k\in S^*_{t} }, S^*_{t}, S_{t+1})\leq  u_t(\{W_{k,t}\}_{k\in S^*_{t} }, S^*_{t}, S^*_{t+1})
\end{equation}
for any set $S_{t+1}\subset S_t$ satisfying $\gamma_{S_{t+1}}\leq\alpha$.

Recall that $k_1,\cdots, k_{|S_t^*|}$ are chosen so that $W_{k_1,t}\leq\cdots\leq W_{k_{|S_t^*|},t}$. We discuss two cases: $S_{t+1}=\{k_1,\cdots, k_n\}$ for some $n\in\{0,\cdots, |S_t^*|\}$; and $S_{t+1}\neq\{k_1,\cdots, k_n\}$ for all $n$. In what follows, we show that \eqref{eq:to-show} holds for both cases.

{\em Case 1:} $S_{t+1}=\{k_1,\cdots, k_n\}$ for some $n\in\{0,\cdots, |S_t^*|\}$.

In this case, $\gamma_n=r_t(\{W_{k,t}\}_{k \in S^*_{t}},S^*_{t}, \{k_1,\cdots, k_n\})\leq \alpha$. Recall that $n^*$ obtained in the forth line of Algorithm~\ref{alg:simp} satisfies $\mu_{n^*}=\max_{n':\gamma_{n'}\leq \alpha} \mu_{n'}$. Thus, $\mu_{n^*}\geq \mu_{n}.$ 
We complete the proof of \eqref{eq:to-show} in Case 1 by recalling  that $\mu_{n^*}=u_{t}(\{W_{k,t}\}_{k \in S^*_{t}},S^*_{t}, S^*_{t+1})$ and $\mu_n= u_{t}(\{W_{k,t}\}_{k \in S^*_{t}},S^*_{t}, S_{t+1})$.

{\em Case 2: $S_{t+1}\neq\{k_1,\cdots, k_n\}$ for all $n$.}

Let $n=|S_{t+1}|$. In this case, there exists $h\in\{0,\cdots, n\}$ and
$2\leq j_1<j_2<\cdots<j_{n-h}$ such that
$S_{t+1}=\{k_{1},\cdots, k_{h}, k_{h+j_1}, k_{h+j_2},\cdots, k_{h+j_{n-h}}\}$ (we set $k_0=0$ for the ease of presentation).

Note that $W_{k_{h+j_1}}\geq W_{k_{h+1}}$, $k_{h+j_1}\in S_{t+1}$, and $k_{h+1}\notin S_{t+1}$. According to Assumption~\ref{assump:one-step optimal}, we have
\begin{equation}\label{eq:r_t-bound-1}
\begin{split}
   & r_t(\{W_{k,t}\}_{k\in S^*_{t} }, S^*_{t}, S_{t+1})\\
    \geq & 
    r_{t}(\{W_{k,t}\}_{k\in S^*_{t} }, S^*_{t}, S_{t+1}\setminus\{k_{h+j_1}\}\cup\{k_{h+1}\})\\
     = &
     r_{t}(\{W_{k,t}\}_{k\in S^*_{t} }, S^*_{t}, \{k_{1},\cdots, k_{h+1}, k_{h+1+j_2}, \cdots, k_{h+j_{n-h}}\}).
\end{split}
\end{equation}
That is, the risk evaluated at $\{k_{1},\cdots, k_{h+1}, k_{h+j_2}, \cdots, k_{h+j_{n-h}}\}$ is no greater than at $\{k_{1},\cdots, k_{h}, k_{h+j_1}, \cdots, k_{h+j_{n-h}}\}$. With similar arguments, we have
\begin{equation}\label{eq:r_t-bound-2}
\begin{split}
     & r_{t}(\{W_{k,t}\}_{k\in S^*_{t} }, S^*_{t}, \{k_{1},\cdots, k_{h+1}, k_{h+j_2}, \cdots, k_{h+j_{n-h}}\})\\
     \geq & \cdots\\
     \geq & r_{t}(\{W_{k,t}\}_{k\in S^*_{t} }, S^*_{t}, \{k_{1},\cdots, k_{n}\}).
\end{split}
\end{equation}
Combining \eqref{eq:r_t-bound-1} and \eqref{eq:r_t-bound-2}, we obtain
\begin{equation}\label{eq:r_t-bound-3}
    \begin{split}
   r_t(\{W_{k,t}\}_{k\in S^*_{t} }, S^*_{t}, S_{t+1})%
     \geq r_{t}(\{W_{k,t}\}_{k\in S^*_{t} }, S^*_{t}, \{k_{1},\cdots, k_{n}\}).
\end{split}
\end{equation}
Recall that $r_t(\{W_{k,t}\}_{k\in S^*_{t} }, S^*_{t}, S_{t+1})\leq \alpha$. Thus, the above inequality also implies
\begin{equation}\label{eq:r_t-bound-4}
    \begin{split}
\gamma_n= r_{t}(\{W_{k,t}\}_{k\in S^*_{t} }, S^*_{t}, \{k_{1},\cdots, k_{n}\})\leq\alpha.
\end{split}
\end{equation}
Next, we consider $u_t(\{W_{k,t}\}_{k\in S^*_{t} }, S^*_{t}, S_{t+1})$. By replacing $r_t$ with $u_t$ and `$\geq$' with `$\leq$' in \eqref{eq:r_t-bound-1} -- \eqref{eq:r_t-bound-3}, we obtain
$     u_t(\{W_{k,t}\}_{k\in S^*_{t} }, S^*_{t}, S_{t+1}) \leq  u_{t}(\{W_{k,t}\}_{k\in S^*_{t} }, S^*_{t}, \{k_{1},\cdots, k_{n}\})
$
under Assumption~\ref{assump:one-step optimal}. This implies
$     u_t(\{W_{k,t}\}_{k\in S^*_{t} }, S^*_{t}, S_{t+1}) \leq \mu_n.$
According to equation \eqref{eq:r_t-bound-4}  and the proof in Case 1, we have 
$\mu_n\leq \mu_{n^*}=u_t(\{W_{k,t}\}_{k\in S^*_{t} }, S^*_{t}, S^*_{t+1}).$
Combining these inequalities, we obtain \eqref{eq:to-show}.
\end{proof}

\begin{proof}[Proof of Corollary \ref{corollary:eg:one-step optimal-monotone risk}]
We start with the proof of the first statement of the corollary.  According to \eqref{eq:LFWER} -- \eqref{eq:IADD}, we have
 $\text{LFWER}_t=\text{GLFWER}_t=\text{LFNR}_t=\text{IARL}_t=\text{IADD}_t =0$ if $S_{t+1}=\emptyset$, and $\text{LFDR}_t=0$ if $S_{t+1}=S_{t}$. In addition, according to the definition, $\text{LFWER}_t, \text{GLFWER}_t, \text{LFNR}_t$, and $\text{LFDR}_t$ fall into the interval $[0,1]$, and $\text{IADD}_t$ and $\text{IARL}_t$ belong to $[0,K]$. Thus, Assumption \ref{assump:non-empty} is verified for $\alpha>0$. 
 
Next, we show that if $R_t\in\{ \text{LFWER}_t, \text{GLFWER}_t, \text{LFNR}_t,\text{IADD}_t \}$ and $U_t\in \{ \text{IARL}_t, -\text{LFDR}_t \}$ then Assumption \ref{assump:one-step optimal} is satisfied. Without loss of generality, we assume $S_0 = [n]$, where $1 \leq n \leq K$. Also, assume $\ww=(w_1,\cdots,w_n)\in [0,1]^n$, $S\subset [n]$, $i\in S, j\in [n]\setminus S$, and $w_i \geq w_j$. We consider different choices of $R_t$ and $U_t$ below.

If $R_t = \text{GLFWER}_t$, we apply the next lemma, whose proof is postponed to Appendix \ref{sec:proof-uniform-optimality} on page \pageref{proof-lemma:GLFWER}.
\begin{lemma} \label{lemma:GLFWER}
For $\uu  = (u_1, \ldots, u_p), \vv = (v_1, \ldots, v_p) \in [0,1]^p$, then
$$1-\sum_{l = 0}^{m-1} \sum_{\substack{I \subset \langle p \rangle \\ |I| = l}} \big(\prod_{q \in I} u_q \big) \prod_{k\in \langle p \rangle \setminus I}(1-u_k) \leq
1-\sum_{l = 0}^{m-1} \sum_{\substack{I \subset \langle p \rangle \\ |I| = l}} \big(\prod_{q \in I} v_q \big) \prod_{k\in \langle p \rangle \setminus I}(1-v_k)
$$ if $0 \leq u_k\leq v_k \leq 1$ for all $k \in \langle p \rangle$.
\end{lemma}
We apply the above lemma with $\uu = (w_k)_{k \in S} = (u_1, \ldots, u_{|S|})$, and $\vv = (w_k)_{k \in (S\setminus\{i\})\cup \{j\}} = (v_1, \ldots, v_{|S|})$. %
Since
$
0\leq u_r \leq v_r \leq 1\text{ for all }r\in [|S|],
$
we obtain $$1-\sum_{l = 0}^{m-1} \sum_{\substack{I \subset [|S|] \\ |I| = l}} \big(\prod_{q \in I} u_q \big) \prod_{k\in [|S|] \setminus I}(1-u_k) 
\leq
 1-\sum_{l = 0}^{m-1} \sum_{\substack{I \subset [|S|]  |I| = l}} \big(\prod_{q \in I} v_q \big) \prod_{k\in [|S|] \setminus I}(1-v_k).
$$
It follows that
\sloppy $r_t(\ww, [n], S)
=  1-\sum_{l = 0}^{m-1} \sum_{\substack{I \subset S \\ |I| = l}} \big(\prod_{q \in I} w_q \big) \prod_{k\in S \setminus I}(1-w_k)
\geq 
1-\sum_{l = 0}^{m-1} \sum_{\substack{I \subset (S\setminus\{i\}) \cup \{j\} \\ |I| = l}} \big(\prod_{q \in I} w_q \big) \prod_{k\in (S\setminus\{i\}) \cup \{j\} \setminus I}(1-w_k)
=  r_t(\ww, [n], (S\setminus\{i\}) \cup \{j\}).$

If  $R_t = \text{LFNR}_t$, we have
$       r_t(\ww, [n], (S\setminus\{i\})\cup \{j\})
        =  \frac{\sum_{k \in (S\setminus\{i\})\cup \{j\} } w_k}{|S| \vee 1}
        =  \frac{\sum_{k \in S\ } w_k + w_i- w_j}{|S| \vee 1}
        \leq  \frac{\sum_{k \in S\ } w_k}{|S| \vee 1}
        = r_t(\ww, [n], S).
$

If $R_t = \text{IADD}_t$, we have $r_t(\ww, [n], S) = \sum_{k\in S} w_k = \sum_{k\in (S\setminus\{i\})\cup \{j\}} w_k + w_i - w_j 
\geq \sum_{k\in (S\setminus\{i\})\cup \{j\}} w_k=r_t(\ww, [n], (S\setminus\{i\})\cup \{j\}).
$

If $U_t = \text{IARL}_t$, then
$u_t(\ww, [n], S) = \sum_{k\in S} \{1-g(w_k)\} = \sum_{k\in (S\setminus\{i\})\cup \{j\}} \{1-g(w_k)\} -g( w_i) + g(w_j)  \leq \sum_{k\in (S\setminus\{i\})\cup \{j\} } \{1-g(w_k)\} = u_t(\ww, [n], (S\setminus\{i\})\cup \{j\}),
$
where we used the fact that $g(w_i) \geq g(w_j)$.

If $U_t = - \text{LFDR}_t$, then
$u_t(\ww, [n], S) =  - (|[n] \setminus S| \vee 1)^{-1}\sum_{k \in [n] \setminus S}  (1- w_k)  
=  - (|[n] \setminus S| \vee 1)^{-1}\sum_{k \in [n] \setminus (S\setminus\{i\}\cup \{j\})}  (1- w_k) +w_i - w_j\leq  -(|[n] \setminus S| \vee 1)^{-1}\sum_{k \in [n] \setminus (S\setminus\{i\}\cup \{j\}) }  (1- w_k)
- = u_t(\ww, [n], (S\setminus\{i\})\cup \{j\})$.

Since $\text{LFWER}_t$ is a special case of $\text{GLFWER}_t$ where $m = 1$ in Example \ref{eg:GLFWER}, if $R_t= \text{LFWER}_t$, we also have $r_t(\ww, [n], S)\leq r_t(\ww, [n], (S\setminus\{i\})\cup \{j\})$.
The proof for the other set of $R_t$ and $U_t$ can be obtained similarly by flipping their signs. We omit the repetitive details.
\end{proof}

\section{Proof of Results in Section~\ref{sec:uniform-optimality}} \label{sec:proof-uniform-optimality}
\subsection{Proof sketch for  Theorem~\ref{thm:uniform optimality}}
The proof of Theorem~\ref{thm:uniform optimality} is involved. In this section, we provide a high level summary of steps in proving Theorem~\ref{thm:uniform optimality} and an overview of the supporting lemmas in Appendix~\ref{subsec:partial}--\ref{subsec:coupling}. We will wrap up these supporting results and provide the proof of Theorem~\ref{thm:uniform optimality} in Appendix \ref{subsec:proof-of-thm-3}.

In Appendix~\ref{subsec:partial}, we define a special partial order relationship `$\leqc$' over the space $\So=\cup_{k=1}^{K}\{(v_{1}, \cdots, v_{k}) : 0 \leq v_{1} \leq \cdots v_{k} \leq 1\} \cup\{\varnothing\}$  so that one can compare vectors in $\So$ even when they have different dimensions. We also study monotone functions in terms of this special partial order relation. It turns out that the concepts of `entrywise monotonicity' and `appending monotonicity' (defined in Definitions~\ref{def:entrywise-increasing} and \ref{def:appending-increasing}) of functions are closely related to their monotonicity in terms of the partial order `$\leqc$'. In particular, we show that the utility function $\tilde{u}_t(\cdot)$ is a decreasing function over $\So$. See Lemma~\ref{lemma:decreasing} for more details.

Let $\delta^*$ denotes the proposed method. Heuristically, if we could argue that $[W^{\delta^*}_{S^{\delta^*}_t,t}]\leqc [W^{\delta}_{S^{\delta}_t,t}]$ for all decision $\delta\in\mathcal{D}_{\alpha}$ all $t$, then Theorem~\ref{thm:uniform optimality} is proved by combining this with the assumption $U_t(\delta)=\tilde{u}_t([W^{\delta}_{S^{\delta}_{t+1},t}])$ and that $\tilde{u}_t(\cdot)$ is decreasing. However, this statement does not hold almost surely for the stochastic processes $[W^{\delta^*}_{S^{\delta^*}_t,t}]$ and $ [W^{\delta}_{S^{\delta}_t,t}]$. Instead, we  show this stochastic version of this statement using concepts such as {\em stochastic ordering}. That is  $\EE(g([W^{\delta^*}_{S^{\delta^*}_t,t}])\leq \EE(g([W^{\delta}_{S^{\delta}_t,t}]))$ for any increasing functions over $\So$. The main analysis is carried out through induction using supporting lemma developed in Appendix~\ref{subsec:property-of-algos} and Appendix~\ref{subsec:coupling}.

In particular, in Appendix~\ref{subsec:property-of-algos}, we study the monotonicity (in terms of the partial order `$\leqc$') of the proposed one-step selection rule. We show that the proposed decision induces a monotone mapping over $\So$ (Lemma \ref{lemma:increasing rule}). 
We also show that the order statistic of the posterior probability of the remaining stream will become larger by following decisions other than the proposed one (Lemma~\ref{lemma:one-step optimal}).
Roughly, these result suggests that proposed method tends to make $[W_{S_t,t}]$ `smaller' at the `current time', when compared with other methods.
In Appendix~\ref{subsec:coupling}, we show several stochastic ordering results regarding the process $[W_{S_t,t}]$ following different decisions. 
In particular,  Lemma~\ref{prop:coupling} states that $[W_{S_{t+s},t+s}]$ is stochastically increasing in $[W_{S_t,t}]$ following the proposed method. Lemma~\ref{lemma:one coupling} states that $[W_{S_{t+1},t+1}]$ becomes `stochastically larger' if we follow another method that also controls the risk for one step, when compared with the proposed method.

We note that the proof of the results in Appendix~\ref{subsec:coupling} follows similar ideas of that in \cite{chen2019compound} with the following main differences. First,
  \cite{chen2019compound} only allows geometric priors and time-homogeneous pre-/post- change distributions, which leads to homogeneous Markov chains $\{W_{k,t}\}_{t\geq 0}$ for different $k$. Under the current settings, we allow a general class of prior distributions and time-heterogeneous of pre/post-change distributions. As a result, $\{W_{k,t}\}_{t\geq 0}$  are still Markov chains for different $k$, but their transition kernels will be time-heterogeneous. Second, \cite{chen2019compound} only considers the risk measure LFNR while we need to take care of a range of more complicated risk and utility functions. To address the additional challenges, we leverage results in Appendix~\ref{subsec:partial} and \ref{subsec:property-of-algos} to perform detailed analysis on the risk and utility processes.

\subsection{Partial order spaces and monotone mapping}\label{subsec:partial}
In this section, we first introduce the classic definition of partial order spaces, then define a partial order relation over the space $\So$. After that, we provide several useful supporting lemmas connecting the proposed sequential decision with processes over $\So$.

\begin{definition}
A partially ordered space (or pospace) $(\mathcal{S}, \leqc)$ is a topological space $\mathcal{S}$ with a closed partial order $\leqc$. That is, `$\leqc$' satisfies 1) $\uu \leqc \uu$ for all $\uu \in \So$; 2) $\uu \leqc \vv$ and $\vv \leqc \uu$ implies that $\uu = \vv$; 3) $\uu \leqc \vv$ and $\vv \leqc \ww$ imply that $\uu \leqc \ww$, and the set $\{(\uu, \vv) \in \So^2:\uu \leqc \vv\}$ is a closed set.
\end{definition}
Recall $\So$ is defined  as
\begin{align*}
&\bigcup_{k=1}^{K}\{(v_{1}, \cdots, v_{k}) : 0 \leq v_{1} \leq \cdots v_{k} \leq 1\} \cup\{\varnothing\},
\end{align*}
where $\varnothing$ denotes the vector with zero dimension. The elements in $\So$ are order statistics of elements in the following space $\Su$.
$$
\Su = \bigcup_{k=1}^{K}[0,1]^{k} \cup\{\varnothing\}.
$$
It is easy to verify that for $\uu \in \Su$, $[\uu] \in \So$, and $\big[\uu\big] = \uu$ for any $\uu \in \So$.

We define a partial order relation over $\So$. Let the function $\dim (\cdot)$ denote the dimension of a vector.
\begin{definition}
For $\uu=(u_1,\cdots, u_{\dim(\uu)}), \vv=(v_1,\cdots,v_{\dim(\vv)}) \in \So$, $\uu \leqc \vv$ if $\dim(\uu) \geq \dim(\vv)$ and $u_{i} \leq v_{i}$ for $i=1, \ldots, \dim(\vv).$ In addition,  $\uu \leqc \varnothing$ for all $\uu \in \So$.  
\end{definition}
The next lemma states that $(\So,\leqc)$ is a polished  partial order space.
\begin{lemma}[Lemma F.1 in \cite{chen2019compound}]\label{lemma:pospace}
$(\So, \leqc )$ is a  partially ordered polish space  equipped with a closed partial order $\leqc$.
\end{lemma}
Next, we present the definition for monotone functions mapping from a partial order space to another one.
\begin{definition}
Let $(\mathcal{S}_1,\leqc_{\mathcal{S}_1})$ and $(\mathcal{S}_2,\leqc_{\mathcal{S}_2})$ be two partially ordered polish spaces.
For a function $f: \mathcal{S}_1 \rightarrow \mathcal{S}_2$, $f$ is said to be increasing if $f(\uu) \leqc_{\mathcal{S}_2} f(\vv)$ for all $\uu, \vv \in \mathcal{S}_1$ satisfying $\uu \leqc_{\mathcal{S}_1} \vv$. In addition, a function $f$ is said to be decreasing if $-f$ is increasing.
\end{definition}
In particular,  a function $f:\So\to \mathbb{R}$ is said to be increasing, if $f(\uu)\leq f(\vv)$ for all $\uu\leqc \vv$, where `$\leq$' refers to the typical `smaller or equal' relation over real numbers; a function $g:\So\to\So$ is said to be increasing, if $g(\uu)\leqc g(\vv)$ for all $\uu\leqc\vv$.

The next lemma presents the connection between monotone functions with respect to the partial order $\leqc$ and its entrywise and appending monotonicity.
\begin{lemma} \label{lemma:decreasing}
If a function $f: \So \rightarrow \RR$ is entrywise decreasing and appending increasing, then $f$ is decreasing with respect to the partial order relation `$\leqc$'. In particular, the utility function $\tilde{u}_t(\cdot)$ is a decreasing function over $\So$ in terms of `$\leqc$'. 
\end{lemma}
\begin{proof}
For $\uu \leqc \vv$ with $\uu, \vv \in \So$, there are two cases: 1) $\dim (\uu) = \dim(\vv)$ and $u_i \leq v_i$ for $i = 1, \ldots, \dim(\uu)$; and 2)  $\dim (\uu) \geq \dim(\vv)$ and $u_i \leq v_i$ for $i = 1, \ldots, \dim(\vv)$. For the first case, $f(\uu) \geq f(\vv)$ because $f$ is entrywise decreasing. For the second case, $f(u_1, \ldots, u_{\dim(\vv)}) \geq f(\vv)$ because $f$ in entrywise increasing. In addition, since $f$ is appending decreasing, $f(\uu) \geq f(u_1, \ldots, u_{\dim(\vv)})$. Combining these two inequalities, we arrive at $f(\uu) \geq f(\vv)$.
\end{proof}

\subsection{Property of the one-step selection rule in Algorithm~\ref{alg:simp}} \label{subsec:property-of-algos}
Recall that $R_t=\rtd([W_{S_{t+1},t}])$ in Assumption~\ref{assump:risk}. We define two related maps $\Io: \So \rightarrow\{0, \cdots, K\}$ and $\Ho: \So \rightarrow \So$ below.
For any $\uu = \big(u_{1}, \ldots, u_{\dim(\uu)}\big) \in \So$, if $\dim(\uu)=0$, we define  $\Io (\uu) = 0$, otherwise we define $\Io (\uu)$ as
\begin{equation}\label{eq:io}
\Io (\uu)
    =\sup \{n \in\{0, \ldots, \dim (\uu)\} : \rtd (\{u_{i}\}_{i=1}^n) \leq \alpha \},   
\end{equation}
and
\begin{equation}
\Ho(\uu)=
\begin{cases}
\big(u_{1}, \cdots, u_{\Io (\uu)}\big) & \text { if } \Io (\uu) \geq 1 \\
\varnothing & \text { otherwise }.
\end{cases} 
\end{equation}

$\Io$ in \eqref{eq:io} is well-defined, thanks to the following lemma.
\begin{lemma} \label{lemma:empty-set}
Under Assumptions \ref{assump:non-empty} and \ref{assump:risk}, $\rtd(\varnothing)  \leq \alpha$. %
\end{lemma}
\begin{proof}
Under Assumption \ref{assump:risk}, $r_{t}(\{W_{k,t}\}_{k\in S_{t} }, S_{t}, \varnothing) = \rtd (\varnothing)$ and $r_{t}(\{W_{k,t}\}_{k\in S_{t} }, S_{t}, S_t) = \rtd \big([\{W_{k,t}\}_{k \in S_{t}}]\big)$. Again by Assumption \ref{assump:risk}, $\rtd$ is appending increasing, thus, $\rtd (\varnothing) \leq \rtd \big([\{W_{k,t}\}_{k \in S_{t}}]\big)$. It follows that $\min_{S\in\{\emptyset,S_t\}}r_{t}(\{W_{k,t}\}_{k\in S_{t} }, S_{t}, S) = \rtd (\varnothing)$. Hence, Assumption \ref{assump:non-empty} implies that $\rtd (\varnothing)\leq \alpha$.
\end{proof}

Next, we show the connection between the maps $\Io$ and $\Ho$ and Algorithm~\ref{alg:simp}.
\begin{lemma} \label{lemma:def-I&J}
Under Assumptions~\ref{assump:non-empty}, \ref{assump:risk}, and \ref{assump:utility},
if we input $W_{S_t,t} = \uu$ and an arbitrary index set $S_t$ satisfying $|S_t|=\dim(\uu)$ in Algorithm \ref{alg:simp}, then the selected $S_{t+1}$ 
satisfies 
\begin{equation} \label{eq:def-I&J}
|S_{t+1}| = \Io ([\uu]) \text{ and } [W_{S_{t+1}, t}] = \Ho ([\uu])    
\end{equation}
\end{lemma}
\begin{proof}
We first show that $\mu_n$ defined in line 3--4 of Algorithm~\ref{alg:simp} is non-decreasing in $n$.  Recall that, $\mu_n = \ud_t \big([W_{k,t}]_{k \in \{k_i\}_{i = 1}^{n}}\big) = \ud_t \big([\uu]_1, \ldots, [\uu]_n \big)$ for $n = 1, \ldots, |S_t|$, and $\mu_0 = \ud_t (\varnothing)$, where $[\uu]_k$ denotes the $k$th element of the order statistics $[\uu]$. According to Assumption~\ref{assump:utility}, $\ud_t$ is appending increasing. Thus, $
\mu_0 \leq \mu_1 \leq \ldots \leq \mu_{|S_t|}.
$

Next, we show that $|S_{t+1}| = \Io ([\uu])$.
Recall that $S_{t+1}=\big\{k_{1}, \ldots, k_{n^*}\big\}$ if $n^* \geq 1$ and $S_{t+1}=\varnothing$ if $n^*=0$ from the line 5 in Algorithm~\ref{alg:simp}, where $k_1,
\ldots, k_{n^*}$ satisfy $W_{k_{1}, t} \leq W_{k_{2}, t} \leq \cdots \leq W_{k_{n^*},t}$ and are obtained in the line 2 of Algorithm~\ref{alg:simp}. 

Let $n^{**} = \Io ([\uu])$.
Comparing $S_{t+1}$ output by Algorithm~\ref{alg:simp} and the definition of $\Io$, we have
$n^* \leq n^{**}$. On the other hand, we have $\mu_{n^*} \geq \mu_{n^{**}}$ according to lines 3--4 of Algorithm~\ref{alg:simp}. Because $\mu_n$ is non-decreasing in $n$, this implies $n^*\geq n^{**}$ or ($\mu_{n^*} = \mu_{n^{**}}$ and $n^*< n^{**}$). Note that the latter case  ($\mu_{n^*} = \mu_{n^{**}}$ and $n^*< n^{**}$) is not possible due to footnote \ref{footnote:largest-n}. Thus, $n^*\geq n^{**}$ and, consequently, $|S_{t+1}| = \Io ([\uu])$.

The second equation in \eqref{eq:def-I&J} holds due to the definition of $\Ho$ and lines 2 and 5 in Algorithm~\ref{alg:simp}.
\end{proof}

The following lemma compares the posterior probability associated with the proposed sequential decision rule and that of another sequential rule so that the risk process is controlled at the desired level.

\begin{lemma}\label{lemma:one-step optimal}
Let $\uu =(u_1, \ldots, u_m)\in \Su$. 
Under Assumptions~\ref{assump:non-empty}, \ref{assump:risk}, and \ref{assump:utility},
if $\{k_1, \ldots, k_l \} \subset \{1, \ldots, m \}$ satisfies $\rtd ([\uu^{\prime}]) \leq \alpha$ for $\uu^{\prime} = (u_{k_1}, \ldots, u_{k_l})$ and $l\geq 1$,
then
$
\Ho ([\uu]) \leqc [\uu^{\prime}].
$
In addition, if $\Ho ([\uu]) = \varnothing$, then no nonempty $\uu^{\prime} = (u_i)_{i \in A} \in \Su$ with $A \subset \{1, \ldots, m \} $
such that $\rtd ([\uu^{\prime}]) \leq \alpha$.
\end{lemma}

\begin{proof}
We first show the ``In addition" part by contradiction. Suppose there exists $\uu^{\prime} = (u_{k_1}, \ldots, u_{k_l}) \in \Su$ ($l \geq 1$) such that $\{k_1, \ldots, k_l \} \subset \{1, \ldots, m \}$ and $\rtd([\uu^{\prime}])\leq \alpha$.
Because $[\uu]$ is the order statistic of $\uu$, $[\uu]_i\leq [\uup]_i$ for $i=1,\cdots, l$. 
Combining this with the assumption that $\rtd$ is entrywise increasing, we arrive at
$\rtd ([\uu]_1, \ldots, [\uu]_l) \leq \rtd ([\uu^{\prime}]) \leq \alpha.
$
This implies
$\Io([\uu]) \geq l$, which  contradicts with $\Ho ([\uu]) = \varnothing$.

For the first part of the lemma, 
for any $\uu^{\prime} = (u_{k_1}, \ldots, u_{k_l})$ with $l\geq 1$ which satisfies $\{k_1, \ldots, k_l \} \subset \{1, \ldots, m \}$ and $\rtd ([\uu^{\prime}]) \leq \alpha$,
to show $\Ho ([\uu]) \leqc [\uup]$, it suffices to show $\Io (\uu) \geq l$ and $[\uu]_{i} \leq u_{k_i}$ for $i=1, \ldots, l$. Similar to the previous arguments, since $[\uu]_i\leq [\uup]_i$ for $i=1,\cdots, l$, 
and the function $\rtd$ is entrywise increasing, $\rtd (([\uu]_1, \ldots, [\uu]_l)) \leq \rtd ([\uu^{\prime}]) \leq \alpha,
$ which implies $\Io([\uu]) \geq l$. 
\end{proof}

\begin{lemma} \label{lemma:increasing rule}
Under Assumptions~\ref{assump:non-empty}, \ref{assump:risk}, and \ref{assump:utility},
 the mapping $\Ho(\uu)$ is increasing in $\uu$ with respect to the partial order relation `$\leqc$'. That is, for any $\uu \leqc \vv \in \mathcal{S}_{\mathrm{o}}$, $\Ho(\uu) \leqc \Ho(\vv)$.
\end{lemma}
\begin{proof}
If $\vv = \varnothing$, then $\Ho(\vv) = \varnothing$. It follows that $\Ho(\uu) \leqc \varnothing = \Ho(\vv)$. We then assume $\vv \neq \varnothing$ in the rest of the proof. Assume $\vv = (v_1, \ldots, v_{\dim (\vv)} )$, $\dim (\vv) \geq 1$ and $\uu \leqc \vv$.

Recall that $\Ho(\uu) = (u_1,\cdots, u_{\Io(\uu)})$ and $\Ho(\vv)=(v_1,\cdots, v_{\Io(\vv)})$, it is sufficient to show $\Io(\uu) \geq \Io(\vv)$ and $u_i \leq v_i$ for $i = 1, \ldots, \Io(\vv)$. 

According to the definition of the partial order relation, $\dim(\uu)\geq \dim(\vv)$ and $u_i\geq v_i$ for $i=1,2,\cdots, \dim(\vv)$. Also note that $\Io(\vv)\leq\dim(\vv)$. This implies that 
$\Io(\vv)\leq\dim(\uu)$ and 
$u_i\leq v_i$ for $i= 1,2,\ldots, \Io(\vv)$. 

Next we show $\Io(\uu) \geq \Io(\vv)$ by contradiction. If on the contrary $\Io(\uu) < \Io(\vv)$, then $\Io(\uu) + 1 \leq \Io(\vv) \leq \dim(\vv)$. Since  $\rtd$ is entrywise increasing under Assumption~\ref{assump:risk}, we have
\begin{equation}\label{eq:ineq1}
\rtd ((u_1, \ldots, u_{\Io(\uu)+1})) \leq \rtd ((v_1, \ldots, v_{\Io(\uu)+1})).
\end{equation}
Because $\rtd$ is appending increasing and $\Io(\uu)+1\leq \Io(\vv)$, 
\begin{equation}\label{eq:ineq2}
\rtd \big((v_1, \ldots, v_{\Io(\uu)+1})\big) \leq 
\rtd \big((v_1, \ldots, v_{\Io(\vv)})\big).
\end{equation}
According to the definition of $\Io(\vv)$, we have 
$ \rtd \big((v_1, \ldots, v_{\Io(\vv)})\big)\leq \alpha.$
Combining this with \eqref{eq:ineq1} and \eqref{eq:ineq2}, we obtain $\rtd \big((u_1, \ldots, u_{\Io(\uu)+1})\big)\leq\alpha$. This contradicts with the definition of $\Io(\uu)$.
\end{proof}
\subsection{Monotone coupling of stochastic processes living on $\So$}\label{subsec:coupling}

In this section, we first introduce the definition and
classic results on stochastic dominance and coupling over a partial order space (see, e.g.,  \cite{strassen1965existence,kamae1977stochastic, lindvall2002lectures} for more details). Then, we present results for several stochastic processes living on $\So$ which are useful for comparing the proposed sequential decision with other decisions.

\begin{definition}\label{def:stoc-ordering}
Let $(\mathcal{S}, \leqc)$ be a partially ordered polish space. Assume $X, Y$ are two $\mathcal{S}$-valued random variables.  $X$ is  stochastically dominated by $Y$ (denoted by  $X \leqc_{st} Y$) if for all increasing, bounded and measurable functions $f: \mathcal{S} \rightarrow \RR$, $\EE (f(X)) \leq \EE (f(Y))$. 
\end{definition}
Let $\dd$ denote that random variables on both sides have the same distribution. The next result \cite{lindvall2002lectures} connects coupling with stochastic ordering.
\begin{fact}[Strassen’s theorem for a polish pospace] \label{fact:Strassen-thm}
Let $(\mathcal{S}, \leqc)$ be a polish partially ordered space, and let $X$ and $Y$ be $\mathcal{S}$-valued random variables. Then,  $X \leqc_{st} Y$ if and only if there exists a coupling $(\hat{X}, \hat{Y})$ such that $\hat{X} \dd X$ and $\hat{Y} \dd Y$ and $\hat{X} \leqc \hat{Y}$ a.s.
\end{fact}

Next we introduce the `monotonicity' of a Markov kernel over a pospace.
\begin{definition} \label{def:stochastically-monotone}
Let $K_1, K_2$ be two transition kernels over a partially ordered polish space $(\mathcal{S}, \leqc)$. $K_1$ is said to be {\em stochastically dominated} by $K_2$ (denoted by or $K_1 \prec_{st} K_2$) if  $$K_1 (\uu, \cdot) \leqc_{st} K_2 (\vv, \cdot)$$ for all $\uu \leqc \vv$. Moreover, if $K \prec_{st} K$, then the kernel $K$ is said to be stochastically monotone. 
\end{definition}

\begin{lemma} \label{lemma:Kt-MC}
The process $\{W_{k,t}\}_{t \geq 1}$ defined in \eqref{eq:rec2} and \eqref{eq:rec} is a Markov chain. Moreover, if  Assumption \ref{assump:iid} holds, $\{W_{k,t}\}_{t \geq 1}$ have the same  transition kernel for different $k\in\{1,\cdots,K\}$. Denote by  $\Kt (\cdot, \cdot)$ this transition kernel of $\{W_{k,t}\}_{t \geq 1}$ at time $k$. Then,  $\Kt$ are stochastically monotone for $t\geq 1$.
\end{lemma}
\begin{proof}
Combining equations \eqref{eq:rec2} and \eqref{eq:rec}, we have a recursive formula for $W_{k, t}$
\begin{equation} \label{eq: process W}
W_{k, t+1} = \frac{L_{k,t+1} }
{L_{k,t+1} + \frac{\bpi_{t+1}}{\bpi_{t}}   
\frac{1 - W_{k,t}} 
{\frac{\pi_{t}}{\bpi_{t}} + (1 - \frac{\pi_{t}}{\bpi_{t}})  W_{k,t} } }.
\end{equation}

Recall that $L_{k,t+1}= q_{k,{t+1}}\big(X_{k,{t+1}}\big) / p_{k,{t+1}}(X_{k,{t+1}})$. From the right-hand side of \eqref{eq: process W} we know $W_{k, t+1}$ is a function of $W_{k,t}$ and $X_{k,t+1}$%
. For $X_{k, t+1}$, conditioning on $W_{k,1}, W_{k,2}, \ldots, W_{k,t}$, its density function is 
$
g(W_{k,t}) q_{k,t+1}(\cdot) + (1 - g(W_{k,t})) p_{k,t+1}(\cdot),
$
which is determined by  $W_{k,t}$. Thus,  $W_{k, t+1}$ only depends on $W_{k,t}$, given $W_{k,1}, W_{k,2}, \ldots, W_{k,t}$, which further implies that $\{W_{k,t}\}_{t \geq 1}$ is a Markov chain.

Next, from the above argument we know the probability density function of $X_{k,t+1}$ conditioning on $W_{k,t}$ is 
$
g(W_{k,t}) q_{k,t+1}(\cdot) + (1 - g(W_{k,t})) p_{k,t+1}(\cdot),
$
which  is independent of the index $k$ given $W_{k,t}$ under Assumption~\ref{assump:iid}. This, combined with  \eqref{eq: process W}, implies that the conditional probability of $W_{k, t+1}$ given $W_{k,t}$ is independent of the index $k$. That is, all the streams share the same transition kernel $K_t(\cdot, \cdot)$.
In the rest of the proof, we show that the kernels $K_t$ are stochastically monotone, for which we extend the proof of Lemma F.9 in \cite{chen2019compound}. 

Under Assumption~\ref{assump:iid},  $p_{k,t} = p_t$ and $q_{k,t} = q_t$ for all $k$ for some functions $p_t$ and $q_t$. We first describe a random variable $M_{t+1}(x)$, which we will see to have the distribution $K_t(x,\cdot)$. Let $\zeta_t(x) = \bar{\pi}_{t}^{-1} \pi_{t} + \big(1- \bar{\pi}_{t}^{-1} \pi_{t} \big) x$. 
For $x\in(0,1)$, first generate a random variable $Z_{t+1}(x)$ with the density 
$
\zeta_t(x) q_{t+1}(\cdot) + (1 - \zeta_t(x)) p_{t+1}(\cdot).
$
and then compute $L_{t+1}(x) = q_{t+1}(Z_{t+1}(x))/p_{t+1}(Z_{t+1}(x))$ and let
$
M_{t+1}(x) = (L_{t+1}(x) + \frac{\bpi_{t+1}}{\bpi_{t}}   
\frac{1 - x} 
{\frac{\pi_{t}}{\bpi_{t}} + (1 - \frac{\pi_{t}}{\bpi_{t}})  x })^{-1}  L_{t+1}(x) .
$
From this generation process, we can see that $M_{t+1}(x)$ has the same distribution as $W_{k,t+1} | W_{k,t} = x$. That is, $M_{t+1}(x)$ has the density function $K_t(x,\cdot)$.

Next, to show the kernel $K_t(x,\cdot)$ is stochastically monotone, by Definition~\ref{def:stochastically-monotone}, it is sufficient to show $K_t(x,\cdot) \leq_{st} K_t(x^{\prime},\cdot)$ for any $x,x^{\prime}$ with $0 \leq x \leq x^{\prime} \leq 1$. 
Because $\zeta_t(x)$ is increasing in $x$ and according to Lemma F.6 in \cite{chen2019compound},  
we have
$
L_{t+1}(x) \leq_{st} L_{t+1}(x^{\prime}).$
Combine this result with Fact \ref{fact:Strassen-thm}, there exists a coupling $(\hat{L}_{t+1},\hat{L}_{t+1}^{\prime})$ such that $\hat{L}_{t+1} \dd L_{t+1}(x)$, $\hat{L}_{t+1}^{\prime} \dd L_{t+1}(x^{\prime})$ and $\hat{L}_{t+1} \leq \hat{L}_{t+1}^{\prime}$ a.s. Let 
$\hat{M}_{t+1} = (\hat{L}_{t+1} + \frac{\bpi_{t+1}}{\bpi_{t}}   
\frac{1 - x} 
{\frac{\pi_{t}}{\bpi_{t}} + (1 - \frac{\pi_{t}}{\bpi_{t}})  x })^{-1} \hat{L}_{t+1} $
and 
$\hat{M}_{t+1}^{\prime} = (\hat{L}_{t+1}^{\prime} + \frac{\bpi_{t+1}}{\bpi_{t}}   
\frac{1 - x^{\prime}} 
{\frac{\pi_{t}}{\bpi_{t}} + (1 - \frac{\pi_{t}}{\bpi_{t}})  x^{\prime} })^{-1}
 \hat{L}_{t+1}^{\prime}.
$
Then,
we have 
$\hat{M}_{t+1} 
= 
(\hat{L}_{t+1} + \frac{\bpi_{t+1}}{\bpi_{t}}   
\frac{1 - x} 
{\frac{\pi_{t}}{\bpi_{t}} + (1 - \frac{\pi_{t}}{\bpi_{t}})  x } )^{-1}\hat{L}_{t+1} \leq 
(\hat{L}_{t+1}^{\prime}  + \frac{\bpi_{t+1}}{\bpi_{t}}   
\frac{1 - x} 
{\frac{\pi_{t}}{\bpi_{t}} + (1 - \frac{\pi_{t}}{\bpi_{t}})  x } )^{-1}\hat{L}_{t+1}^{\prime}  \leq (\hat{L}_{t+1}^{\prime} + \frac{\bpi_{t+1}}{\bpi_{t}}   
\frac{1 - x^{\prime}} 
{\frac{\pi_{t}}{\bpi_{t}} + (1 - \frac{\pi_{t}}{\bpi_{t}})  x^{\prime} })^{-1} \hat{L}_{t+1}^{\prime} = \hat{M}_{t+1}^{\prime} \text{ a.s.}
$
By Fact \ref{fact:Strassen-thm}, this inequality implies $M_{t+1}(x) \leq_{st} M_{t+1}(x^{\prime})$, which further implies $K_t(x,\cdot) \leq_{st} K_t(x^{\prime},\cdot)$.
\end{proof}

For the ease of presentation, let $S_t^{\delta}$ and $H^{\delta}_t$ denote the active set $S_t$ and the historical information  following the decision $\delta$. Similarly, we let $W^{\delta}_{k,t}=\mathbb{P}(\tau_k<t|H^{\delta}_t)$ be the posterior probability for the change point to occur before time $t$ at the $k$-th stream, following the decision $\delta$. We also let $W_{S, t}^{\delta}=\big(W_{k, t}^{\delta}\big)_{k \in S}$  denote the vector of posterior probability associated with the subset $S \subset\{1, \cdots, K\}$ of data streams following the decision $\delta$.

\begin{lemma} \label{lemma:Kt}
Under Assumption~\ref{assump:iid}, for any sequential decision $\delta$,  $[W^{\delta}_{S^{\delta}_{t+1}, t+1}]$ is independent of $H_t^{\delta}$ given $[W^{\delta}_{S^{\delta}_{t+1}, t}]$. In addition, the conditional density of $[W^{\delta}_{S^{\delta}_{t+1}, t+1}]$ at $\vv$ given $[W^{\delta}_{S^{\delta}_{t+1}, t}] = \uu$ with $\dim (\uu) = m$ is
\begin{equation}
\begin{split}
&\KKt(\uu, \vv):=\\
&\begin{cases}
\sum_{\pi \in \Gamma_{m}} \prod_{l=1}^{m} \Kt (u_{l}, v_{\pi(l)} ) &\text{if } \dim(\vv) = m \geq 1 \\
1 &\text {if } \dim(\vv)= m = 0 \\
0 & \text {otherwise},
\end{cases}   
\end{split}
\end{equation}
where $\Gamma_{m}$ represents the set of all permutations over $[m]$. $\KKt$ is a transition kernel on $\So \times \So$.
\end{lemma}
\begin{proof}
The proof follows that of Lemma F.10 in \cite{chen2019compound} by replacing the time-homogeneous kernel $K$ in Lemma F.10 in \cite{chen2019compound} with $K_t$.  %
\end{proof}

\begin{lemma}\label{lemma:increasingK}
Under Assumption \ref{assump:iid},  $\KKt(\uu, \cdot) \leqc_{st} \KKt(\uup, \cdot)$ for $\uu, \uup \in \So$ with $\uu \leqc \uu^{\prime}$, and $t\geq 1$.
\end{lemma}
\begin{proof}
The proof follows that of Lemma F.12 in \cite{chen2019compound} by replacing $K$ in Lemma F.12 in \cite{chen2019compound} by $\Kt$ and replacing $\KK$ in Lemma F.12 in \cite{chen2019compound} with $\KKt$. 
\end{proof}
Next, we compare several decisions described below. Let 
$
\sD= (d_1^*, d_2^*,\cdots),
$
 be the proposed sequential decision, 
and 
$\delta = ( d_{1}, d_{2},\ldots)$
be an arbitrary decision procedure.
Let $\psi_t$ be an operator over the space of sequential decisions, and for each decision $\delta$, 
\begin{equation}\label{eq:circ-decision}
\psi_t\circ \delta = (d_{1}, d_{2},\ldots, d_{t-1}, d_{t}^{*}, \ldots).
\end{equation}
That is, $\psi_t$ maps $\delta$ to another sequential decision rule which makes the same decision as  $\delta$ at time $1,2,\cdots,t-1$  and switch to the proposed  $\sD$ at time $t$ and afterwards.

In what follows, we will compare $\psi_{t_0}\circ \delta$ and $\psi_{t_0+1}\circ \delta$ for a fixed $t_0$ and an arbitrary $\delta\in\mathcal{D}_{\alpha}$. For the ease of presentation, we write
\begin{equation} \label{eq:done}
\done = \psi_{t_0}\circ \delta = ( d_{1}, d_{2},\ldots, d_{t_0-1}, d_{t_0}^{*}, \ldots).
\end{equation}
and \begin{equation} \label{eq:dtwo}
\dtwo = \psi_{t_0+1}\circ \delta = ( d_{1}, d_{2},\ldots, d_{t_0-1},d_{t_{0}}, d^{*}_{t_{0}+1}, \ldots).
\end{equation}
Note that $\psi_{t_0+1}\circ \done  = \done$. Also, $\dtwo$ turn to the proposed method $\sD$ one time unit later than $\done$.

The next lemma provides the transition kernel for the posterior probability process following the decision $\delta_1$.
\begin{lemma} \label{lemma:K_tilde}
Under Assumptions~\ref{assump:non-empty}, \ref{assump:risk}, \ref{assump:utility} and \ref{assump:iid},
For any $t_{0} \geq 1$ and $s \geq 0$, 
$[W^{\done}_{S^{\done}_{t_{0}+s+1}, t_{0}+s+1}]$ is conditionally independent of $H^{\done}_{t_0+s}$ given $[W^{\done}_{S^{\done}_{t_{0}+s}, t_{0}+s}]$. In addition, the conditional density of $[W^{\done}_{S^{\done}_{t_{0}+s+1}, t_{0}+s+1}]$ at $\vv$ given $[W^{\done}_{S^{\done}_{t_{0}+s}, t_{0}+s}] = \uu$ is
$
\KKtts(\uu, \vv):= \KKts( \Hts (\uu), \vv) .
$
\end{lemma}
\begin{proof}
Under Assumptions~\ref{assump:non-empty}, \ref{assump:risk} and \ref{assump:utility}, Lemmas~\ref{lemma:def-I&J} and \ref{lemma:Kt} hold. 
Then, the proof follows similar arguments as that of Lemma F.13 in \cite{chen2019compound} by replacing $K$ in Lemma F.13 in \cite{chen2019compound} with $K_{t_{0}+s}$, $\KK$ in Lemma F.13 in \cite{chen2019compound} with $\KKts$, and Lemmas D.2 and D.10 in \cite{chen2019compound} with Lemmas~\ref{lemma:def-I&J} and \ref{lemma:Kt}, respectively. The rest of the proof is omitted to avoid repetitions.
\end{proof}

The next lemma compares the decisions $\delta_1$ and $\delta_2$ at time $t_0+1$ conditional on the history up to time $t_0$, where $\done$ and $\dtwo$ are given in \eqref{eq:done} and \eqref{eq:dtwo}, respectively. %

\begin{lemma} \label{lemma:one coupling}
Let $\done$ and $\dtwo$ be defined in \eqref{eq:done}  and \eqref{eq:dtwo}. Then, $H_{t_{0}}^{\done}=H_{t_{0}}^{\dtwo}$ a.s. Moreover, let $h_{t_{0}}$ be in the support of $H_{t_{0}}^{\done}$ and $ H_{t_{0}}^{\dtwo}$. Then, $\big[ W_{S_{t_{0}+1}^{\done}, t_{0}+1}^{\done}\big]$ is stochastically dominated by $\big[W_{S_{t_{0}+1}^{\dtwo}, t_{0}+1}^{\dtwo} \big]$, i.e., $\big[ W_{S_{t_{0}+1}^{\done}, t_{0}+1}^{\done}\big] \leqc_{st} \big[W_{S_{t_{0}+1}^{\dtwo}, t_{0}+1}^{\dtwo} \big]$, conditional on $H_{t_{0}}^{\done} = h_{t_{0}}$, under Assumptions~\ref{assump:non-empty}, \ref{assump:risk}-- \ref{assump:iid}.

\end{lemma}
\begin{proof}
From the definition of $\delta_l$ and $H_{t}^{\delta_l}$, we have the iterative formula
$S_t^{\delta_l} = d_{t-1}(H^{\delta_l}_{t-1})$ and $H_{t}^{\delta_l}=\{H_{t-1}^{\delta_l}, \{X_{k,t}\}_{k\in S^{\delta_l}_{t}}, \{S^{\delta_l}_{t}\}\}$ for $t=1,\cdots, t_0$ and $l=1,2$. Also, for $t=1$, $S_{t}^{\delta_l}=\{1,\cdots,K\}$. By induction, we have $H_t^{\done}= H_t^{\dtwo}$ for $t=1,\cdots, t_0$. In particular, $H_{t_0}^{\done}= H_{t_0}^{\dtwo}$. This proves the first part of the lemma.

We proceed to prove that $[ W_{S_{t_{0}+1}^{\done}, t_{0}+1}^{\done}]$ is stochastically dominated by $[W_{S_{t_{0}+1}^{\dtwo}, t_{0}+1}^{\dtwo} ]$ conditional on $H_{t_{0}}^{\done} = h_{t_{0}}$.
By Lemma~\ref{lemma:Kt}, we can see that 
$[W^{\done}_{S^{\done}_{t_{0}+1}, t_{0}+1}]$ is independent of the history $H_{t_{0}}^{\done}$ given  $ [W^{\done}_{S^{\done}_{t_{0}+1}, t_{0}}] = \uu$. Also, given the history $H_{t_{0}}^{\done} = h_{t_{0}}$, $[W^{\done}_{S^{\done}_{t_{0}+1}, t_{0}}]$ is a deterministic function of the history $h_{t_{0}}$ and the sequential decision $\done$. 
Let $\ww_{t_{0}+1}: = [W^{\done}_{S^{\done}_{t_{0}+1}, t_{0}}] \big|H_{t_{0}}^{\done} = h_{t_{0}}$. Then $[W^{\done}_{S^{\done}_{t_{0}+1}, t_{0}+1}] | [W^{\done}_{S^{\done}_{t_{0}+1}, t_{0}}] = \ww_{t_{0}+1}$ has the conditional probability density $\mathbb{K}_{t_{0}} (\ww_{t_{0}+1}, \cdot)$.
Similarly, assume that $\ww_{t_{0}+1}^{\prime}: = [W^{\dtwo}_{S^{\dtwo}_{t_{0}+1}, t_{0}}] |H_{t_{0}}^{\dtwo} = h_{t_{0}}$. Then,  $[W^{\dtwo}_{S^{\dtwo}_{t_{0}+1}, t_{0}+1}] \big| [W^{\dtwo}_{S^{\dtwo}_{t_{0}+1}, t_{0}}] = \ww_{t_{0}+1}^{\prime}$ has the conditional probability density $\mathbb{K}_{t_{0}} (\ww_{t_{0}+1}^{\prime}, \cdot)$. 

According to the above arguments it suffices to show $\mathbb{K}_{t_{0}} (\ww_{t_{0}+1}, \cdot) \leqc_{st} \mathbb{K}_{t_{0}} (\ww_{t_{0}+1}^{\prime}, \cdot)$ to prove the lemma. According to Assumption~\ref{assump:iid} and Lemma~\ref{lemma:increasingK}, it suffices to show that $\ww_{t_{0}+1} \leqc \ww_{t_{0}+1}^{\prime}$. Next we compare $\ww_{t_{0}+1}$ and $\ww_{t_{0}+1}^{\prime}$ under two cases: $\ww_{t_{0}+1}^{\prime} = \varnothing$ and $\ww_{t_{0}+1}^{\prime} \neq \varnothing$. If $\ww_{t_{0}+1}^{\prime} = \varnothing$, then $\ww_{t_{0}+1} \leqc \varnothing = \ww_{t_{0}+1}^{\prime}$ holds due to the definition of the partial order. If $\ww_{t_{0}+1}^{\prime} \neq \varnothing$, since given the history $H_{t_{0}}^{\done} = h_{t_{0}}$, $[W^{\done}_{S^{\done}_{t_{0}}, t_{0}}]$ is a deterministic function of the history $h_{t_{0}}$, let $\ww_{t_{0}}$ denote $[W^{\delta}_{S^{\delta}_{t_{0}}, t_{0}}]|H_{t_{0}}^{\delta} = h_{t_{0}}$. According to Lemma~\ref{lemma:def-I&J}, $\ww_{t_{0}+1} = J_{t_0} (\ww_{t_{0}})$. By Lemma~\ref{lemma:one-step optimal}, we have $\ww_{t_{0}+1} \leqc \ww_{t_{0}+1}^{\prime}$.

\end{proof}

The above Lemma~\ref{lemma:one coupling} shows that the decision $\done$ can select streams with `smaller' posterior probabilities and the order statistics of these posterior probabilities remains `stochastically smaller' one time unit further. 

Specifically, given the history at $t_0$, at time $t_{0}+1$, the ordered posterior probabilities of the remaining stream are ``stochastically smaller" following  $\done$ when compared with that of $\dtwo$. 
Next, we provide results on combining comparison results on consecutive time points.
We need the following result concerning the composition of stochastic monotone transition kernels, which is a corollary of Proposition 1 in \cite{kamae1977stochastic}.

\begin{fact}[Strassen’s theorem for Markov chains over a polish pospace]\label{fact:strassen-composite}
Assume $\{X_t\}_{t \geq 0}$ and $\{Y_t\}_{t \geq 0}$ are two Markov chains over a partially ordered pospace $(\calS, \leqc)$. Denote by $\{\KK_{X,t}\}_{t \geq 0}$ and $\{\KK_{Y,t}\}_{t \geq 0}$ their transition kernels respectively. 

Assume that $\KK_{X,t} \prec_{st} \KK_{Y,t} $ for all $t\geq 0$, where `$\prec_{st}$' is defined in Definition~\ref{def:stochastically-monotone}.  and let
\begin{align*}
\KK_{X,0:n} &:= \KK_{X,0} \circ \KK_{X,1} \circ \cdots \circ \KK_{X,n}, \\
\KK_{Y,0:n} &:= \KK_{Y,0} \circ \KK_{Y,1} \circ \cdots \circ \KK_{Y,n}, 
\end{align*}
where $\circ$ denotes the composition of Markov kernels (see, e.g.,  \cite{fristedt1997construction} for the definition of composition of Markov kernels).

Then, for all $n\geq 0$, $\KK_{X,0:n}  \prec_{st} \KK_{Y,0:n}. $

\end{fact}

Next, we will show for any sequential decision $\delta$ and $t_0 \geq 1$, the transition kernel of the conditional distribution of $\big[W^{\psi_{t_0}\circ \delta}_{S^{\psi_{t_0}\circ \delta}_{t_{0}+s}, t_{0}+s}\big]$ given $\big[W^{\psi_{t_0}\circ \delta}_{S^{\psi_{t_0}\circ \delta}_{t_{0}}, t_{0}}\big]$ is stochastically monotone for all $s\geq 0$.

\begin{lemma} \label{prop:coupling}
Under Assumptions~\ref{assump:non-empty}, \ref{assump:risk}, \ref{assump:utility} and \ref{assump:iid}, 
For any sequential decision $\delta$, $t_0 \geq 1$, and any bounded, decreasing and measurable function $f:\So\to \mathbb{R}$, the function
\begin{equation}
\begin{split} \mathbb{E}\Big[f(\big[W^{\psi_{t_0}\circ \delta}_{S^{\psi_{t_0}\circ \delta}_{t_{0}+s}, t_{0}+s}\big])|\big[W^{\psi_{t_0}\circ \delta}_{S^{\psi_{t_0}\circ \delta}_{t_{0}}, t_{0}}\big]=\ww\Big]
\end{split}
\end{equation}
is decreasing in $\ww$ and is a measurable function.

\end{lemma}
\begin{proof}
For any sequential decision $\delta$ and $t_0 \geq 1$, let $Y_{s} = \big[W^{\psi_{t_0}\circ \delta}_{S^{\psi_{t_0}\circ \delta}_{t_{0}+s}, t_{0}+s}\big]$ for $s \geq 0$. By Lemma~\ref{lemma:K_tilde}, we can see that $\{Y_s\}_{s \geq 0}$ is a Markov chain with the transition kernels $\{\KKtts\}_{s \geq 0}$. By Lemma~\ref{lemma:increasing rule}, $\Hts(\uu) \leqc \Hts(\vv)$ for any $\uu, \vv \in \So$ satisfying $\uu \leqc \vv $. This further implies $\KKts(\Hts(\uu), \cdot) \leqc_{st} \KKts(\Hts(\vv), \cdot)$ according to Lemma~\ref{lemma:increasingK}. By Lemma~\ref{lemma:K_tilde}, we arrive at $\KKtts (\uu, \cdot) \leqc_{st} \KKtts (\vv, \cdot)$ for any $\uu \leqc \vv$. That is, $\KKtts$ is stochastically monotone for all $s \geq 0$.

Note that for any $s\geq 0$, $Y_s$ conditional on $Y_0$ has the composite transition kernel $\KKK_{t_0} \circ \KKK_{t_{0}+1} \circ \cdots \circ \KKK_{t_{0}+s}$. Thus, according to Fact~\ref{fact:strassen-composite}, such composite transition kernel $\KKK_{t_0} \circ \KKK_{t_{0}+1} \circ \cdots \circ \KKK_{t_{0}+s}$ is also stochastically monotone. That is, for any $\ww, \ww^{\prime} \in \So$ satisfying $\ww \leqc \ww^{\prime}$, we have for all $s\geq 0$
$Y_s | Y_0 = \ww \leqc_{st} Y_s | Y_0 = \ww^{\prime}. 
$
By Definition~\ref{def:stoc-ordering}, we further have
$\EE(f(Y_s)|Y_0 = \ww) \geq \EE(f(Y_s)|Y_0 = \ww^{\prime}),
$ for any bounded, decreasing and measurable function $f$, which implies that $\EE(f(Y_s)|Y_0 = \ww)$ is decreasing in $\ww$.

\end{proof}

Let $h_{t}$ be in the support of $H_{t}^{\delta}$ and $h_t=\{x_{k,l}, s_l,  \text{ for } k \in s_{l}, 1 \leq l \leq t\}$.
\begin{lemma}
\label{thm:one-step optimal pro}
Assume Assumptions~\ref{assump:non-empty}, \ref{assump:risk} and \ref{assump:utility} hold.
For any $t_0 \geq 1$,
let $\delta$ be an arbitrary sequential decision in the class $\mathcal{D}_{\alpha}$.
Then,
\begin{equation} \label{eq:local-utility}
\EE(U_{t_{0}} (\delta) \mid H_{t_{0}}^{\delta}) = \tilde{u}_{t_0}([W^{\delta}_{S^{\delta}_{t_{0}+1},t_0}]),
\text{ a.s.}
\end{equation}
where $U_{t_{0}}(\delta)$ denotes the utility at time $t_{0}$ following the sequential decision $\delta$. Moreover, let $\done$ be defined in \eqref{eq:done}. Then, 
$\EE(U_{t_{0}} (\delta) \mid H_{t_{0}}^{\delta}) \leq 
\EE(U_{t_{0}} (\done) \mid H_{t_{0}}^{\delta})
\text{ a.s.}
$
\end{lemma}
\begin{proof}
According to the definition of $U_t$ in \eqref{eq:utility} %
and Assumption~\ref{assump:utility}, we can see \eqref{eq:local-utility}, which proves the first part of the lemma.

For the rest  of the lemma, it is sufficient to show $\tilde{u}_{t_0}([W^{\delta}_{S^{\delta}_{t_{0}+1},t_0}])\leq \tilde{u}_{t_0}([W^{\done}_{S^{\done}_{t_{0}+1},t_0}])$. According to Lemma \ref{lemma:decreasing},  $\ud_t (\cdot)$ is decreasing. Thus, we only need to show  $[W^{\done}_{S^{\done}_{t_{0}+1},t_0}]\leqc [W^{\delta}_{S^{\delta}_{t_{0}+1},t_0}]$, which is our focus for the rest of the proof.

Denote by $S^{\done}_{t_0+1}$ and $S^{\delta}_{t_0+1}$ the index set at time $t_0+1$ following $\done$ and $\delta$, respectively. According to the definition of $\done$, $S^{\done}_{t_0+1}$ is obtained by Algorithm \ref{alg:simp} with the input $W^{\done}_{S^{\done}_{t_0},t_0}$ and $S^{\done}_{t_0}$. This further implies that $[W^{\done}_{S^{\done}_{t_{0}+1},t_0}] = J_{t_0} ([W^{\done}_{S^{\done}_{t_0},t_0}])$ according to Lemma \ref{lemma:def-I&J}. Note that $[W^{\done}_{S^{\done}_{t_0},t_0}]= [W^{\delta}_{S^{\delta}_{t_0},t_0}]$, so $[W^{\done}_{S^{\done}_{t_{0}+1},t_0}] = J_{t_0} ([W^{\delta}_{S^{\delta}_{t_0},t_0}])$, and it is sufficient to show 
\begin{equation}\label{eq:toshow}
J_{t_0} ([W^{\delta}_{S^{\delta}_{t_0},t_0}])\leqc [W^{\delta}_{S^{\delta}_{t_{0}+1},t_0}].
\end{equation}

On the other hand, $\delta\in\mathcal{D}_{\alpha}$, so $\tilde{r}_{t_0}([W^{\delta}_{S^{\delta}_{t_{0}+1},t_0}])=R_{t_0}(\delta)\leq \alpha$. According to \ref{lemma:one-step optimal} (with $\uu$ replaced by $[W^{\delta}_{S^{\delta}_{t_0},t_0}]$ and $\uu'$ replaced by $W^{\delta}_{S^{\delta}_{t_{0}+1},t_0}$), there are two possible cases: 1) $J_{t_0} ([W^{\delta}_{S^{\delta}_{t_0},t_0}])=\varnothing$ and and $[W^{\delta}_{S^{\delta}_{t_{0}+1},t_0}]=\varnothing$; or 2) $\dim(J_{t_0} ([W^{\delta}_{S^{\delta}_{t_0},t_0}]))\geq 1$ and $J_{t_0} ([W^{\delta}_{S^{\delta}_{t_0},t_0}])\leqc [W^{\delta}_{S^{\delta}_{t_{0}+1},t_0}]$. We can see that in both cases \eqref{eq:toshow} holds, which completes the proof.

\end{proof}

\subsection{Proof of Theorem~\ref{thm:uniform optimality}}\label{subsec:proof-of-thm-3}

\begin{proof}[Proof of Theorem~\ref{thm:uniform optimality}]
Theorem~\ref{thm:uniform optimality} is implied by the following stronger result, which will be the focus of the proof: for any $t_0 \geq 1, s \geq 0$, and any $\delta \in\mathcal{D}_{\alpha}$, %
\begin{equation} \label{eq:induction}
\EE(U_{t_{0}+s} (\delta) \mid H_{t_{0}}^{\delta}) \leq 
\EE(U_{t_{0}+s} (\psi_{t_0}\circ \delta) \mid H_{t_{0}}^{\delta})
\text{ a.s.}
\end{equation}
where $U_{t_{0}+s} (\delta)$  denotes the utility at time $t_{0}+s$ following sequential decisions $\delta$ and $\psi_{t_0}\circ \delta$ is defined in \eqref{eq:circ-decision}. Notice that $H_{t_{0}}^{\psi_{t_0}\circ \delta} = H_{t_{0}}^{\delta}$ a.s. If the above equation \eqref{eq:induction} is proved, Theorem~\ref{thm:uniform optimality} follows by setting $t_0 = 1$ in \eqref{eq:induction} and taking expectation on both sides.

In the rest of the proof, we show that \eqref{eq:induction} holds by induction on $s$.

For the base case $s = 0$, the equation \eqref{eq:induction} holds for all $t_0\geq 1$ according to Lemma~\ref{thm:one-step optimal pro}.

Now we assume the induction assumption \eqref{eq:induction} holds for $s = s_0$ and all $t_0 \geq 1$ and all decisions $\delta\in\mathcal{D}_{\alpha}$. That is,
\begin{equation} \label{eq:induction1}
\EE(U_{t_{0}+s_{0}} (\delta) | H_{t_{0}}^{\delta}) \leq 
\EE(U_{t_{0}+s_{0}} (\psi_{t_0}\circ \delta) | H_{t_{0}}^{\delta})
\text{ a.s.}
\end{equation}
for all $\delta$, and $t_0 \geq 1$. In the rest of the proof, we show that \eqref{eq:induction} also  holds for $s = s_{0}+1$ and all $t_0 \geq 1$ to complete the induction.

First, by replacing $t_0$ with $t_0+1$ in \eqref{eq:induction1}, we obtain
   $\EE(U_{t_{0}+s_{0}+1} (\delta) \mid H_{t_{0}+1}^{\delta} )
   \leq 
\EE(U_{t_{0}+s_{0}+1} (\psi_{t_0+1}\circ \delta) \mid H_{t_{0}+1}^{\delta} )
\text{ a.s.} $. Taking conditional expectation $\EE(\cdot | H_{t_{0}}^{\delta})$ on both sides, we arrive at
\begin{equation} \label{eq:induction2}
\begin{split}
\EE(U_{t_{0}+s_{0}+1} (\delta) \mid H_{t_{0}}^{\delta})  \leq 
\EE(U_{t_{0}+s_{0}+1} (\psi_{t_0+1}\circ \delta) \mid H_{t_{0}}^{\delta})
\text{ a.s.}
\end{split}
\end{equation}

Next, we consider $ U_{t_0+s_0+1}(\done)$
where we recall $\done=\psi_{t_0}\circ\delta$.
According to the definition of $U_{t_0+1+s_0}(\cdot)$, we have
\begin{equation}\label{eq:similar-start}
    U_{t_0+s_0+1}(\done) = \tilde{u}_{t_0+s_0+1}(\big[W_{S^{\done}_{t_{0}+s_0+2},t_{0}+s_0+1}^{\done}\big]),
\end{equation}
where $\tilde{u}(\cdot)$ is defined in Assumption~\ref{assump:utility}. According to Lemma~\ref{lemma:def-I&J}, $$\big[W_{S_{t_{0}+s_0+2},t_{0}+s_0+1}^{\done}\big] = J_{t_0+s_0+1}\Big(\big[W_{S^{\done}_{t_{0}+s_0+1},t_{0}+s_0+1}^{\done}\big]\Big).$$
Combining the above two equations, we arrive at
\begin{equation}
     U_{t_0+s_0+1}(\done) = \tilde{u}_{t_0+s_0+1}\Big(J_{t_0+s_0+1}\Big(\big[W^{\done}_{S^{\done}_{t_{0}+s_0+1},t_{0}+s_0+1}\big]\Big)\Big).
\end{equation}
Note that $\psi_{t_0+1}\circ \done = \done$. We further write the above equation as
\begin{equation}
\begin{split}
     &U_{t_0+s_0+1}(\psi_{t_0+1}\circ \done)\\
     = & \tilde{u}_{t_0+s_0+1}\Big(J_{t_0+s_0+1}\Big(\big[W^{\psi_{t_0+1}\circ\done}_{S^{\psi_{t_0+1}\circ\done}_{t_{0}+s_0+1},t_{0}+s_0+1}\big]\Big)\Big)\\
     =: & \varphi(\big[W^{\psi_{t_0+1}\circ\done}_{S^{\psi_{t_0+1}\circ\done}_{t_{0}+s_0+1},t_{0}+s_0+1}\big]),
\end{split}
\end{equation}
where we define the function $\varphi=\tilde{u}_{t_0+s_0+1}\circ J_{t_0+s_0+1}$ as the composition of $\tilde{u}_{t_0+s_0+1}$ and $J_{t_0+s_0+1}$. According to  Lemma \ref{lemma:decreasing}, $\tilde{u}_{t_0+s_0+1}(\cdot)$ is a decreasing function. By Lemma \ref{lemma:increasing rule}, $J_{t_0+s_0+1}(\cdot)$ is an increasing mapping. Thus, the function $\varphi(\cdot)$ is a decreasing function over $\So$.  Applying Lemma~\ref{prop:coupling} (with $f$ replaced by $\varphi$, $\ww$ replaced by $\uu$, $s$ replaced by $s_0$, $t_0$ replaced by $t_0+1$, $\delta$ replaced by $\delta_1$) to $U_{t_0+s_0+1}(\psi_{t_0+1}\circ \done)= \varphi\big(\big[W^{\psi_{t_0+1}\circ\done}_{S^{\psi_{t_0+1}\circ\done}_{t_{0}+s_0+1},t_{0}+s_0+1}\big]\big)$, we can see that
\begin{equation}
    \begin{split}
       \phi(\uu) := \EE\Big[ U_{t_0+s_0+1}(\psi_{t_0+1}\circ \done) \mid \big[W_{S^{\psi_{t_0+1}\circ\done}_{t_0+1},t_0+1}^{\psi_{t_0+1}\circ\done}\big] =\uu \Big]
    \end{split}
\end{equation}
is decreasing and bounded function in $\uu$ for $\uu\in \So$. We point out that we also used the assumption that $\tilde{u}_{t_0+s_0+1}(\cdot)$ is bounded (see Assumptions~\ref{assump:non-empty} and \ref{assump:utility}) in order to apply Lemma \ref{prop:coupling}. 
We further simplify the above equation using the fact that $\psi_{t_0+1}\circ\done=\done$ and obtain that
\begin{equation}
    \begin{split}
       \phi(\uu) = \EE\Big[ U_{t_0+s_0+1}(\done) \mid \big[W_{S^{\done}_{t_0+1},t_0+1}^{\done}\big] =\uu \Big]
    \end{split}
\end{equation}
is decreasing in $\uu$ and it is a bounded function.

Now, we consider $\EE(U_{t_{0}+s_{0}+1} (\done) \mid H_{t_{0}}^{\delta})$. By law of total expectation, we have
\begin{equation}\label{eq:similar:end}
    \begin{split}
    &\EE(U_{t_{0}+s_{0}+1} (\done) \mid H_{t_{0}}^{\delta})\\
    =& \EE\big( \EE\big(U_{t_{0}+s_{0}+1} (\done) \mid \big[W_{S^{\done}_{t_{0}+1},t_{0}+1}^{\done}\big] \big) \mid H_{t_{0}}^{\delta}\Big)\\
    =& \EE\Big( \phi\big(\big[W_{S^{\done}_{t_{0}+1},t_{0}+1}^{\done}\big]\big) \mid H_{t_{0}}^{\delta} \Big).
    \end{split}
\end{equation}

Recall that $\dtwo=\psi_{t_0+1}\circ \delta$.
Similar to the derivations in \eqref{eq:similar-start}--\eqref{eq:similar:end} with $\done$ replaced by $\delta$, we also obtain 
\begin{equation}\label{eq:similar:end-2}
    \begin{split}
    \EE(U_{t_{0}+s_{0}+1} (\dtwo) \mid H_{t_{0}}^{\delta}) 
    = \EE\Big( \phi\big(\big[W_{S^{\dtwo}_{t_{0}+1},t_{0}+1}^{\dtwo}\big]\big) \mid H_{t_{0}}^{\delta} \Big),
    \end{split}
\end{equation}
where we used the fact that $\psi_{t_0+1}\circ\dtwo=\dtwo$ in the derivations.

According to Lemma \ref{lemma:one coupling}, $\big[ W_{S_{t_{0}+1}^{\done}, t_{0}+1}^{\done}\big]$ is stochastically dominated by $\big[W_{S_{t_{0}+1}^{\dtwo}, t_{0}+1}^{\dtwo} \big]$. Combining this result with that $\phi(\cdot)$ is a bounded, decreasing and measurable function, we arrive at
\begin{equation}\label{eq:total-dtwo}
\begin{split}
    \EE\Big( \phi\big(\big[W_{S^{\dtwo}_{t_{0}+1},t_{0}+1}^{\dtwo}\big]\big) \mid H_{t_{0}}^{\delta} \Big) 
    \leq 
    \EE\Big( \phi\big(\big[W_{S^{\done}_{t_{0}+1},t_{0}+1}^{\done}\big]\big) \mid H_{t_{0}}^{\delta} \Big).
\end{split}
\end{equation}
Combining the above inequality with \eqref{eq:similar:end} and \eqref{eq:similar:end-2}, and noting that $\done=\psi_{t_0}\circ\delta$ and $\dtwo=\psi_{t_0+1}\circ\delta$, we arrive at
\begin{equation} \label{eq:induction3}
\begin{split}
&\EE(U_{t_{0}+s_{0}+1} (\psi_{t_0+1}\circ \delta) \mid H_{t_{0}}^{\delta}) \\
\leq 
&\EE(U_{t_{0}+s_{0}+1} (\psi_{t_0}\circ \delta) \mid H_{t_{0}}^{\delta})
\text{ a.s.}
\end{split}
\end{equation}
Finally, by combining equations \eqref{eq:induction2} and \eqref{eq:induction3} we have 
\begin{equation} 
\begin{split}
 \EE(U_{t_{0}+s_{0}+1} (\delta) \mid H_{t_{0}}^{\delta})  \leq \EE(U_{t_{0}+s_{0}+1} (\psi_{t_0}\circ \delta) \mid H_{t_{0}}^{\delta})
\text{ a.s.}
\end{split}
\end{equation}
holds for all $t_0$ and all decision $\delta\in\mathcal{D}_{\alpha}$. In other words, we proves that \eqref{eq:induction} holds for $s = s_{0}+1$ and all $t_0 \geq 1$ and $\delta\in\mathcal{D}_{\alpha}$, which completes the induction.
\end{proof}

\subsection{Proof of Propositions~\ref{prop:poly1}--\ref{prop:coonection-to-poor}, Remark~\ref{corollary:poly2}, Corollary~\ref{corollary:eg}--\ref{corollary:EU}, Lemma~\ref{lemma:GLFWER} and Counterexample \ref{eg:counterexample}}
The proof of proposition~\ref{prop:poly1} is based on the following lemma.
\begin{lemma}\label{lemma:maximum}
Let $f: [0,1]^p \rightarrow \RR$ be a differentiable function satisfying that $\partial_j f(u_1,\cdots,u_p)$ is continuous and does not depend on $u_j$ for all $j\in\langle p \rangle$.
Then, 
\begin{equation}\label{eq:sup}
\sup\limits_{\uu \in [0,1]^p} f(\uu) =\max\limits_{\uu \in \{0,1\}^p} f(\uu),
\end{equation}
\end{lemma}
\begin{proof}
Since the function $f$ on $[0,1]^p$ is a continuous function over a compact set, the supremum on the left-hand side of \eqref{eq:sup} is attainable. 
Because 
$
\sup_{\uu \in [0,1]^p} f(\uu)  \geq \max_{\uu \in \{0,1\}^p} f(\uu)
$, 
to show \eqref{eq:sup}, it suffices to show that 
$
\sup_{\uu \in [0,1]^p} f(\uu)  \leq \max_{\uu \in \{0,1\}^p} f(\uu)
$
which is equivalent to 
$
f(\uu)  \leq \max_{\uu \in \{0,1\}^p} f(\uu), \text{ for any } \uu  \in [0,1]^p.
$

To verify this, we compare function values coordinate-wise iteratively. For any $\uu =(u_1^0, \ldots, u_p^0) \in [0,1]^p$, we first look at the first cooridate $u_1^0$. Under the assumption that $\frac{\partial}{\partial u_j} f (\uu)$ is continuous and does not depend on $u_j$,
it follows that 
\begin{equation}
\begin{split}
    f(u_1^0, u_2^0, \ldots, u_p^0) 
    &\leq 
    f(1, u_2^0, \ldots, u_p^0) \\
    &\leq \max\limits_{u_1 \in \{0,1\}} f(u_1,u_2^0, \ldots, u_p^0) ,
\end{split}   
\end{equation}
if $\frac{\partial}{\partial u_1} f (\uu) |_{u_2 = u_2^0, \ldots, u_p = u_p^0} \geq 0$. Similarly,
\begin{equation}
\begin{split}
    f(u_1^0, u_2^0, \ldots, u_p^0)  
    &\leq 
    f(0, u_2^0, \ldots, u_p^0) \\
    &\leq \max\limits_{u_1 \in \{0,1\}} f(u_1,u_2^0, \ldots, u_p^0) ,
\end{split}   
\end{equation}
if $\frac{\partial}{\partial u_1} f (\uu)
|_{u_2 = u_2^0, \ldots, u_p = u_p^0} \leq 0$. 
Without loss of generality, assume $f(1,u_2^0,\cdots,u_p^0)\geq f(0,u_2^0,\cdots,u_p^0)$. Then, we look at the second coordinate $u_2^0$ for $(1, u_2^0, \ldots, u_p^0)$. With similar arguments, we have
\begin{equation} 
\begin{split}
&\max_{u_1 \in \{0,1\}} f(u_1,u_2^0, \ldots, u_p^0)\\
=& f(1, u_2^0, u_3^0, \ldots, u_p^0) \\
\leq & \max\limits_{u_1,u_2 \in \{0,1\}} f(u_1,u_2, u_3^0, \ldots, u_p^0).   
\end{split}
\end{equation}

We continue similar reasoning for the third to $p$-th coordinates. At the end, we arrive at 
       $ f(u_1^0, u_2^0, \ldots, u_p^0) 
       \leq  \max_{u_1\in\{0,1\}}f(u_1,u_2^0,\cdots,u_p^0)
       \leq  \cdots
       \leq  \max_{\uu \in \{0,1\}^p} f(\uu),
$
which completes the proof.
\end{proof}

\begin{proof}[Proof of Proposition \ref{prop:poly1}]
We start with showing that \eqref{eq:poly-entrywise-increasing} implies $\tilde{r}$ is entrywise increasing. It suffices to show that for each $p\geq 1$, if $\uu=(u_1,\cdots,u_p)\in \So$, $\vv=(v_1,\cdots,v_p)\in \So$, and $u_i\leq v_i$ for all $i\in\langle p \rangle$, then $\tilde{r}(\uu)\leq \tilde{r}(\vv)$.
We define a composite function $\eta: \Su \rightarrow \RR$ as $\eta(\uu) := \rd([\uu])$. We can see that $\eta$ extends the domain of $\tilde{r}$ to the unordered space $\Su$.
Because $\eta(\uu)=\tilde{r}(\uu)$ for $\uu\in\So$, to show $\tilde{r}$ is entrywise increasing, it suffices to show that for each $p\geq 1$, if $\uu=(u_1,\cdots,u_p)\in[0,1]^p$, $\vv=(v_1,\cdots,v_p)\in[0,1]^p$, and $u_i\leq v_i$ for all $i\in\langle p \rangle$, then $\eta(\uu)\leq \eta(\vv)$.
Note that for $\uu\in[0,1]^p$, $\eta$ is twice differentiable. Thus, it suffices to show that for all $p\geq 1$ and all $i\in\langle p \rangle$,
$\partial_i \eta(\uu)\geq 0$ for all $\uu\in[0,1]^p$. %

Note that for each $p\geq 1$, $\uu=(u_1,\cdots, u_p)$,
$\eta(\uu)= \sum_{k=1}^p C_{p,k}\sum_{i_1<\cdots<i_k}\prod_{j=1}^k u_{i_j}$ with the partial derivatives and second derivatives
\begin{equation*} \label{eq:derivative-1}
\begin{split}
\frac{\partial}{\partial_i} \eta(\uu) =& C_{p,1} +C_{p,2} \sum\limits_{l \neq i} u_{l} +C_{p,3} \sum\limits_{\substack{i_1, i_2 \neq i \\ distinct}} u_{i_1} u_{i_2} + \ldots \\
&+ C_{p,p} \sum\limits_{\substack{i_1, i_2, \cdots i_{p-1} \neq i \\ distinct}} u_{i_1} u_{i_2} \cdots u_{i_{p-1}},
\end{split}
\end{equation*}
\begin{equation*} \label{eq:derivative-2}
\begin{split}
\frac{\partial}{\partial_i\partial_j} \eta(\uu) =& C'_{p,2} +C'_{p,3} \sum\limits_{l\notin \{i,j\}} u_l+C'_{p,4} \sum\limits_{\substack{i_1, i_2 \notin \{i,j\} \\ distinct}} u_{i_1} u_{i_2} + \ldots \\
&+ C'_{p,p} \sum\limits_{\substack{i_1, i_2, \cdots i_{p-1} \notin \{i,j\} \\ distinct}} u_{i_1} u_{i_2}\cdots u_{i_{p-2}},
\end{split}
\end{equation*}
for some constants $C_{p,1}, C_{p,2}, \ldots, C_{p,p}$ and $ C'_{p,2}, \ldots, C'_{p,p}$. We can see that for each $j\in\langle p \rangle$, $\partial_{ij} \eta(\uu)$ does not depend on $u_i$. Applying Lemma~\ref{lemma:maximum} to $-\partial_{i}\eta$ for $\uu\in[0,1]^p$, we have $ \inf_{\uu\in[0,1]^p}\partial_i\eta(\uu) = \min_{\uu\in\{0,1\}^p}\partial_i\eta(\uu)$.
Note that $\partial_i\eta(\uu)=\eta(u_1,\cdots, u_{i-1},1,u_{i+1},\cdots,u_p)-\eta(u_1,\cdots, u_{i-1},0,u_{i+1},\cdots,u_p)$. Thus, \begin{equation}
\begin{split}
     &\inf_{\uu\in[0,1]^p}\partial_i\eta(\uu)\\ 
     = & \min_{\uu\in\{0,1\}^p}\{\eta(u_1,\cdots, u_{i-1},1,u_{i+1},\cdots,u_p)\\
     &\qquad\qquad-\eta(u_1,\cdots, u_{i-1},0,u_{i+1},\cdots,u_p)\}.
\end{split}
\end{equation}
Since $\eta$ is symmetric in its arguments, the right-hand side of the above equation is the same as 
$
\min_{l\in \langle p \rangle}\{\eta(\uu^{l-1,p-l+1})-\eta(\uu^{l,p-l})\},
$
where we recall $\uu^{l,p-l}$ is the vector whose first $l$-th elements are $0$'s and the $l+1$-th to $p$-th elements are $1$'s. Thus,
\begin{equation}\label{eq:similar-arg}
\begin{split}
     \inf_{\uu\in[0,1]^p}\partial_i\eta(\uu) 
     =  \min_{l\in \langle p \rangle}\{\eta(\uu^{l-1,p-l+1})-\eta(\uu^{l,p-l})\}.
\end{split}
\end{equation}
According to \eqref{eq:poly-entrywise-increasing}, the right-hand side of the above equation is non-negative. Thus, $     \inf_{\uu\in[0,1]^p}\partial_i\eta(\uu)\geq 0$, which completes the proof of the first part of the proposition.

We proceed to the proof of the `Moreover' part. %
To show $\rd$ is appending increasing, it is sufficient to show that for all $p\geq 0$
$\uu = (u_1, \ldots, u_p) \in [0,1]^{p}$ and $u_p\in[0,1]$
$$
\eta(\uu) \leq \eta (u_1, \ldots, u_p, u_{p+1}). 
$$

With similar arguments as those for \eqref{eq:similar-arg}, for each $p\geq 0$,
\begin{equation}
    \inf_{\uu\in[0,1]^p}\{\eta(\uu,0)-\eta(\uu)\}=\min_{0\leq i\leq p}\{ \eta (\uu^{i+1,p-i})-\eta (\uu^{i,p-i})\}.
\end{equation}
According to \eqref{eq:poly-appending-increasing}, the right-hand side of the above equation is non-negative. Thus, $$ \inf_{\uu\in[0,1]^p}\{\eta(\uu,0)-\eta(\uu)\}\geq 0,$$ which implies
$ \eta(\uu)\leq \eta(\uu,0)\leq \eta(\uu,u_{p+1})
$ for all $u_{p+1}\in[0,1]$.
\end{proof}

\begin{proof}[Proof of Remark \ref{corollary:poly2}]
First, a direct calculation gives $\rd (\mathbf{u}^{i,p-i}) = \sum_{k=1}^{p-i}  C_{p,k} { p-i \choose k }.$
Plugging this equation in  \eqref{eq:poly-entrywise-increasing}, for $i=1,\cdots, p$, we can see that  \eqref{eq:poly-entrywise-increasing} is equivalent to
\begin{equation*}
0 \leq C_{p,1}  \leq \ldots \leq
\sum_{k=1}^{p-1}  C_{p,k} {p-1 \choose k}
\leq \sum_{k=1}^{p}  C_{p,k} {p \choose k}.
\end{equation*}
The above inequalities are equivalent to
$\sum_{k=1}^{p-i}  C_{p,k} \Big[ {p-i \choose k} - { p-i-1 \choose k}\Big] \geq 0,\text{ for } i =0, \ldots, p-1 $.
By the Pascal's triangle, we simplify the above inequalities as $\sum_{k=1}^{p-i}  C_{p,k}   { p-i-1 \choose k-1}  \geq 0,
\text{ for } i =0, \cdots, p-1 .$
Combining the above equations and inequalities for $i=0,..., p$ and obtain $\sum_{k=1}^{p-i}  C_{p,k} {p-i \choose k}
\leqslant \sum_{k=1}^{p-i}  C_{p+1,k} {p-i \choose k}$
which is further simplified as $\sum_{k=1}^{p-i}  \big(C_{p+1,k} - C_{p,k} \big) {p-i \choose k}
\geq 0.
$
\end{proof}

\begin{proof}[Proof of Corollary \ref{corollary:eg}]
First note that Assumption~\ref{assump:non-empty} is verified in the proof of Corollary~\ref{corollary:eg:one-step optimal-monotone risk}. Thus, it suffices to show that the risk process $R_t$ satisfies Assumption  and \ref{assump:risk} and the utility process $U_t$ satisfies  and \ref{assump:utility}. 

From their definitions in Example \ref{eg:LFWER}, \ref{eg:GLFWER}, and \ref{eg:LFNR} and \ref{eg:ARL} respectively, we can see they are all symmetric in their arguments and functions of their ordered statistics $\big[W_{S_{t+1},t}\big]$. We define a function $\xi: \So \rightarrow \RR$ such that $R_t$ and $U_t$ are both in the form of $\xi \big( \big[W_{S_{t+1},t}\big] \big)$ for any choice of $S_{t+1}$.

We start with verifying that $R_t\in\{\text{LFNR}_t,\text{GLFWER}_t, \text{LFWER}_t\}$ satisfies Assumption~\ref{assump:risk}.

If $R_t = \text{LFNR}_t$, by the definition in Example \ref{eg:LFNR}, we have $
\xi(\uu) = \frac{\sum_{i =1}^{\dim(\uu)} u_i }{\dim(\uu) \vee 1}.
$
for $\uu = (u_1, \ldots, u_{\dim (\uu)}) \in \So$. Clearly, $\xi$ is entrywise increasing%
. To see $\xi$ is also appending increasing, first we can see that $\xi (\varnothing) = 0 < u_1 = \xi(u_1)$ for any $u_1 \in [0,1]$. For  $\uu = (u_1, \ldots, u_p) \in \So$ with $\dim (\uu) = p \geq 1$ and $u_{p+1}\geq u_p$, $\xi(\uu)=\frac{\sum_{i =1}^{p} u_i }{p} 
\leq 
\frac{\sum_{i =1}^{p+1} u_i }{p+1}=\xi(\uu,u_{p+1}). $
Thus, $\xi$ is appending increasing.

For $R_t = \text{GLFWER}_t$, by the definition in Example \ref{eg:GLFWER}, \begin{equation} \label{eq:gfwer}
\xi_{m} (\uu) =  1-\sum_{j = 0}^{m-1} \sum\limits_{\substack{I \subset [\dim(\uu)] \\ |I| = j}} \Big(\prod_{i \in I} u_i \Big)
   \prod_{k\in [\dim(\uu)] \setminus I}(1-u_k)  
\end{equation}
for $\uu = (u_1, \ldots, u_{\dim (\uu)}) \in \So$. 
Clearly, $\xi_m(\uu)$ is a polynomial function of $\uu$ as defined in Proposition \ref{prop:poly1}. To show $\xi_m$ is entrywise increasing and appending increasing, we only need to verify \eqref{eq:poly-entrywise-increasing} and \eqref{eq:poly-appending-increasing} by Proposition \ref{prop:poly1}. We use the following arguments to avoid tedious calculations.

We revisit the definition of $\text{GLFWER}_{m,t}=\prob(E_{m,t}|\mathcal{F}_t)$. It is the conditional probability of the event that there are at least $m$ false non-detection errors given the information filtration $\fil_t$. 
Note that
$
\prob(E_{m,t}|\fil_t) = \prob(E_{m,t}| \big[W_{S_{t+1},t}\big])
$ and $\prob(E_{m,t}|\mathcal{F}_t) = \xi_m (\big[W_{S_{t+1},t}\big])$. Thus, $ \xi_m(\uu)
   =  \prob\big(E_{m,t}|\big[W_{S_{t+1},t}\big] = \uu \big).$
Then, $\xi_m(\mathbf{u}^{i,p-i})= \prob\big(E_{m,t}|W_{1,t}=\cdots = W_{i,t}=0,W_{i+1,t}=\cdots= W_{p,t}=1, S_{t+1}=\langle p \rangle\big)$
 due to the symmetry of GLFWER. Recall $W_{k,t}=\prob(\tau_k<t|\mathcal{F}_t)=\prob(\tau_k<t|X_{k,1},\cdots, X_{k,t})$. Note that $W_{k,t}=1$ is equivalent to $\tau_k<t$ a.s. and $W_{k,t}=0$ is equivalent to $\tau_k\geq t$ a.s. Thus, we further have
$\xi_m(\mathbf{u}^{i,p-i})
       = \prob\big(E_{m,t}|\tau_1 \geq t,\cdots,\tau_i \geq t, \tau_{i+1}< t,\cdots, \tau_{p}< t, S_{t+1}=\langle p \rangle\big)$.
Note that
if $\tau_1 \geq t,\cdots,\tau_i \geq t, \tau_{i+1} < t,\cdots, \tau_{p} < t$ and $S_{t+1}=\langle p \rangle$, then $| k:\tau_k<t \text{ and }k\in S_{t+1}|= p-i$. 
On the other hand, recall $E_{m,t}=\{| k:\tau_k<t \text{ and }k\in S_{t+1}|\geq m\}$. Thus,
\begin{equation}
    \xi_m(\mathbf{u}^{i,p-i}) = \ind(p-i\geq m) \text{ a.s.}
\end{equation}
Based on the above equation, we have
$\xi_m(\mathbf{u}^{i,p-i})=\ind(p-i\geq m)\leq \ind(p-i+1\geq m)=\xi_m(\mathbf{u}^{i-1,p-i+1})$ for all $i\in\langle p \rangle$. This verifies \eqref{eq:poly-entrywise-increasing}. Moreover, $\xi_m(\mathbf{u}^{i,p-i})=\ind(p-i\geq m)=\xi_m(\mathbf{u}^{i+1,p-i})$ for all $0\leq i\leq p$. This verifies \eqref{eq:poly-appending-increasing}.

Since $\text{LFWER}_t$ is a special case of $\text{GLFWER}_t$ where $m = 1$, we also have \eqref{eq:poly-entrywise-increasing} and \eqref{eq:poly-appending-increasing} verified for $R_t=\text{LFWER}_t$.

Now, we proceed to the utility process $U_t$. For $U_t = \text{IARL}_t$, by the definition in Example \ref{eg:ARL}, 
$
\xi(\uu) = \sum_{i =1}^{\dim(\uu)} 1 - g(u_i). 
$
for $\uu = (u_1, \ldots, u_{\dim (\uu)}) \in \So$, where the function $g$ is defined in \eqref{eq: tau leq t}. Because $g$ is an increasing function bounded between $0$ and $1$,  $\xi$ is entrywise decreasing and appending increasing.
\end{proof}
\begin{proof}[Proof of Lemma \ref{lemma:GLFWER}]\label{proof-lemma:GLFWER}
Let $f_m:[0,1]^p \rightarrow \RR$ be a function such that 
$$
f_m (\uu) =  1-\sum_{j = 0}^{m-1} \sum_{\substack{I \subset \langle p \rangle \\ |I| = j}} \Big(\prod_{i \in I} u_i \Big)
\prod_{k\in \langle p \rangle \setminus I}(1-u_k)
$$
for $\uu=(u_1,\cdots,u_p)\in [0,1]^p$.
Then,  $
f_m(\uu) = \xi_m ([\uu])
$
for $\uu \in [0,1]^p$ for $\xi_m$ defined in \eqref{eq:gfwer}, and $\xi_m$ is entrywise increasing from the proof of Corollary~\ref{corollary:eg}. Also, for any $\uu=(u_1,\cdots,u_p)\in[0,1]^p$ and $\vv=(v_1,\cdots,v_p)\in[0,1]^p$ satisfying $u_k\leq v_k$ for all $k\in\langle p \rangle$, we have $[\uu]\leqc[\vv]$. Thus,
$f_m(\uu)=\xi_m([\uu])\leq \xi_m([\vv])=f_m(\vv)$
if $u_k\leq v_k$ for all $k$.
\end{proof}

\begin{proof}[Proof of Counterexample \ref{eg:counterexample}]
Denote by $\mathcal{H}$ the collection of all possible values of $H_t$ for all $t$. Denote by $\mathcal{A}$ the collection of all the subsets of $\{1,2,\ldots, K\}$. Then, the problem is formulated through a Markov Decision Process (MDP) with the state  $\HH$ and the action space $\Acal$ and the optimal solution can be obtained by backward induction.

We define a value function $V_t^{(r)}: \HH \rightarrow \RR$ and an action-value function $Q_t^{(r)}: \HH \times \Acal \rightarrow \mathbb{R}$ as follows. 
Recall that $H_t=\{X_{k,l},  k \in S_{l}, S_l, 1 \leq l \leq t\}$. Denote by $h_t$, $x_{k,l}$, and $s_l$ the realization of $H_t$, $X_{k,l}$ and $S_l$, respectively.
For any $0\leq t \leq r$, the value function $h_t \mapsto V_t^{(r)}(h_t)$ is defined as 
\begin{equation} \label{eq:def-V}
    V_t^{(r)}(h_t) = \max_{\delta\in\mathcal{D}_{\alpha}} \mathbb{E}(U_r|H_t=h_t)
\end{equation}
for all $h_t \in \HH$. Of note, $\EE(U_r)=V_0^{(r)}(\varnothing)$ is the quantity we would like to maximize.

For any $0\leq t \leq r$, we define the action-value function $Q_t^{(r)}(\cdot, \cdot)$  as 
\begin{equation}\label{eq:def-Q}
    Q_t^{(r)}(h_t,s) = \max_{\delta\in\mathcal{D}_{\alpha}}\mathbb{E}(U_r|H_t=h_t, S_{t+1}=s )
\end{equation}
for all $(h_t,s) \in (\HH, \Acal)$. 
According to the definition of $U_t$ in \eqref{eq:utility}, for $t=r$
\begin{equation} \label{eq:Q-r}
\begin{split}
Q_r^{(r)}(h_r,s) &= \mathbb{E}(U_r|H_r=h_r, S_{r+1}=s ) \\
&= u_{r}(\{w_{k,r}\}_{k\in s_{r} }, s_{r}, s ),
\end{split}
\end{equation}
where $s_r$ and $\{w_{k,r}\}_{k\in s_{r} }$ are the realization of $S_r$ and $\{W_{k,r}\}_{k\in s_r}$ given that $H_r=h_r$. Note that $s_r$ and $\{w_{k,r}\}_{k\in s_{r} }$ are determined by $h_r$.

According to Proposition 3.1 in \cite{bertsekas1996stochastic}, the following optimality equations hold for for $0\leq t \leq r-1$
\begin{equation} \label{eq:relation-V-Q}
V_t^{(r)}(h_t) = \max_{\substack{s\subset s_t \\ r_{t}(w_{s_{t},t}, s_{t}, s) \leq \alpha}} Q_t^{(r)}(h_t,s),
\end{equation}
and 
\begin{equation}\label{eq:relation-Q-V}
Q_t^{(r)}(h_t,s) = \EE\Big( V^{(r)}_{t+1}(H_{t+1}) \big| H_t = h_t, S_{t+1} = s\Big).
\end{equation}
Combining \eqref{eq:relation-V-Q} and \eqref{eq:relation-Q-V} yields a backward induction equation for $V_t^{(r)}(h_t)$:
\begin{equation}
V_t^{(r)}(h_t) = \max_{\substack{s\subset s_t \\ r_{t}(w_{s_{t},t}, s_{t}, s) \leq \alpha}} \EE\Big( V^{(r)}_{t+1}(H_{t+1}) \big| H_t = h_t, S_{t+1} = s\Big).
\end{equation}

By solving the above optimality equations, we are able to enumerate all sequential decisions that maximize $\EE(U_1)$ and $\EE(U_2)$. We omit the detailed calculation for the ease of presentation. 
In particular, for $r=1$, any sequential decision that maximizes $\EE(U_1)$ selects $S_2$ as
\begin{equation}\label{eq:req1}
S_2=
    \begin{cases}
    \{1\} \text{ or } \{2\} & \text{ if } X_{S_1,1}= (0,0,1)\\
    \{1\} \text{ or } \{3\} & \text{ if } X_{S_1,1}= (0,1,0)\\
    \{2\} \text{ or } \{3\} & \text{ if } X_{S_1,1}= (1,0,0)\\
    \{1,2,3\} & \text{ if } X_{S_1,1}= (0,0,0)\\
       \varnothing & \text{ if } X_{S_1,1}\in \begin{aligned} \{(1,1,1),(0,1,1),\\(1,0,1), (1,1,0)\}
   \end{aligned}
    \end{cases}
\end{equation}

Of note, our proposed method is one such decision.
On the other hand, for $r=2$, any the sequential decision that maximizes $\EE(U_2)$ selects $S_2$ as
\begin{equation}\label{eq:req2}
     S_2 =
     \begin{cases}
        \{1,2,3\} & \text{ if } X_{S_1, 1} \in\{ (0,0,1),(0,1,0), (1,0,0)\}\\
           \{1,2,3\} & \text{ if } X_{S_1,1}= (0,0,0)\\
       \varnothing & \text{ if } X_{S_1,1}\in \begin{aligned} \{(1,1,1),(0,1,1),\\(1,0,1), (1,1,0)\}
   \end{aligned}
    \end{cases}
\end{equation}
Note that $S_1=\{1,2,3\}$, so the $S_2$ described above is well-defined.

Because it is not possible for a sequential decision to satisfy both equations \eqref{eq:req1} and \eqref{eq:req2}, we conclude that there is no sequential decision in $\mathcal{D}_{\alpha}$ that maximizes $\mathbb{E}(U_r)$ for both $r=1$ and $r=2$. Thus, the uniformly optimal sequential decision does not exist.
\end{proof}

\begin{proof}[Proof of Proposition~\ref{prop:AR}]
	According to the definition of $\mathcal{D}_{\alpha}$, $\delta \in \mathcal{D}_{\alpha}$ implies
	$
	R_{t}(\delta) \leq \alpha \text { a.s.} 
	$ for all $t\in\mathbb{Z}_+$. Thus, 
	$\text{AR} (\delta) = \EE\big(\sum_{t=1}^{\infty}a_t R_t\big) \leq \alpha \EE\big(\sum_{t=1}^{\infty}a_t\big)=\alpha$, where the last equation is due to  $\sum_{t = 1}^{\infty} a_t = 1$.
\end{proof}

\begin{proof}[Proof of Proposition~\ref{prop:AU}]
		$
	\text{AU} (\delta^{\prime}) 
	= \EE\big(\sum_{t=1}^{\infty}b_t U_t (\delta^{\prime})\big)
	= \sum_{t=1}^{\infty}b_t \EE\big( U_t (\delta^{\prime})\big)\leq \sum_{t=1}^{\infty}b_t \EE\big( U_t (\delta)\big)=\text{AU}(\delta)
	$ for  any $\delta^{\prime} \in \mathcal{D}_{\alpha}$, 
	where the second last inequality is due to the assumption that $\delta$ is uniformly optimal, and the last equation is obtained based on the definition of aggregated risk.
\end{proof}

\begin{proof}[Proof of Corollary~\ref{corollary:HT}]
	By comparing the definition of $\text{GFWER}_m$ with that of $\text{GLFWER}_{m,t}$ defined in \eqref{eq:GLFWER}, we have
	$
	\text{GFWER}_m = \EE (R_{T}) = \EE \big(\sum_{t=1}^{\infty}\ind(T=t)R_t\big).
	$ The proof is completed by applying Proposition \ref{prop:AR} with $a_t = \ind(T = t)$. 
\end{proof}

\begin{proof}[Proof of Proposition \ref{prop:coonection-to-poor}]
Note that
\begin{equation} \label{eq:poor-afdr-reorg}
\begin{split}
    &\sum_{t = 1}^{\bar{N}-1} \EE\Big(  \frac{ K^{-1} \sum_{k = 1}^K \ind \big( N_k = t \big) }{ K^{-1} \big[ \{\sum_{s = 1}^{\bar{N}-1} \sum_{k = 1}^K \ind ( N_k = s )\} \vee 1 \big] } 
\text{FDP}_t \Big)\\ 
= & \sum_{t = 1}^{\bar{N}-1} \EE \Big(\frac{G_{K}^t}{M_K \vee 1} \text{FDP}_t \Big),
\end{split}
\end{equation}
where we define $G_{K}^t = K^{-1} \sum_{k = 1}^K \ind \big( N_k = t \big)$,
$M_K = \sum_{t = 1}^{\bar{N}-1} G_{K}^t = K^{-1} \sum_{t = 1}^{\bar{N}-1} \sum_{k = 1}^K \ind \big( N_k = t \big)$, and $M = \sum_{t = 1}^{\bar{N}-1} C_t$. 
For each 
$t$, because $\frac{G_{K}^t}{M_K \vee 1} \text{FDP}_t \rightarrow \frac{C_{t}}{M \vee 1} A_t$ in probability as $K$ goes to infinity and $\frac{G_{K}^t}{M_K \vee 1} \text{FDP}_t \in [0,1]$, %
we have 
\begin{equation} \label{eq:BCT-poor-afdr-1}
\lim_{K \rightarrow \infty} \EE \big(\frac{G_{K}^t}{M_K \vee 1} \text{FDP}_t \big) = \EE\big(\frac{C_{t}}{M \vee 1} A_t\big)=\frac{C_{t}}{M \vee 1} \EE\big( A_t \big)
\end{equation}
according to the dominated convergence theorem.
For $\EE(A_t)$, by dominated convergence theorem again, we have
\begin{equation} \label{eq:BCT-poor-afdr-2}
\lim_{K\rightarrow\infty} \EE(\text{FDP}_t) = \EE(A_t).
\end{equation}
Combining the above two equations, we obtain
\begin{equation} \label{eq:BCT-poor-afdr-combined}
\lim_{K \rightarrow \infty} \EE \big(\frac{G_{K}^t}{M_K \vee 1} \text{FDP}_t \big) =\frac{C_{t}}{M \vee 1}\lim_{K\rightarrow\infty} \EE(\text{FDP}_t).
\end{equation}
On the other hand, by following a sequential decision in $\mathcal{D}_{\alpha}$, we have $\EE(\text{FDP}_t | \mathcal{F}_t) \leq \alpha$ a.s., which implies $\EE(\text{FDP}_t) \leq \alpha$. This, combined with \eqref{eq:BCT-poor-afdr-combined}, implies
\begin{equation} \label{eq:inequality-poor-afdr}
\EE\big(\frac{C_{t}}{M \vee 1} A_t\big) \leq \alpha \frac{C_{t}}{M \vee 1}.
\end{equation}
Combining \eqref{eq:poor-afdr-reorg}, \eqref{eq:BCT-poor-afdr-1} and \eqref{eq:inequality-poor-afdr}, we further have 
\begin{equation}\label{eq:fdp-lim}
\begin{split}
\lim_{K \rightarrow \infty} \sum_{t = 1}^{\bar{N}-1} \EE \big(\frac{G_{K}^t}{M_K \vee 1} \text{FDP}_t \big) 
\leq \sum_{t = 1}^{\bar{N}-1} \alpha \frac{C_{t}}{M \vee 1} .
\end{split}
\end{equation}
Since $\sum_{t = 1}^{\bar{N}-1} \frac{G^t_{K}}{M_{K} \vee 1} = \ind(M_K \neq 0) \leq 1$ for all $K$,  we have $\sum_{t = 1}^{\bar{N}-1} \frac{C_{t}}{M \vee 1} \leq 1$. Combine this with \eqref{eq:poor-afdr-reorg} and  \eqref{eq:fdp-lim}, we arrive at 
$\lim_{K\to\infty}\text{AFDR}(\delta)=\lim_{K \rightarrow \infty} \sum_{t = 1}^{\bar{N}-1} \EE \big(\frac{G_{K}^t}{M_K \vee 1} \text{FDP}_t \big)
\leq \alpha,$
which completes the proof.
\end{proof}

\begin{proof}[Proof of Corollary~\ref{corollary:EU}]
	According to the definition of $\text{IRL}_t$ in Example~\ref{eg:ARL}, $\text{TARL} = \sum_{s = 0}^{\infty} \text{IRL}_s $. This implies $\EE(\text{TARL}) 
	= \EE(\sum_{s = 0}^{\infty} \EE(\text{IRL}_s | \mathcal{F}_{s})) = \EE(\sum_{s = 0}^{\infty} \text{IARL}_s)$. 
	According to Corollary~\ref{corollary:eg},
	$\sD$ is uniformly optimal when $U_t=\text{IARL}_t$. We complete the proof by letting $b_t=1$ for all $t$, $U_t=\text{IARL}_t$ and $\text{AU}(\delta)=\text{TARL}(\delta)$ in Proposition~\ref{prop:AU}.
\end{proof}

\section{Proofs of equations \eqref{eq:rec2} and \eqref{eq:rec}} \label{sec:proof-eqs}

\begin{proof}[Proof of Equation \eqref{eq:rec2}]
If $t=0$, then $Q_{k, 1} = \bar{\pi}_{1}^{-1}  \pi_{0}  L_{k,1}$.
Next, we consider the case where $t\geq1$.
Let $L_{k,(s+1): t}:=\prod_{r=s+1}^{t} \frac{q_{k,r}\big(X_{k,r} \big)}{p_{k,r}\big(X_{k,r}\big)}$.
Then,
\begin{align*} 
	Q_{k, t+1} 
	&= \frac{q_{k,t+1}\big(X_{k, t+1}\big)}{p_{k,t+1}\big(X_{k, t+1}\big)}\big[\sum_{s=0}^{t-1} \frac{\pi_{s}}{\bpi_{t+1}} L_{k,s+1: t}+\frac{\pi_{t}}{\bpi_{t+1} }\big] \\ 
	&= L_{k,t+1}\big[\sum_{s=0}^{t-1} 
	\frac{\bpi_{t+}}{\bpi_{t+1}} \frac{\pi_{s}}{\bpi_{t+}} L_{k,(s+1): t}+\frac{\pi_{t}}{\bpi_{t+1}}\big] 
\end{align*}
We complete the proof by combining the above equation with the the definition of $Q_{k, t}$. 
\end{proof}

\begin{proof}[Proof of Equation \eqref{eq:rec}]
By definition and Bayes formula, $$       W_{k,t}
        =\{\sum_{s=0}^{t-1} \pi_s \prod_{r=1}^{s} p_{k,r} (X_{k, r}) \prod_{r=s+1}^{t} q_{k,r}(X_{k, r})\}\cdot\{\sum_{s=0}^{t-1}\pi_s\prod_{r=1}^{s} p_{k,r} (X_{k, r}) \prod_{r=s+1}^{t} q_{k,r} (X_{k, r})+ \bar{\pi}_t \prod_{r=1}^{t} p_{k,r} (X_{k, r})\}^{-1} .$$
This is further simplified as
$W_{k,t} 
=\frac{\sum_{s=0}^{t-1} \pi_{s} L_{k,(s+1): t}}
{\sum_{s=0}^{t-1} \pi_{s} L_{k,(s+1): t}+\bpi_{t} } 
=\frac{\sum_{s=0}^{t-1} \frac{\pi_{s}}{\bpi_{t}} L_{k,(s+1): t}}
{\sum_{s=0}^{t-1} \frac{\pi_{s}}{\bpi_{t}} L_{k,(s+1): t}+ 1} 
=\frac{Q_{k, t}}{Q_{k, t}+ 1}.$

\end{proof}

\begin{proof}[Proof of Equation \eqref{eq: tau leq t}]
\begin{equation*}
    \begin{split}
        g(W_{k,t})
        =&\big\{\sum_{s=0}^{t-1} \prob\big(\tau_{k}=s\big) \prod_{r=1}^{s} p_{k,r} \big(X_{k, r}\big) \prod_{r=s+1}^{t} q_{k,r} \big(X_{k, r}\big)\\
        & + \prob \big(\tau_{k} = t\big) \prod_{r=1}^{t} p_{k,r} \big(X_{k, r}\big)
        \big\}\\
        & \cdot \big\{
        \sum_{s=0}^{t-1} \prob \big(\tau_{k}=s\big) \prod_{r=1}^{s} p_{k,r} \big(X_{k, r}\big) \prod_{r=s+1}^{t} q_{k,r} \big(X_{k, r}\big) \\
        & + \prob \big(\tau_{k} \geq t\big) \prod_{r=1}^{t} p_{k,r} \big(X_{k, r}\big)
         \big\}^{-1}
    \end{split}
\end{equation*}
By the definition of $L_{1,(s+1): t}$, we simplify the above result as 
$
g(W_{k,t})
=\frac{\sum_{s=0}^{t-1} \pi_{s} L_{k,(s+1): t} +\pi_{t} }
{\sum_{s=0}^{t-1} \pi_{s} L_{k,(s+1): t}+\bpi_{t} }.
$
Then in light of the relationship between $Q_{k, t}$ and $W_{k, t}$, we further simplify it to 
$
 g(W_{k,t})
=\frac{Q_{k, t}  + \frac{\pi_{t}}{\bpi_{t}}  }{Q_{k, t}+ 1} 
= \frac{\frac{W_{k, t}}{1-W_{k, t}}+\frac{\pi_{t}}{\bpi_{t}}}{\frac{W_{k, t}}{1-W_{k, t}}+1} 
=  W_{k, t}+\frac{\pi_{t}}{\bpi_{t}}\big(1-W_{k, t}\big) 
= \frac{\pi_{t}}{\bpi_{t}} + \big( 1- \frac{\pi_{t}}{\bpi_{t}} \big) W_{k, t} .$
\end{proof}
\bibliographystyle{plainnat}
\bibliography{Ref}

\end{document}